\documentclass[acmsmall,nonacm,11pt]{acmart}
\usepackage[utf8]{inputenc}
\usepackage[T1]{fontenc}
\usepackage[utf8]{inputenc}

\usepackage{mdframed}

\usepackage{mathtools,thmtools,thm-restate,enumitem}

\usepackage{wrapfig}
\usepackage{mathrsfs}

\graphicspath{{./}{../}{figures/}{figs/}{pictures/}{pics/}{logos/}}

\usepackage{tikz}
\usetikzlibrary{positioning,shapes.geometric, arrows}
\tikzset{node distance = 1.5cm, every text node part/.style={align=center}}
\tikzstyle{box} = [rectangle, rounded corners, minimum width=3cm, minimum height=1cm,text centered, draw=black, fill=red!30]
\tikzstyle{arrow} = [thick,->,>=stealth]

\usepackage[procnumbered,ruled,vlined,linesnumbered]{algorithm2e}
\SetKw{Input}{Input:}
\SetKw{Output}{Output:}
\SetKwProg{InEach}{inside each}{ do}{end}
\DontPrintSemicolon
\let\oldnl\nl \newcommand{\nonl}{\renewcommand{\nl}{\let\nl\oldnl}}

\usepackage[showdow=true,showseconds=false,showisoZ=false]{datetime2} 
\usepackage[title]{appendix}

\definecolor{ForestGreen}{rgb}{0.1333,0.5451,0.1333}
\definecolor{DarkRed}{rgb}{0.65,0,0}
\definecolor{Red}{rgb}{1,0,0}
\definecolor{DarkRed}{rgb}{0.5,0.1,0.1}
\definecolor{DarkBlue}{rgb}{0.1,0.1,0.5}
\definecolor{internallinkcolor}{rgb}{0,.5,0}
\definecolor{mygray}{gray}{0.95}

\declaretheorem[numberwithin=section]{lemma}

\declaretheorem[numberlike=lemma]{theorem}

\declaretheorem[numberlike=lemma]{fact}

\declaretheorem[numberlike=lemma,name=Proposition]{prop}
\declaretheorem[numberlike=lemma]{corollary}
\declaretheorem[numberlike=lemma]{claim}

\declaretheorem[numberlike=lemma]{remark}

\theoremstyle{definition}
\declaretheorem[numberlike=lemma]{definition}

\newenvironment{informalbox}{\begin{mdframed}[backgroundcolor=lightgray!40,topline=false,rightline=false,leftline=false,bottomline=false,innertopmargin=2pt]}{\smallskip\end{mdframed}}

\usepackage[noabbrev,capitalize,nameinlink]{cleveref}
\crefname{theorem}{Theorem}{Theorems}
\crefname{section}{Section}{Sections}
\crefname{lemma}{Lemma}{Lemmas}
\crefname{simplelemma}{`Lemma'}{`Lemmas'}
\crefname{algorithm}{Algorithm}{Algorithms}
\crefname{line}{Step}{Steps}
\crefname{fact}{Fact}{Facts}
\crefname{prop}{Proposition}{Propositions}
\crefname{claim}{Claim}{Claims}
\crefname{part}{Property}{Properties}
\crefname{cond}{Condition}{Conditions}
\crefname{problem}{Problem}{Problems}

\renewcommand{\qedsymbol}{\nobreak \ifvmode \relax \else
      \ifdim\lastskip<1.5em \hskip-\lastskip
      \hskip1.5em plus0em minus0.5em \fi \nobreak
      \vrule height0.75em width0.5em depth0.25em\fi}

\newcommand{\emphdef}[1]{\emph{\textbf{#1}}}

\renewcommand{\epsilon}{\ensuremath{\varepsilon}}
\newcommand{\eps}{\ensuremath{\varepsilon}}
\newcommand{\eqdef}{\stackrel{\text{\tiny\rm def}}{=}}

\renewcommand{\leq}{\leqslant}
\renewcommand{\geq}{\geqslant}
\renewcommand{\le}{\leqslant}
\renewcommand{\ge}{\geqslant}

\let\abs\relax
\DeclarePairedDelimiter{\abs}{\lvert}{\rvert}\DeclarePairedDelimiter{\card}{\lvert}{\rvert}\let\set\relax
\DeclarePairedDelimiter{\set}{\lbrace}{\rbrace}\DeclarePairedDelimiter{\range}{\lbrack}{\rbrack}\DeclarePairedDelimiter{\parens}{\lparen}{\rparen}\DeclarePairedDelimiter{\floor}{\lfloor}{\rfloor}\DeclarePairedDelimiter{\ceil}{\lceil}{\rceil}\DeclarePairedDelimiter{\I}{\llbracket}{\rrbracket}

\DeclareMathOperator{\sign}{sign}

\let\Pr\relax
\DeclareMathOperator*{\Pr}{\ensuremath{\mathsf{Pr}}}
\DeclareMathOperator*{\Exp}{\ensuremath{{\mathbb{E}}}}

\DeclareMathOperator*{\poly}{\operatorname{poly}}

\newcommand{\calC}{\mathcal C}

\newcommand{\calE}{\mathcal E}
\newcommand{\calF}{\mathcal F}
\newcommand{\calG}{\mathcal G}
\newcommand{\calH}{\mathcal H}

\newcommand{\calL}{\mathcal L}

\newcommand{\calP}{\mathcal P}

\newcommand{\calU}{\mathcal U}

\newcommand{\calX}{\mathcal X}
\newcommand{\calY}{\mathcal Y}

\newcommand{\N}{\mathbb N}

\newcommand{\model}[1]{\ensuremath{\mathsf{#1}}\xspace}
\newcommand{\local}{\model{LOCAL}}
\newcommand{\LOCAL}{\local}
\newcommand{\congest}{\model{CONGEST}}
\newcommand{\CONGEST}{\congest}

\newcommand{\alg}[1]{\ensuremath{\mathtt{#1}}\xspace}
\newcommand{\trycolor}{\alg{TryColor}}
\newcommand{\trymulticolor}{\alg{TryPseudorandomColors}}

\newcommand{\multitrial}{\alg{MultiColorTrial}}

\newcommand{\ColoringNonCabals}{\alg{ColoringNonCabals}}
\newcommand{\ColoringCabals}{\alg{ColoringCabals}}

\newcommand{\slackgeneration}{\alg{SlackGeneration}}
\newcommand{\sct}{\alg{SynchronizedColorTrial}}
\newcommand{\computeACD}{\alg{ComputeACD}}
\newcommand{\computePutAside}{\alg{ComputePutAside}}
\newcommand{\colorPutAside}{\alg{ColorPutAsideSets}}
\newcommand{\colorfulmatching}{\alg{ColorfulMatching}}
\newcommand{\colorfulmatchingcabal}{\alg{ColorfulMatchingCabal}}

\newcommand{\FindCandidateDonors}{\alg{FindCandidateDonors}}
\newcommand{\FindSafeDonors}{\alg{FindSafeDonors}}
\newcommand{\DonateColors}{\alg{DonateColors}}

\newcommand{\col}{\varphi}
\newcommand{\colscript}[2][]{\col_{\mathrm{#2}#1}}  \newcommand{\coltotal}{\colscript{total}}
\newcommand{\colsg}{\colscript{sg}}
\newcommand{\colcm}[1][]{\colscript[#1]{cm}}
\newcommand{\colsct}{\colscript{sct}}

\newcommand{\ID}{\ensuremath{\mathsf{ID}}}

\newcommand{\congestion}{\ensuremath{\mathsf{c}}}
\newcommand{\dilation}{\ensuremath{\mathsf{d}}}

\newcommand{\Kcabal}{\mathcal{K}_{\mathsf{cabal}}}

\newcommand{\Vsparse}{V_{\mathsf{sparse}}}

\newcommand{\Vdense}{V_{\mathsf{dense}}}
\newcommand{\Vcabal}{V_{\mathsf{cabal}}}
\newcommand{\Vactive}{V^{\mathsf{active}}}

\newcommand{\Deltalow}{\Delta_{\mathsf{low}}}
\newcommand{\lmin}{\ell}
\newcommand{\ls}{\ell_{\mathsf{s}}}

\newcommand{\CSlack}{\gamma_{\ref{prop:slack-generation}}}
\newcommand{\CCSlack}{\gamma_{\ref{lem:reuse-slack}}}

\newcommand{\pg}{p_{\mathsf{g}}}

\newcommand{\crecol}{c^{\mathsf{recol}}}
\newcommand{\cdon}{c^{\mathsf{don}}}
\newcommand{\Qpre}{Q^{\mathsf{pre}}}
\newcommand{\Qactive}{Q^{\mathsf{active}}}

\newcommand{\Pcand}{P^{\mathsf{candidate}}}
\newcommand{\Psafe}{P^{\mathsf{safe}}}

\newcommand{\avail}{\mathsf{avail}}

     \DeclarePairedDelimiterX{\infdivx}[2]{(}{)}{#1\;\delimsize\|\;#2}

\DeclareMathOperator{\dom}{\mathsf{dom}}

\newcommand{\supth}{\ensuremath{^\text{th}}}

\title{Decentralized Distributed Graph Coloring: Cluster Graphs}
\date{}

\author{Maxime Flin}
\affiliation{
    \institution{Reykjavik University}
\country{Iceland}
}
\email{maximef@ru.is}

\author{Magn\'us M. Halld\'orsson}
\affiliation{
    \institution{Reykjavik University}
\country{Iceland}
}
\email{mmh@ru.is}

\author{Alexandre Nolin}
\affiliation{
    \institution{CISPA Helmholtz Center for Information Security}
    \country{Germany}
}
\email{alexandre.nolin@cispa.de}
\thanks{M. Flin was supported by the Icelandic Research Fund (grant 2310015). M.M. Halld\'orsson was partially supported by the Icelandic Research Fund (grant 217965).}

\begin{document}

\pagenumbering{roman}
\thispagestyle{empty}

\begin{abstract}
    Graph coloring is fundamental to distributed computing.
We give the first sub-logarithmic distributed algorithm for coloring cluster graphs. These graphs are obtained from the underlying communication network by contracting nodes and edges, and they appear frequently as components in the study of distributed algorithms.
In particular, we give a $O(\log^* n)$-round algorithm to $(\Delta+1)$-color cluster graphs of at least polylogarithmic degree.
The previous best bound known was $\poly(\log n)$ [Flin et al., SODA'24]. 
This properly generalizes results in the \congest model and shows that distributed graph problems can be solved quickly even when the node itself is decentralized. 
 \end{abstract}

\maketitle

\vfill

\setcounter{tocdepth}{1}
\tableofcontents

\vfill

\newpage

\pagenumbering{arabic}

\section{Introduction}
\label{sec:intro}
Graph coloring is a problem of fundamental importance in computer science.
Given a graph $G=(V_G, E_G)$ we must decide on a color $\col(v)$ for all vertices such that endpoints of all edges $\set{u,v} \in E_G$ are colored differently $\col(u) \neq \col(v)$.
Our focus is on the $\Delta+1$-coloring problem, where $\Delta$ is the maximum degree of $G$ and each color is in $[\Delta+1] = \{1,2,\ldots, \Delta+1\}$.
Despite being trivial in the classical setting, the seemingly sequential nature of this problem raises fundamental challenges in various constrained settings such as semi-streaming \cite{ACK19,ACGS_pods23}, dynamic algorithms \cite{BCHN18,HP22,BGKLS22}, and massively parallel computation \cite{CFGUZ19,CDP21}.
Particularly in the distributed literature, $\Delta+1$-coloring has been extensively studied since the \local model was first introduced by Linial \cite{linial92,PanconesiS97,johansson99,SW10,BEPSv3,Bar2016,fraigniaud16,HSS18,BEG18,CLP20,MT20,GK21,HKMT21,HKNT22,FK23}.

In the \local model, the graph to be colored is seen as a communication network where vertices communicate with their neighbors in synchronous rounds without bandwidth constraints.
In the randomized version of the model, nodes have access to local random bits to break symmetry.
Their goal is to communicate for as few rounds as possible until they each choose their own color while ensuring that, with high probability\footnote{with high probability (w.h.p.) means at least $1 - 1/n^c$ for any constant $c \geq 1$, where $n$ is the number of vertices}, all vertices are properly colored.
It is generally assumed that the graph to be colored is the same as the communication graph because, in \local, we can make this assumption (almost) for free. Indeed, when the graph to be colored $H$ differs from the communication network $G$, a \local algorithm on $H$ can be naively simulated on $G$ with overhead proportional to its \emph{dilation} $\dilation$.
Informally speaking, the dilation is the maximum distance in $G$ between  pairs of vertices aware of conflicts in $H$ and is an unavoidable cost, although constant in many settings. Situations like these arise frequently when \local algorithms are used as sub-routines in others (e.g., \cite{FG17,MU21,FGGKR23,JM23}).

For bandwidth-constrained models, where the aim is to use small $O(\log |V_G|)$-bit messages, the situation is quite different since congestion precludes naive simulation when degrees in $H$ are high.  This is a frequent issue in composing bandwidth-efficient algorithms \cite{G15,MU21,MPU23,GGR20,FGGKR23}.
The utility of known bandwidth-efficient $\Delta+1$-coloring algorithms \cite{HKMT21,HNT22}
is limited precisely because they assume the input and communication graphs coincide.

The overriding question that arises is whether, in presence of bandwidth restrictions, it makes a fundamental difference if the input $H$ and communication graphs $G$ differ. A very recent sparsification result of \cite{FGHKN24} (informally) shows that $H$ can be colored in $O(\log^2 n)$ rounds \emph{when vertices can only communicate $O(\log n)$-bit messages with $O(\log^4 n)$ of their neighbors per round}. While this is the first non-trivial result for such a restricted setting, because of its very weak assumptions on the bandwidth, its complexity only (nearly-)matches the one of classic logarithmic \local/\congest algorithms \cite{johansson99,luby86}. In fact, they show that their technique cannot be significantly improved, not below $\Omega(\frac{\log n}{\log\log n})$. Meanwhile, in the \local/\congest models, a long line of research \cite{SW10,BEPSv3,HSS18,CLP18,HKMT21,HNT21} has improved the complexity of $\Delta+1$-coloring to $\poly(\log\log n)$ rounds and even $O(\log^*n)$ rounds when $\Delta \geq \poly(\log n)$. In this paper, we ask

\vspace{.5em}
\parbox{\linewidth}{
\begin{quote}
\emph{Can we $\Delta+1$-color $H$ as fast as in \local even when communication happens on a bandwidth-constrained nework $G \neq H$?}
\end{quote}
}
\vspace{.5em}

We answer this question positively when $H$ is a \emph{cluster graph}: a formalism that has been used under various names (and with some variations) to capture such situations \cite{ghaffari2013cut,ghaffari2016distributed,GKKLP18,RG20,GGR20,FGLPSY21,GZ22,RozhonGHZL22,GHIR23}.

\subsection{Cluster Graphs}
\label{sec:intro-cluster}
Throughout this paper, $G=(V_G, E_G)$ refers to the communication \emph{network} with $n = |V_G|$ vertices called \emph{machines} and edges called \emph{links} with $O(\log n)$ bandwidth.
Our algorithm colors a \emphdef{cluster graph} defined over this communication network.

\vspace{.5em}
\begin{informalbox} 
A \emphdef{cluster graph} is a graph $H=(V_H,E_H)$ defined over a communication network $G=(V_G,E_G)$ by partitioning the machines of $G$ into disjoint connected \emphdef{clusters} of machines. Each node $v$ of $H$ corresponds to a cluster $V(v)$ in $G$, and two nodes in $H$ are connected if and only if their respective clusters are adjacent in $G$. (See \cref{sec:prelim} for the formal definition.)
\end{informalbox}
\vspace{.5em}

Cluster graphs appear in several places, from maximum flow algorithms \cite{GKKLP18,FGLPSY21} to network decomposition \cite{RG20, GGR20}. See \cref{fig:cluster-graphs} for an example.
They arise naturally when algorithms contract edges, for instance in \cite{GKKLP18,FGLPSY21}. Vertices of the input graph then become connected \emph{sets of vertices} in the communication graph. They can then be seen as low-diameter trees on the communication graphs and each step of the algorithm must be simple enough to be implemented through aggregation.

\begin{figure}[H]
    \centering
    \includegraphics[page=5,width=0.3\textwidth]{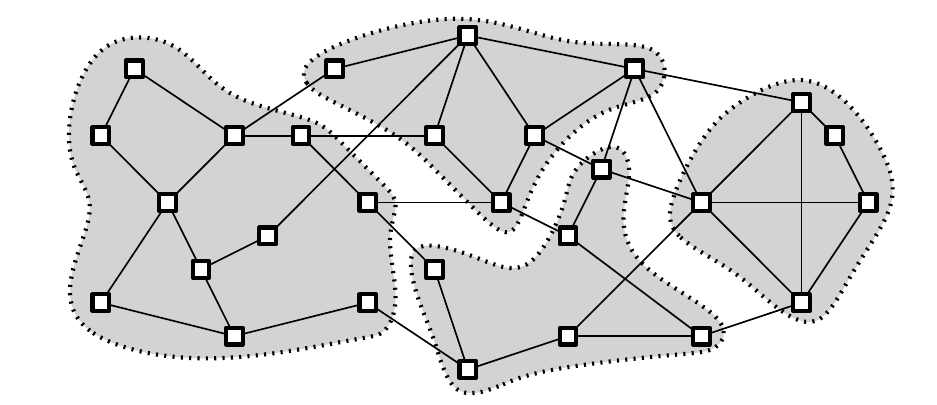}
    \hspace{0.1\textwidth}
    \includegraphics[page=7,width=0.3\textwidth]{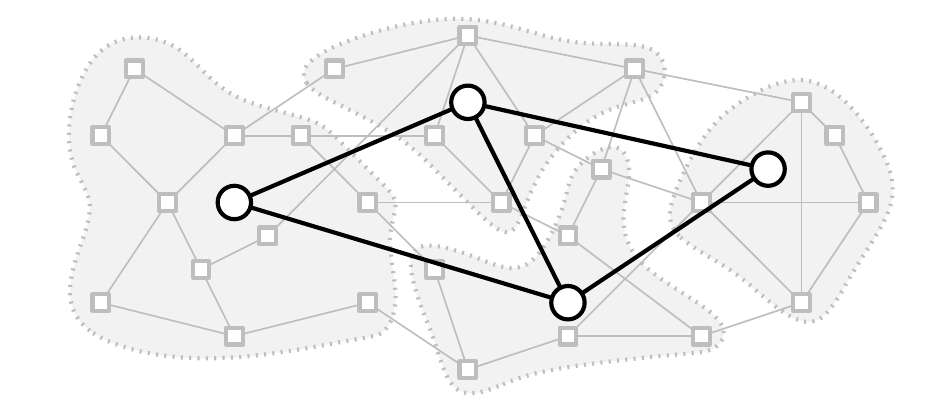}
    \caption{A communication graph $G$ with $4$ clusters (left), and the associated cluster graph $H$ (right).}
    \label{fig:cluster-graphs}
\end{figure}

Notice that this strictly generalizes the \congest model \cite{peleg00} where the input and communication graph coincide $H=G$ and all the clusters $V(v) = \set{v}$ are reduced to a single machine.
We also emphasize that communication on cluster graphs is more permissive than in the setting of \cite{FGHKN24} since it allows for efficient aggregation over messages from \emph{all} neighbors. Nonetheless, communication remains too restricted to naively simulate a \congest algorithm on $H$.

Importantly, the clusters might be internally poorly connected. Imagine, for instance, that each cluster is a tree with $\Delta$ leaves in the network. In that case, communications between different parts of the tree must go through one $O(\log n)$-bandwidth link. In particular, this means that any local computation by vertices of $H$ must admit an efficient communication protocol, which is generally not true of \congest algorithms.

Two natural primitives that illustrate this issue are determining a node's degree and finding a free color.
In \congest, the degree of a node is simply its number of incident edges. A simple aggregation within a cluster suffices to count its incident links. But this quantity can grossly overestimate the cluster's actual degree, given that two clusters can be connected through many links (as in \cref{fig:cluster-graphs}).
In fact, determining which edges connect to the same cluster amounts to a costly set intersection problem.
Finding a color available to $v$ by local computation in $V(v)$ alone requires $\Omega(\Delta/\log n)$ rounds in the worst case by a similar reduction (see \cref{fig:set-intersection}).

\begin{figure}
    \centering
\includegraphics[width=.6\linewidth,page=3]{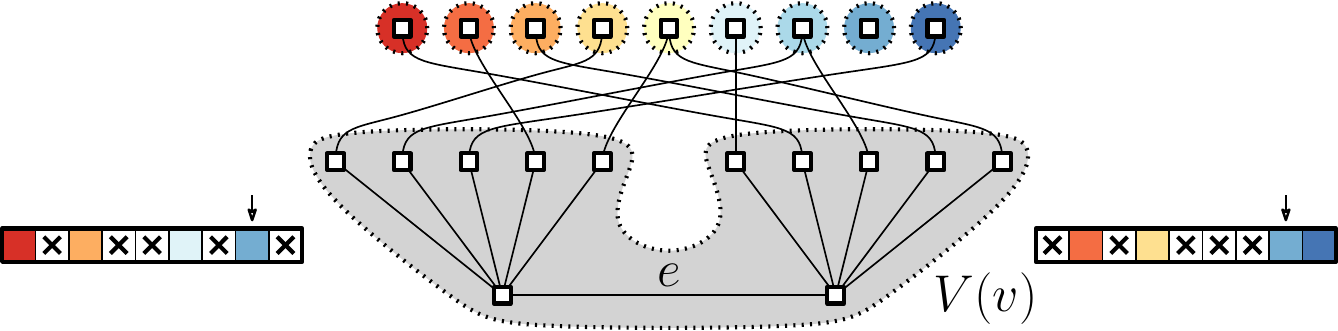}
    \caption{\label{fig:set-intersection} An uncolored cluster $V(v)$ with a tree-topology (in gray) and edges to colored neighbors on both sides of a bridge link $e$. Finding the one available color \emph{by communicating in $V(v)$ only} is at least as hard as solving a set-intersection instance where Alice (resp.\ Bob) encodes her input on the left (resp.\ right) inter-cluster links and all communication must go through the central $O(\log n)$-bit bandwidth link.}
\end{figure}

However, the aforementioned two tasks are easily performed for a node $v$ with the dedication of its neighbors.
Indeed, the neighbors of $v$ can each perform an aggregation to cut off all but one link to $V(v)$ to prevent double-counting. This allows $v$ to compute its degree exactly in one aggregation, and to find a free color using binary search.
Where the difficulty lies is in performing these tasks in parallel across the whole network.
This observation is useful to our algorithm in several parts, as we perform computations within disjoint subgraphs of the cluster graph.

\subsection{Contributions}
\label{sec:model-results}

We give an affirmative answer to our research question with \cref{thm:rand-general}.
We provide an algorithm to $\Delta+1$-color cluster graphs nearly as fast as the best known \local algorithm \cite{GG_focs24}. It also improves exponentially on the only existing algorithm for $\Delta+1$-coloring cluster graphs \cite{FGHKN24}. Note the linear dependency on the diameter of cluster graphs, which is unavoidable in general.

\begin{restatable}{theorem}{ThmRandGeneral}
\label{thm:rand-general}
There is a $O(\dilation \cdot \log^7\log n)$-round algorithm to $\Delta+1$-color cluster graphs of maximum degree $\Delta$, with high probability, where $\dilation$ is the maximum diameter of a cluster in the communication network.
\end{restatable}

Our result can be seen as a win for \emph{decentralization} in a distributed context. In more powerful distributed models like \local,
the distributed aspect often revolves around \emph{gathering} the information, which is then computed and acted on by a single processor.
In cluster graphs, the operation of each vertex is dispersed between multiple machines and these machines can only share a synopsis of the messages they receive. Computation in cluster graphs must therefore be collaborative and primarily based on aggregation (especially when $\Delta$ is large).

\paragraph{High-Degree Graphs.}
When $\Delta$ is larger than some $\poly(\log n)$, the \local (and \congest) algorithm colors all the vertices in $O(\log^* n)$ rounds with high probability. \cref{thm:high-degree} shows that when $\Delta$ is large enough, our cluster-graph algorithm matches this runtime.

\begin{restatable}{theorem}{ThmRandHigh}
    \label{thm:high-degree}
    There is a $O(\dilation \cdot \log^*n)$-round algorithm to $\Delta+1$-color cluster graphs of maximum degree $\Delta = \Omega(\log^{21} n)$, with high probability, where $\dilation$ is the maximum diameter of a cluster in the communication network.
\end{restatable}

This complexity is essentially tight given the slow growth of the log-star function.
The high-degree case often exhibits interesting behavior in the randomized setting, as it can make random events whose failure probability depends on the degree of the graph to become vanishingly rare. High-degree graphs are also outside the scope of the known connection between the deterministic and randomized complexities due to Chang, Kopelowitz, and Pettie~\cite{CKP_siamcomp19}, which hints at why some problems admit considerably faster algorithms in the high-degree regime.

\paragraph{Virtual Graphs.}
Cluster graphs form a particular type of a \emph{virtual graph}, where vertices correspond to connected --- but not necessarily disjoint --- subgraphs of a basegraph. We define these in a sibling paper \cite{us:partii}. As the definitions are somewhat technical, we defer them to \cref{sec:related-work-cluster-graph}. Virtual graphs unify and generalize various coloring problems, including vertex coloring, edge coloring, and distance-$k$ coloring. 
Since our algorithms' basic building blocks are aggregation and broadcast operations on trees spanning each cluster, when clusters overlap, the algorithms can be performed with an overhead proportional to the overlap between the clusters' spanning trees. Simply put,
\begin{quote}
\emph{
Everything in this paper immediately translates to virtual graphs, with the additional overhead factor of the edge congestion.
}
\end{quote}
In \cite{us:partii}, we gave algorithms that run in $\poly(\log\log n)$-rounds (with constant edge congestion) for the degree+1-coloring problem. The difference is that this "degree" includes multiplicities (i.e., counting multi-edges), so those results are strictly incomparable to the result given here.
In particular, \cref{thm:high-degree} implies the first fast algorithm for distance-2 coloring of the following form (see \cref{ssec:virtual-graphs} for more details).
\begin{corollary}
    There is a $\poly(\log\log n)$-round \congest algorithm for distance-2 coloring using $\Delta_2+1$ colors, where $\Delta_2 = \max_v |N_G^2(v)|$. The algorithm runs in 
    $O(\log^* n)$ rounds for $\Delta_2 = \Omega(\log^{21} n)$.
\end{corollary}

\subsection{Related Work}

In high-degree graphs, there are $O(\log^* n)$-round randomized $\Delta+1$-coloring algorithms in \local \cite{HKNT22} and \congest \cite{HNT22}. For low-degree graphs, the complexity is roughly speaking $\Theta(T_{DET}(\log n))$, where $T_{DET}(n)$ is the complexity of deterministic algorithms on $n$-node graphs. Given the best deterministic bounds, the randomized complexity today is
$\tilde{O}(\log^{5/3}\log n)$ in \local \cite{GG_focs24,CLP20} and $O(\log^3 \log n)$ in \congest \cite{GK21,HNT22}. Improved deterministic bounds are known when $\Delta$ is very small \cite{Bar2016,fraigniaud16,MT20,FK23,FK_podc24_ba}.

The known \CONGEST algorithms \cite{HKMT21,HNT22} work only when the input and communication graphs coincide. There is a $\poly(\log\log n)$-round randomized algorithm for the more restricted \emph{broadcast} \congest model, under the same assumption \cite{FGHKN23}.

\paragraph{Cluster Graphs.}
Cluster graphs are ubiquitous in distributed computing. While not exhaustive, \cref{sec:related-work-cluster-graph} highlights notable attempts to formalize the concept. 
The only non-trivial $(\Delta+1)$-coloring algorithm for cluster graphs is the $O(\log^2 n)$-round algorithm by \cite{FGHKN24}. It is based on a \emph{Distributed Palette Sparsification Theorem} that extends the work of \cite{ACK19}. This algorithm computes in $O(\log^2 n)$ rounds of \CONGEST a $(\Delta+1)$-coloring from $O(\log^2 n)$-sized random lists. This algorithm can be implemented in $\poly(\log n)$ rounds on cluster graphs because each vertex needs to send/receive messages from/to $O(\log^4 n)$ neighbors each round. However, they also show that such algorithms cannot run faster than $\Omega\parens*{\frac{\log n}{\log\log n}}$ rounds, even with arbitrarily large messages. We bypass this by leveraging cluster graphs' ability to aggregate information from all neighbors.

\paragraph{Weak Models.}
While locality is central to distributed computing, there has been a lot of attention to distributed models with weak units or limited communication ability. 
This includes nature-inspired models like population protocols \cite{angluin2004population}, beep model \cite{afek2013beeping}, stone-age model \cite{stoneage13}, and programmable matter \cite{amoebot14} as well as various wireless communication models \cite{BGI92,HW19} and broadcast CONGEST.
Most related to our setting in this list are the beep, wireless, and broadcast CONGEST models, which feature a synchronous notion of time.
Virtual graphs represent another form of locality weakening by restricting the way nodes receive, process, and transmit messages to the computation of \emph{aggregation functions}. While our results have no implication for other weak models, they offer insights into relaxations that enable efficient coloring algorithms.

\paragraph{Distributed Sketching.}
In this paper, we use a sketching technique based on aggregating the maximum of independent geometric variables for approximate counting, referred to as \emph{fingerprints} (see \cref{sec:fingerprints}). While novel for distributed graph algorithms, similar ideas have been extensively used and studied in other models.
For instance, in the streaming setting, the technique was initiated by \cite{FM85} and subsequently improved in, e.g., \cite{DF03,KNW10,B20}. 
Contrary to streaming algorithms, we rely on sampling instead of hashing, which avoids the complication of designing a hash function.

It was observed in \cite[Section 2.2]{P07} that fingerprints could be used for approximate counting in a distributed setting. However, this result uses $O(\log n \cdot \log\log n)$ bandwidth to succeed with probability $1 - 1/n$. We remedy this with a special compression scheme for fingerprints (\cref{lem:fingerprint-encoding}).

\section{Technical Overview}
\label{sec:core-ideas}
In this section, we give an overview of our techniques. As we build on the framework laid out by \cite{HSS18,ACK19,CLP20,HKNT22,FGHKN23}, we begin by presenting a model-independent distributed $(\Delta+1)$-coloring algorithm. Implementing this algorithm in each model (besides \LOCAL) requires additional work and we focus here on the obstructions specific to cluster graphs. These are particularly taxing in the densest regions of the graph --- called \emph{cabals} --- where sparsification arguments become ineffective. While straightforward routing allows for easy resolution in other models, it fails in cluster graphs for reasons expounded on in \cref{sec:intro-cluster}. This calls for a radically different approach that we summarize in \cref{ssec:overview-matching,ssec:overview-putaside}.

\subsection{The Distributed $(\Delta + 1)$-Coloring Meta-Algorithm}
At a high-level, the algorithm revolves on vertices trying
random colors. In very dense graphs such as a $(\Delta+1)$-clique, however, $\Omega(\log n)$ trials are needed to color all vertices.
The first step is then to identify the vertices for which random trials are effective. This uses a structural result of Reed \cite{Reed98} (see also \cite{HSS18,ACK19}) that partitions vertices into sparse nodes and dense clusters called \emph{almost-cliques} (\cref{def:ACD} for the decomposition).

\paragraph{Slack Generation.} Sparse nodes are easy to color because they get $\Omega(\Delta)$ \emph{slack} (see \cref{sec:high-level-alg}):
With a single random color trial, $\Omega(\Delta)$ pairs of nodes in a sparse node's neighborhood get colored the same.
The dense vertices (those in almost-cliques) also get some slack, but less, as their neighborhood contains fewer pairs of non-adjacent nodes, i.e., anti-edges. (See \cref{sec:appendix-slack-gen} for the slack generation step).

\paragraph{Multicolor Trials.}
As it turns out, slack is crucial for fast distributed $(\Delta+1)$-coloring. Indeed, it is well known that when vertices have slack linear in their uncolored degree, they can be colored in $O(\log^*n)$ rounds \cite{SW10,CLP20}. We henceforth refer to this algorithm as \emph{MultiColor Trials (MCT)} (see \cref{sec:appendix-mct} for more details) because vertices try exponentially increasing number of colors. This takes care of the sparse nodes, while we first need to reduce the uncolored degree of dense nodes.

\paragraph{Synchronized Color Trial.}
To handle the dense vertices, the trick is to \emph{synchronize} the color trials within each almost-clique. Each sublogarithmic distributed coloring algorithm proceeds differently and we follow here \cite{FGHKN23}. We skip over the details as they represent no major obstructions in cluster graphs (see \cref{lem:sct} for more details). The crucial point is that after this synchronized color trial, dense vertices have $\Omega(e)$ slack and $O(e)$ degree, where $e$ is the number of original neighbors outside their almost-clique.

\paragraph{Cabals \& Put-Aside Sets.}
The multicolor trial then colors all remaining dense vertices, \emph{except for those in the densest almost-cliques} where probabilistic arguments fail to apply with high probability.
We treat separately the \emph{cabals} -- the almost-cliques with nodes with
fewer than $c\log^{1.1} n$ external neighbors.

To color cabals, the idea is to \emph{put aside} some of their vertices. Indeed, by keeping $c\log^{1.1} n$ vertices uncolored, we provide the other vertices in the cabal enough slack to be colored fast.
Since cabals have so few outgoing edges, we can ensure that no edge connects put-aside sets from different cabals (see \cref{lem:compute-put-aside}). In particular, put aside sets can be colored last with simple information gathering  --- at least in \LOCAL.
\medskip

Let us summarize the steps of this algorithm:
\begin{algorithm}
    \caption{The Distributed $(\Delta+1)$-Coloring Meta-Algorithm\label{alg:meta}}
    \computeACD \tcp*[f]{partition vertices between sparse and almost-cliques}

    \computePutAside \tcp*[f]{in cabals only}

    \slackgeneration \tcp*[f]{same colors pairs of neighbors, generating slack}

    \sct \tcp*[f]{decreases the uncolored degree of dense vertices}

    \multitrial \tcp*[f]{colors vertices with slack in $O(\log^*n)$ rounds}

    \colorPutAside \tcp*[f]{in cabals only}
\end{algorithm}

\subsection{Challenges in Cluster Graphs \& The Actual Coloring Algorithm}

As soon as we move away from the \LOCAL model, \cref{alg:meta} needs sparsification because essentially all steps, besides slack generation, use $\omega(\log n)$-bit messages. Authors of \cite{HN23,FGHKN24} already observed that coloring of sparse vertices was possible in $O(\log^* n)$, so we focus on dense vertices. A long series of work sparsify various versions of \cref{alg:meta} for various models \cite{ACK19,CFGUZ19,HKMT21,HNT22,FGHKN23,FHN23,FGHKN24} but it remains that those sparsification techniques fails in cabals. 
In most models, this is not much of an issue because the algorithm can simply have each vertex learn $O(\log n \cdot \poly(\log \log n))$ bits of information. This is a hindrance for us as we seek a $O(\log^* n)$ round complexity with $O(\log n)$ bandwidth.

Let us now review in more details the challenges specific to \cref{alg:meta} in cluster graphs.

\paragraph{Challenge 1: Finding Anti-Edges in Cabals.}
Since vertices cannot learn their palettes, we approximate them by the \emph{clique palette}: the set of colors not used in the almost-clique. The issue is that in almost-cliques larger than $\Delta+1$, the clique palette may run out of colors before coloring all its vertices. To remedy this, such almost-cliques need to use some colors multiple times. To that end, the authors of \cite{ACK19} introduced the concept of a \emph{colorful matching}: pairs of same-colored vertices in each almost-cliques. As observed by \cite{FGHKN24}, when the average anti-degree in the almost-clique is $\Omega(\log n)$, repeating $O(1)$ random color trials works with high probability. Additionally, if the vertices of the almost-clique are sufficiently sparse, they receive enough slack from slack generation (i.e., a single random color trial). Cabals are precisely the almost-cliques where the number of anti-edges and the sparsity are insufficient for random color trials to be effective.
Algorithms for other constrained models (e.g., \cite{FGHKN24,FHN23}) resolve the issue by routing $\omega(\log n)$ bits of information in the almost-clique.

We resort to a non-conventional sampling technique based on \emph{geometric variables}. It allows us to perform $O(\log n)$ parallel random trials, where each trial costs 2 bits on average. Conveniently, geometric variables are easy to aggregate even in the presence of redundant paths and, surprisingly, this can be used to find anti-edges in cabals with high probability. We detail this in \cref{ssec:overview-matching}.

\paragraph{Challenge 2: Coloring Put-Aside Sets.}
Surprisingly, the most challenging part in cluster graphs is coloring the put-aside vertices. In \LOCAL, this step is trivial because each put-aside set has diameter at most two and put-aside sets from different cabals are independent (i.e., no edge connects them). 

\begin{wrapfigure}{r}{.45\linewidth}
\includegraphics[page=3,width=\linewidth,clip,trim=30 40 30 0]{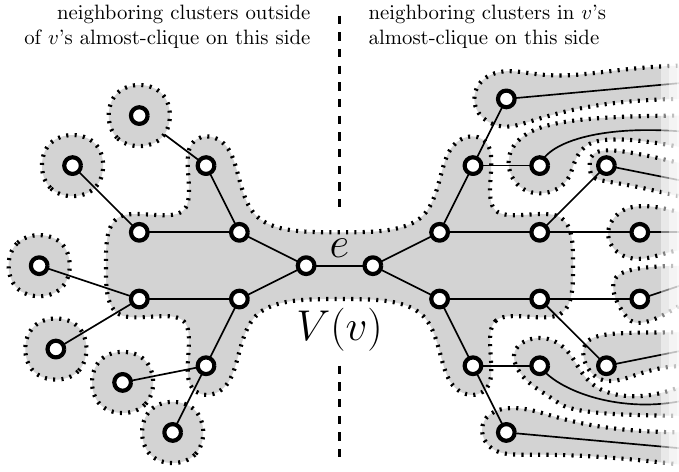}
\caption{All machines linked to the inside, on the right, are linked to the outside, on the left, by a unique link $e$. \label{fig:bridge}}
\end{wrapfigure}
In \CONGEST, one can observe that it suffices to learn the set of colors used by external neighbors of put-aside vertices as well as the clique palette (the colors not used in the almost-clique). 
Because cabals are so dense, this fits in $\log^{2.1} n$ bits, which with some routing tricks can be disseminated efficiently in cabals, even with $\Theta(\log n)$-bit messages.

In cluster graphs, the machines inside a cabal can be linked to the outside only through a single $O(\log n)$-bandwidth link, making routing near-impossible (see \cref{fig:bridge}).
As explained in \cref{sec:intro-cluster}, the mere task of finding a color for a single vertex is hard (recall \cref{fig:set-intersection}). Finally, observe that to obtain the $O(\log^* n)$ runtime claimed in \cref{thm:high-degree}, we must color all put-aside vertices ultrafast.

To solve this issue, we do not attempt to extend the coloring to put-aside sets (like in \LOCAL and \CONGEST), but will use the help of already colored vertices. This leads to a novel and non-trivial parallel color swapping scheme outlined in \cref{ssec:overview-putaside}.

\paragraph{The Actual Cluster Graph Algorithm.}
After this long setup, we can describe the steps of our actual algorithm for \cref{thm:high-degree}. That is an algorithm for $(\Delta+1)$-coloring cluster graphs with $\Delta \geq \poly(\log n)$ in $O(\log^* n)$ rounds. We refer readers interested in low-degree graphs to \cref{sec:low-deg}. Since cabals have so few anti-edges, we cannot take the risk of coloring some of their vertices with slack generation. This forces us to color the graph in more steps than \cref{alg:meta}: first run slack generation in $V \setminus \Vcabal$, then color the sparse vertices, then color almost-cliques that are \emph{not} cabals and, finally, color cabals.

Sparse vertices can be colored by multicolor trials, and the algorithm for non-cabals follows mainly the steps of \cref{alg:meta} for broadcast \CONGEST \cite{FGHKN23}. We, however, remark that since vertices in cluster graphs cannot compute (nor approximate) their anti-degrees, we must subtly modify the definition of inliers/outliers (see \cref{eq:def-inliers}). We then color cabals with the same general approach except that (1) we use a different algorithm for finding anti-edges and (2) we add steps for computing and coloring put aside sets.

We summarized the main steps in \cref{alg:flow}.

\begin{algorithm}
    \caption{\label{alg:flow} High-Level Structure of our Algorithm when $\Delta \geq \poly(\log n)$.}
    \computeACD\;
    \slackgeneration in $V \setminus \Vcabal$ \;
    \multitrial in $\Vsparse$\;
    \InEach(\tcp*[f]{Coloring Non-Cabals First (\cref{sec:non-cabals})}){$K$ which is \emph{not} a cabal}{
\colorfulmatching\;
    Color Outliers with \multitrial\;
    \sct\;
    $\approx$ \multitrial \tcp*{some additional technical work is required here}
    }
    \InEach(\tcp*[f]{Then Color the Cabals (\cref{sec:cabals})}){cabal $K$}{
\colorfulmatching \tcp*{Challenge 1 arises here}
    Color Outliers with \multitrial\;
    \computePutAside\;
    \sct\;
    \multitrial\;
    \colorPutAside \tcp*{Challenge 2 arises here}
    }
\end{algorithm}

\subsection{Fingerprinting \& Finding Anti-Edges in Cabals (\cref{sec:fingerprints,sec:colorfulmatching})}
\label{ssec:overview-matching}

To detect anti-edges, vertices will aggregate the maximum of certain geometric variables. In this paper, we call such aggregates \emph{fingerprints} and use them extensively in our algorithm. Primarily, we employ them to approximate fundamental quantities, such as the degree of vertices within a specific set (see \cref{sec:fingerprints}). As discussed in \cref{sec:intro-cluster}, vertices in cluster graphs cannot compute such basic quantities efficiently. To the best of our knowledge, approximate counting or similarity estimations based on geometric variables is novel in distributed graph algorithms; however, the idea has been used and studied extensively in other contexts, e.g., in \cite{FM85,DF03,KNW10,B20,P07}.

In this technical overview, we focus on a novel way to use fingerprints, specifically to solve the \emph{Challenge 1} mentioned earlier.

To illustrate how we use fingerprints in this context, let us consider a simplified setting where $K$ is a cabal such that at least $\Delta/2$ of its vertices are incident to an anti-edge and we only care to identify one of them. Every vertex samples a random geometric variable $X_u$ and we compute their maximum $Y^K$. Additionally, each vertex $v\in K$ computes the maximum $Y^v$ of the variables in its neighborhood. Observe that if $Y^v \neq Y^K$, we know that $v$ is the endpoint of an anti-edge. In particular, it happens when $Y^K = X_u$ for a \emph{unique} $u$ in $K$ and that $u$ has some anti-neighbor. It is not very difficult to verify that the maximum is unique with constant probability (\cref{lem:unique-maximum}) and that it occurs at a uniform vertex in $K$ (\cref{prop:uniform-maximum}). Thus, this scheme identifies an anti-edge with constant probability. Analysis shows that repeating this random experiment $\Theta(\log n)$ time in parallel yields a ``large enough'' matching of anti-edges with high probability (\cref{lem:fingerprint-matching}). Coloring the endpoints of these anti-edges with equal colors yields a colorful matching suitable to our needs.

The core advantage of using fingerprints resides in that they can be compressed very efficiently (\cref{lem:fingerprint-encoding}). It might seem that since $X_v \leq O(\log n)$ w.h.p., we represent each fingerprint in $O(\log\log n)$ bits, resulting in a total bandwidth usage of $O(\log n \cdot \log\log n)$ bits. However, we leverage the fact that $Y^v$ values are highly concentrated and encode \emph{their deviation} from some baseline value rather than their actual value. Since we only expect each trial to diverge by a constant amount, a concentration argument shows that the total deviation over all trials is $O(\log n)$ with high probability (\cref{lem:concentrated-maximums}).

\subsection{Coloring Put-Aside Sets (\cref{sec:put-aside-sets})}
\label{ssec:new-overview-putaside}
\label{ssec:overview-putaside}
\label{ssec:putaside-overview}

We now describe the most novel and technically involved idea in our coloring algorithm. To convey intuition without overloading the reader with technical details, we consider a simplified setting that illustrates the key ideas: let $\Delta \geq \Omega(\log^{21}n)$, set $r = \Theta(\log^{1.1} n)$ and suppose $H$ consists of $(\Delta +1 - r)$-cliques, where each vertex has $r$ external neighbors (i.e., neighbors outside their clique).

We suppose that we already computed a coloring $\col$ of almost all nodes: in each clique $K$,
only a \emph{put-aside} set $P_k$ of exactly $r$ vertices is uncolored.
Recall that no edges connect put-aside sets from different cliques.
Our claim is that we can fully color $H$ in constant rounds.

However, searching for available colors is computationally expensive (recall \cref{fig:set-intersection}). Instead, we ask the colored vertices to assist by offering to \emph{donate their color} to put-aside vertices. Intuitively, this approach is more efficient because it shifts the burden from a few vertices searching for available colors to many vertices finding ways to recolor themselves and donate their colors.
Nonetheless, donations are feasible only if
\begin{enumerate}
\item each donated color is not used by an external neighbor of the vertex receiving it;
\item each \emph{donor} has a \emph{replacement color} (or \emph{replacement}, for short); and
\item donors have different replacement colors.
\end{enumerate}
Our coloring algorithm thus solves a three-way matching problem in each clique: uncolored {vertices} are matched to {donors} and donors to {replacements} (see \cref{fig:S-sets}). 
We also need to ensure that recolorings do not create conflicts between different cliques.

Let us review the steps for computing this three-way matching.

\paragraph{Step 1: Finding Candidate Donors.}
We first compute a set $Q_K \subseteq K\setminus P_K$ of $\Omega(\Delta/r)$ candidate donors in each clique $K$ such that no edges connect $Q_K$ to candidate \& put-aside sets in other cliques. This ensures that each $K$ can be recolored independently. The sets $Q_K$ are constructed similarly to put-aside sets (see \cref{alg:Q-find}).
We henceforth focus on a fixed clique $K$ with uncolored vertices $u_1, u_2, \ldots, u_r$.

\paragraph{Step 2: Finding Replacement Colors.}
While it can be expensive to learn the colors available to a given node, we can quickly identify which colors are not used in $K$.
Extending a technique of \cite{FHN23}, we show that a node can query the $i\supth$ color in the clique palette $L(K)\eqdef [\Delta+1] \setminus \col(K)$, i.e., the set of colors not used in $K$. Observe that, in our example of $(\Delta+1-r)$-cliques with $r$ uncolored vertices, the clique palette contains $2r$ colors. As such, every vertex has $r$ colors available in $L(K)$ because it has at most $r$ external neighbors.
So if each vertex samples one color from $L(K)$ and check if it is available, we expect at least $|Q_K|/2 \geq \Omega(\Delta/r)$ of them find a valid replacement color.

\paragraph{Step 3: Finding Safe Donors.}
Surprisingly perhaps, we match the uncolored vertices to replacements colors directly instead of matching them to donors. In other words, each $u_i$ decides on a different $c_i$ and will accept only donations from vertices with $c_i$ as their replacement color. We proceed this way to eliminate recoloring conflicts within $K$, as by construction, donors for different $u_i$ use different replacements.

\begin{figure}
    \centering
    \includegraphics[width=.4\linewidth,page=1]{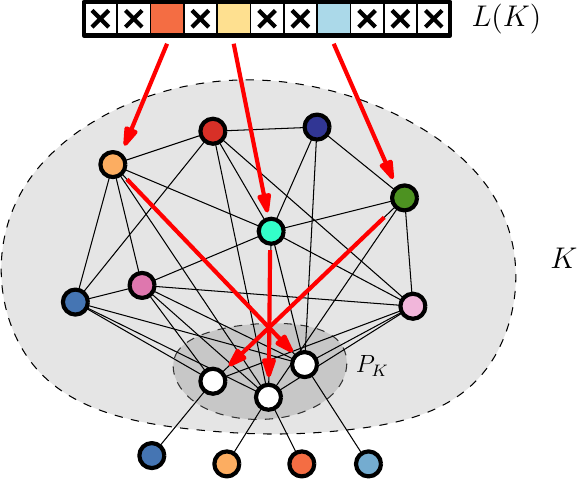}
    \hspace{.1\linewidth}
    \includegraphics[width=.45\linewidth,page=1]{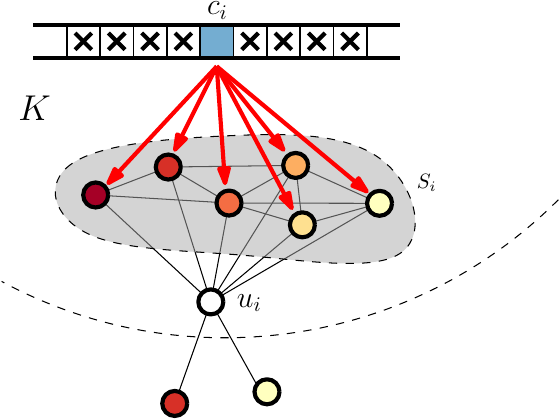}
    \caption{\label{fig:S-sets} On the left, the 3-way matching from replacement colors to colored vertices to put-aside vertices. Note that a put-aside vertex cannot be matched to some colored vertices because their colors are used by neighbors outside the cabal. On the right, the intermediate set $S_i$ of donors with the same replacement $c_i$. Observe that all donors from $S_i$ are colored with shades of red: their colors come from the same block. To describe colors used in $S_i$ to external neighbors, $u_i$ sends a message containing the first color or the block --- here, for instance, the color red --- and for each color the offset from that color.}
\end{figure}

For the final step, $u_i$ needs to represent the colors of its donors succinctly. So it will restrict itself further to only a subset $S_i$ of safe donors with similar colors to ensure they can be compressed into small messages. More precisely, we split the color space into
contiguous \emph{blocks} $B_1, B_2, \ldots, B_{\frac{\Delta+1}{b}}$ of $b = \poly(\log n)$ colors each.
Once vertex $u_i$ settles on a block $B_{j_i}$, it only accepts donations from vertices with colors in that block.
By choosing the block size large enough, each vertex $u_i$ can find a block containing a large set $S_i$ of vertices with colors from $B_{j_i}$ (possibly, and safely, sharing it with other vertices).

\paragraph{Step 4: Donating Colors.} In the last step, we match uncolored vertices to donors.
The only criteria remaining for matching $u_i$ to a donor $v\in S_i$ is that $\col(v)$ is not used by an external neighbor of $u_i$. By choosing $|S_i|$ large enough compared to $r$, the external degree of $u_i$,
at least one of $k = \Theta(\frac{\log n}{\log\log n})$ random donations in $S_i$ will be feasible with high probability.

To enable $u_i$ to test all $k$ donations using $O(\log n)$ bandwidth, we use that all nodes in the donor set $S_i$ use colors from the \emph{same block}. Hence $k$ donations from $S_i$ can be described by concatenating the index of the block and the offsets of the donations within the block. Overall, this scheme uses $O(\log \Delta + k\log b) = O(\log n)$ bits, for our choice of $k$ and $b \leq \poly(\log n)$ (see \cref{eq:put-aside-params}).

\subsection{Roadmap}

\cref{sec:core-ideas} gives a high-level overview of the techniques. We spend significant time introducing the state-of-the-art techniques for $(\Delta+1)$-coloring in an effort to enhance readability, especially among readers unfamiliar with the relevant work. The overall analysis is extensive and thus divided into nearly independent sections for better accessibility. The technical parts of the paper are organized as follow.
\begin{itemize}
    \item \cref{sec:prelim} formally defines cluster graphs and introduces useful aggregation lemmas.
    \item \cref{sec:coloring-alg} formally describes the overall algorithm, focusing on how individual components fit together and contains the proof of \cref{thm:high-degree}. It is structured in three parts:
    \begin{itemize}
        \item \cref{sec:high-level-alg} gives core definitions used throughout the paper,
        \item \cref{sec:non-cabals} describes the algorithm for non-cabals, and
        \item \cref{sec:cabals} describes the algorithm in cabals.
    \end{itemize}
    \item \cref{sec:fingerprints} analyzes the sketching method we call fingerprinting along with applications to approximating degrees in cluster graphs and computing the almost-clique decomposition.
    \item \cref{sec:colorfulmatching} provides the algorithm for the colorful matching in the densest cabals.
\item \cref{sec:put-aside-sets} provides the algorithm that colors put-aside sets in constant time.
    \item \cref{sec:prep-MCT} describes the careful application of multicolor trials in non-cabals.
    \item \cref{sec:low-deg} explains how we adapt the algorithm when $\Delta$ is small, obtaining \cref{thm:rand-general}.
\end{itemize}
The primary technical contributions of this paper are in \cref{sec:colorfulmatching,sec:put-aside-sets}.

We expand the related work on cluster graphs \& virtual graphs in \cref{sec:related-work-cluster-graph}. \cref{app:concentration,app:hashing} contain standard results on concentration of random variables and hash functions. To preserve the flow of the paper, less novel proofs are deferred to \cref{sec:omitted-proofs}.

For convenience, a table of notation is given in \cref{app:notation}.

\section{Preliminaries}
\label{sec:prelim}

\subsection{Notation}
\label{sec:notation}

For an integer $t \ge 1$, let $[t] \eqdef \set{1, 2, \ldots, t}$. For a function $f : \calX \to \calY$, when $X \subseteq \calX$, we write $f(X) \eqdef \set{f(x): x \in X}$; and when $Y \subseteq \calY$, we write $f^{-1}(Y) \eqdef \set{x\in X: f(x) \in Y}$. We abuse notation and write $f^{-1}(y) \eqdef f^{-1}(\set{y})$.
For a tuple $x = (x_1, x_2, \ldots, x_n)$ and $i\in [n]$, we write $x_{\le i} = (x_1, \ldots, x_i)$ (and similarly for $x_{<i}$, $x_{>i}$ and $x_{\ge i}$) and $x_{\neg i} = (x_1, \ldots, x_{i-1}, x_{i+1}, \ldots, x_n)$. When there are multiple indices $x_{i,j}$, we write $x_{i,*}$ to represent $\set{x_{i,j}}_j$.
{We write $f \gg g$ to say that there is a sufficiently large constant $c$ such that $f\geq c \cdot g$, asymptotically. The meaning of ``sufficiently large'' depends on context. Typically, the constant $c$ has to be large enough for the next statement to hold.}

For a graph $H = (V_H, E_H)$, the neighbors of $v$ in $H$ are $N_H(v) \eqdef \set{u\in V_H: \set{u,v}\in E_H}$ the adjacent vertices. The degree of $v$ in $H$ is $\deg(v; H) \eqdef |N_H(v)|$ its number of neighbors. When $H$ is clear from context, we drop the subscript and write $N(v) = N_H(v)$ and $\deg(v) = \deg(v; H)$. 
An unordered pair $\set{u,v} \subseteq V_H$ is called an \emphdef{anti-edge} or \emphdef{non-edge} if $\set{u,v} \notin E_H$.

For any integer $q \geq 1$, a \emphdef{partial $q$-coloring} is a function $\col : V_H \to [q]\cup\set{\bot}$ where $\bot$ means ``not colored''.
The domain $\dom \col \eqdef \set{v\in V: \col(v) \neq \bot }$ of $\col$ is the set of colored nodes. 
A coloring $\col$ is \emphdef{total} when all nodes are colored, i.e., $\dom\col =  V_H$; and we say it is \emphdef{proper} if $\bot \in \col(\set{u,v})$ or $\col(v) \neq \col(u)$ whenever $\set{u,v}\in E_H$.
We write that $\psi \succeq \col$ when a partial coloring $\psi$ \emphdef{extends} $\col$: for all $v\in \dom \col$, we have $\psi(v) = \col(v)$. 
The \emphdef{uncolored degree} $\deg_\col(v) \eqdef |N(v) \setminus \dom\col|$ of $v$ with respect to $\col$ is the number of uncolored neighbors of $v$.
The \emphdef{palette} of $v$ with respect to a coloring $\col$ is $L_\col(v) = [\Delta+1] \setminus \col(N(v))$, the set of colors we can use to extend $\col$ at $v$.

The \emphdef{slack} of $v$ respect to a coloring $\col$ (of $H$) and induced subgraph $H'$ (of $H$) is 
$s_\col(v) \eqdef |L_\col(v)| - \deg_\col(v; H')$.

\subsection{Model}
\label{sec:model}
Throughout the paper, $G = (V_G, E_G)$ is an $n$-vertex communication network and we call its vertices \emph{machines}. Recall that in the \local model, machines communicate on incident edges in synchronous rounds without bandwidth restrictions. Each machine $w \in V_G$ has access to local random bits; hence can generate a unique $O(\log n)$ bits identifiers $ID_w$ with high probability. The \congest model is the same except that each edge can carry at most one $O(\log n)$-bit message per round. We henceforth assume edges have $O(\log n)$ bandwidth.

Throughout, the graph $H = (V_H, E_H)$ to be colored is a \emphdef{cluster graph} on network $G$.

\begin{definition}[Cluster graph]
\label{def:cluster-graph}
    Let $G=(V_G, E_G)$ be a communication network, and assume that each machine in the network is given a \emphdef{cluster identifier} $v \in V_H$. For each $v\in V_H$, let $V(v) \subseteq V_G$ be the set of machines that are assigned the same cluster identifier $v$.
    For every $v \in V_H$, the set $V(v)$ is the \emphdef{cluster} of $v$, and we assume that $G[V(v)]$ is connected.

    The \emph{cluster graph} defined by the clusters $V(v), v\in V_H$ over $G$ is the graph $H=(V_H, E_H)$ where for every pair of nodes $u,v\in V_H$, $H$ contains the edge $uv \in E_H$ if and only if there exists machines $w_u\in V(u)$, $w_v\in V(v)$ such that $w_u w_v \in E_G$.
\end{definition}

We also refer to clusters as \emph{vertices} or \emph{nodes}, not to be confused with the \emph{machines} of the underlying communication graph.
The maximum degree of $H$ is $\Delta$. We assume both $n$ and $\Delta$ to be known to all vertices. 

We can assume without loss of generality that each cluster elected a leader and computed a \emphdef{support tree} $T(v) \subseteq E_G$ spanning $V(v)$. 
It is convenient to describe algorithms in terms of vertices/clusters rather than at the machine scale. Hence, we will abuse notation and associate the leader of each cluster $V(v)$ with the vertex it represents $v \in V_H$. Each round on $H$ consists of three steps: an initial broadcast in $T(v)$ from the leader, some computation on inter-cluster edges, and finally, some aggregation on $T(v)$ back to the leader.
The \emph{dilation} $\dilation$ is the maximum diameter in $G$ of a support tree. It is easy to see that if we count the rounds performed in $G$ (rather than counting the complexity in terms of $O(\log n)$-bit broadcast and aggregation operations), we get a multiplicative overhead of $\dilation$ and that such a linear factor is necessary.
For succinctness, we shall hide the multiplicative $\dilation$ dependency in the big-Oh notation.

\subsection{Aggregation on Cluster Graphs}
Breadth-first search is ubiquitous in distributed algorithms.
On cluster graphs in particular, a breadth-first search tree on $H$, the conflict graph, which induces a tree in $G$, the communication graph, that allows for concise aggregation.
Indeed, recall that two vertices in $H$ can be adjacent through multiple paths in the communication graph $G$, so aggregation over all paths in $G$ leads to double counting.
In contrast, aggregation on a tree $T_G \subseteq E_G$ ensures each vertex contributes exactly once to the aggregation. This is achieved with basic flooding. We emphasize, however, that
running a BFS from multiple vertices at the same time
could create congestion, 
and we 
only
perform parallel BFS in vertex-disjoint subgraphs; see \cref{sec:proof-bfs} for a proof.

\begin{restatable}{lemma}{LemBFS}
\label{fact:bfs}
Let $t\ge 1$. Let $H_1,\ldots,H_k \subseteq H$ be a collection of vertex-disjoint subgraphs of $H$ locally known to machines, each $H_i$ containing a single source node $s_i$. In $O(t)$ rounds of communication on $G$, we can simulate a $t$-hop BFS in each $H_i$ with source vertex $s_i$ in parallel for each $i\in[k]$. Let $T_{H,i} \subseteq E_{H_i}$ be the resulting BFS tree in $H_i$.
Each $T_{H,i}$ induces a tree $T_{G,i} \subseteq E_G$ on the communication graph $G$ such that each machine knows the edge leading to its parent in $T_{G,i}$. Moreover, $T_{G,i}$ has height (at most) $\dilation t$ and $T_{G,i} \subseteq \bigcup_{u\in V(T_{H_i})} T(u) \subseteq \bigcup_{u\in V(H_i)} T(u)$.
\end{restatable}

It is sometimes useful to count and order vertices.
An ordered tree is a rooted tree on which each node knows an (arbitrary) ordering of its children. Ordering the children of every vertex leads to a total order of the vertices in the tree: a vertex is always ordered after its ancestors and, if $u$ and $v$ are two vertices of the tree with $w$ as their lowest common ancestor, then we order $u$ and $v$ the same way $w$ orders its two children on the $wu$- and $wv$-path.
\cref{lem:prefix-sum} is obtained from a basic recursive algorithm on an ordered tree. See \cref{sec:proof-prefix-sum} for a proof.

\begin{restatable}{lemma}{LemPrefixSums}
\label{lem:prefix-sum}
    Let $T_1,\ldots,T_k$ be a collection of edge-disjoint ordered trees of depth at most $d$ on $G$, and for each $i\in [k]$, let $S_i \subseteq V_{T_i}$ be a subset of its vertices
    such that each $u\in S_i$ holds some integer $|x_u| \le \poly(n)$.
    There exists a $O(d)$-round algorithm for
    each $u\in S_i$ to learn $\sum_{w\in S_i: w \prec u} x_w$,
    where $\prec$ denotes the ordering on $S_i$ induced by $T_i$. The algorithm can be executed in parallel in each $T_i$ in the same $O(d)$ runtime.
\end{restatable}

\cref{lem:prefix-sum} can be used, for instance, to give uncolored vertices of some set $S \subseteq V_H$ distinct identifiers in $\set{1, 2, \ldots, |S|}$ by having $T$ be an arbitrary tree spanning $S$ and setting $x_u = 1$ if $u$ is uncolored and zero otherwise.

\section{The Coloring Algorithm}
\label{sec:coloring-alg}
The goal of this section is to describe our coloring algorithm when $\Delta \geq \poly(\log n)$. For the low-degree case, we refer readers to \cref{sec:low-deg}.
The emphasis is on how we assemble each piece together, while the most involved intermediate steps (coloring put-aside sets, colorful matching in cabals and reserved colors) are deferred to later sections. \cref{alg:general} describes the carefully chosen order in which nodes are colored.
The main coloring steps are \cref{line:color-sparse,line:color-non-cabals,line:color-cabals}. In \cref{line:compute-ACD}, we determine when each vertex will be colored. \cref{line:slackgen} provides slack necessary for \cref{line:color-sparse,line:color-non-cabals}. 

\begin{algorithm}[H]
    \caption{High-Level Coloring Algorithm for \cref{thm:high-degree}\label{alg:general}}

    \nonl\Input{A cluster graph $H$ on $G$ such that $\Delta \ge \Deltalow$}

    \nonl\Output{A $\Delta+1$-coloring} 

    \computeACD  \hfill (\cref{sec:ACD})
    \label{line:compute-ACD}

    \slackgeneration in $V \setminus \Vcabal$ \hfill (\cref{prop:slack-generation})
    \label{line:slackgen}

    \alg{ColoringSparse} 
    \label{line:color-sparse}

    \ColoringNonCabals \hfill (\cref{sec:non-cabals})
    \label{line:color-non-cabals}

    \ColoringCabals \hfill (\cref{sec:cabals})
    \label{line:color-cabals}
\end{algorithm}

Each step of \cref{alg:general} operates \emph{almost} independently of the other ones.
We provide pre/post-conditions of each step.

\subsection{The High-Level Algorithm}
\label{sec:high-level-alg}

In this section we give the definitions essential to our algorithm and state the properties of each step in \cref{alg:general}. We conclude the section with a proof of \cref{thm:high-degree}.

\paragraph{Global Parameters.}
We define here the precise values of some parameters used throughout the paper.
\begin{equation}
    \label{eq:params}
    \eps = 1/2000 \ , \quad
    \delta = \CSlack/300 \ , \quad
    \Deltalow = \Theta \parens*{ \log^{21} n } \quad\text{and}\quad 
    \lmin = \Theta \parens*{ \log^{1.1} n } \ ,
\end{equation}
where $\CSlack = \CSlack(\eps) \in (0,1)$ is the small universal constant from slack generation (\cref{prop:slack-generation}).

\paragraph{Sparse-Dense Decomposition.}
As with all sub-logarithmic $\Delta+1$-coloring algorithms, we begin by computing an almost-clique decomposition, which partitions nodes of the graph according to their sparsity.

\begin{definition}
    \label{def:sparsity}
    For any node $v$ in a graph $H=(V_H,E_H)$, its \emphdef{sparsity} is defined as the quantity 
    $\zeta_v \eqdef \frac{1}{\Delta}\parens*{\binom{\Delta}{2} - \frac{1}{2}\sum_{u \in N_H(v)} \card{N_H(u) \cap N_H(v)}}$.
\end{definition}

A node $v$ is said to be \emph{$\zeta$-sparse} if $\zeta_v \geq \zeta$, and $\zeta$-dense if $\zeta_v \leq \zeta$. Intuitively, sparsity counts the number of missing edges in a node's neighborhood.

\begin{definition}
    \label{def:ACD}
    For any graph $H=(V_H,E_H)$ and $\epsilon \in (0,1/3)$, an $\epsilon$-\emphdef{almost-clique decomposition} is a partition $V_H=\Vsparse\cup \Vdense$ such that
    \begin{enumerate}
\item each $v\in \Vsparse$ is $\Omega(\epsilon^2 \Delta)$-sparse,
\item\label[part]{def:AC} $\Vdense$ is partitioned into \emphdef{$\epsilon$-almost-cliques}: sets $K$ such that
            \begin{enumerate}[label=$(\roman*)$]
                \item $|K| \leq (1 + \eps)\Delta$, and 
                \item for each $v \in K$, $|N(v) \cap K| \geq (1-\eps)|K|$.
            \end{enumerate}
    \end{enumerate}
\end{definition}
Achieving this on cluster graphs requires using our fingerprinting technique. See \cref{sec:fingerprint-acd} for details on \cref{prop:ACD}.
\begin{prop}
    \label{prop:ACD}
    Let $\epsilon \in (0, 1/100)$. 
    Given a cluster graph $H$, \computeACD computes an almost-clique decomposition of $H$ in $O(1/\eps^2)$ rounds with high probability.
\end{prop}

\paragraph{Communication in Almost-Cliques.}
Significant information dissemination within almost-cliques can be achieved by splitting the cliques into random groups and using flooding within each group. Additionally, each group can perform more complex aggregation, including enumeration (order vertices such that each vertex knows its index in the ordering) and prefix sums (\cref{lem:prefix-sum}).
As each almost-clique $K$ is very dense, it follows from Chernoff bound that each random group is highly connected within $K$.
\begin{lemma}[\cite{FGHKN23}]
    \label{fact:random-groups}
    Let $K$ be an almost-clique and $x \ge 1$ be an integer such that $|K|/x \in \Omega(\log n)$.
Suppose each node picks a uniform $x_v \in [x]$.
    Then, w.h.p., sets $X_i = \set{v\in K: x_v = i}$ have size $\Theta(|K|/x)$ and each $v\in K$ is adjacent to more than half of $X_i$ for each $i\in[k]$. In particular, each $H[X_i]$ has diameter 2.
\end{lemma}

\paragraph{Cabals \& Non-Cabals.} 
For $v \in \Vdense$, let $K_v$ be the almost-clique containing $v$. The \emphdef{external-neighborhood} of $v$ is $E_v \eqdef N_H(v) \setminus K_v$ and the \emphdef{anti-neighborhood} of $v$ is  $A_v \eqdef K_v\setminus N_H(v)$.
We denote by $e_v \eqdef |E_v|$ its \emphdef{external-degree} and by $a_v \eqdef |A_v|$ its \emphdef{anti-degree}.
For almost-clique $K$, the average external- and anti-degree are $a_K \eqdef \sum_{v\in K} a_v/|K|$ and $e_K \eqdef \sum_{v\in K} e_v/|K|$. 
Each vertex approximates its external degree $\tilde{e}_v \in (1\pm \delta)e_v$ using the fingerprinting technique (\cref{lem:fingerprint}). A \emphdef{cabal} is an almost-clique such that $\tilde{e}_K \eqdef \sum_{v\in K} \tilde{e}_v / |K| < \lmin = \Theta(\log^{1.1} n)$. 
We denote by $\Kcabal$ the set of all cabals and by $\Vcabal$ the set of vertices $v$ such that $K_v \in \Kcabal$.

\paragraph{Reserved Colors.}
We avoid using certain colors in the earlier steps of the algorithm. Each almost-clique $K$ reserves the colors $\set{1,2,\ldots, r_K}$, where 
\begin{equation}
    \label{eq:reserved}
    r_K = 250 \cdot \max\set{\tilde{e}_K, \lmin} 
\end{equation} 
depends on the density of $K$. By extension, define $r_v=r_{K_v}$ (and $r_v = 0$ if $v \notin \Vdense$).
Note that in all $K$ the number of reserved colors is $r_K \leq 250(1+\delta)\eps\Delta \leq 300\eps\Delta$, a small fraction of the color space. Hence, those colors are dispensable in slack generation and in finding a colorful matching.

\paragraph{Types of slack (Degree, Temporary \& Reuse).}
Recall that the slack of $v$ (w.r.t.\ a (partial) coloring $\col$ and active subgraph $H'$) is $s_\col(v) = |L_\col(v)| - \deg_\col(v; H')$. A vertex can have slack for different reasons and it is important in our algorithm that we identify each one. First, if a vertex has a small degree, it gets slack $\Delta+1 - \deg(v)$. We call this \emphdef{degree slack}. 
In that vein, if only \emph{a subset} of the vertices try to get colored, the competition is decreased. 
Namely, when coloring an induced subgraph $H'$, we count only the uncolored neighbors in $H'$ against the available colors.
We call this \emphdef{temporary slack} because vertices 
not in $H'$ will need to get colored eventually, destroying the slack they provided while inactive. The third way a vertex receives slack is from neighbors using the same color multiple times. We call this \emphdef{reuse slack}. Formally, the reuse slack of $v$ is the difference between the number of colored neighbors and the number of colors used to color them: 
$|N(v) \cap \dom\col| - |\col(N(v))|$. This is the kind of slack provided by slack generation and the colorful matching.

\paragraph{Slack Generation.}
In \slackgeneration, each $v\in V\setminus \Vcabal$ tries one random colors in $[\Delta+1] \setminus [300\eps\Delta]$. (See \cref{alg:slack-generation} in \cref{sec:appendix-slack-gen} for the pseudo-code.) 
This results in many pairs of nodes in a neighborhood getting the same color, 
thereby generating \emph{reuse} slack.
Slack generation is brittle and must therefore be executed before coloring any other nodes.
We emphasize it does not use reserved colors and needs only to color a small fraction of the vertices (\cref{part:activation}).

\begin{restatable}[\cite{HKMT21,HNT21}]{prop}{PropSlackGeneration}
    \label{prop:slack-generation}
    Suppose $\Vsparse,\Vdense$ is an $\eps$-almost-clique decomposition. There exists a constant $\CSlack=\CSlack(\epsilon) \in (0,1)$ such that if $\Delta \geq \Omega(\CSlack^{-1}\log n)$ and $\colsg$ is the (partial) coloring produced by \slackgeneration, then $\colsg(V_H) \cap [300\eps\Delta] = \emptyset$ and with high probability,
    \begin{enumerate}
    \item\label[part]{part:slack-sparse} $s_{\colsg}(v) = |L_{\colsg}(v)| - \deg_{\colsg}(v) \ge \CSlack \cdot \Delta$ for all $v\in \Vsparse$;
\item\label[part]{part:slack-ext} $|N(v)\cap\dom\colsg| - |\colsg(N(v))| \ge \CSlack \cdot e_v$ for all $v\in \Vdense$ with $e_v \geq \Omega(\CSlack^{-1} \log n)$; and
    \item\label[part]{part:activation} each $K$ contains $|K \cap \dom \colsg| \leq |K|/100$ colored nodes.
    \end{enumerate}
\end{restatable}

\paragraph{Coloring Algorithms.}
We can now state properties required by each of the coloring algorithms in \cref{alg:general}. There are three coloring steps: sparse vertices $\Vsparse$, non-cabals $\Vdense \setminus \Vcabal$, and cabals $\Vcabal$. Coloring sparse nodes follows from \cref{prop:slack-generation,lem:mct} (MCT), and hence we defer details to the proof of \cref{thm:high-degree}. 
Non-cabal nodes need the slack from \slackgeneration, but cabal nodes must remain uncolored for \ColoringCabals to work.
Thus, we run \slackgeneration everywhere but in cabals, and then color $\Vdense \setminus \Vcabal$ first. Only then do we color $\Vcabal$ (\cref{part:non-cabal-cabal-uncol}).

Note that we need to retain all reserved colors for \multitrial (\cref{part:non-cabal-reserved}).

\begin{restatable}{proposition}{PropColorNonCabal}
    \label{prop:coloring-non-cabal}
    Let $\col$ be a coloring such that
    \begin{enumerate}[label=(NC-\arabic*)]
        \item\label[cond]{part:non-cabal-slackgen} 
            we did slack generation in $V \setminus \Vcabal$, i.e., $\col \succeq \colsg$;
        \item\label[cond]{part:non-cabal-cabal-uncol} 
            cabals are uncolored, i.e., $\Vcabal \subseteq V \setminus \dom\col$;
        \item\label[cond]{part:non-cabal-reserved} 
            no reserved color is used in non-cabals, i.e., $\col(K) \cap [300\eps\Delta] = \emptyset$ for all $K \notin \Kcabal$.
    \end{enumerate}
    Then, w.h.p., \ColoringNonCabals colors all vertices in $\Vdense \setminus \Vcabal$ in $O(\log^*n)$ rounds.
\end{restatable}

Finally, we color cabals.

\begin{restatable}{proposition}{PropColorCabal}
    \label{prop:coloring-cabals}
    Let $\col$ be a coloring where cabals are not colored.
Then, w.h.p., in $O(\log^*n)$ rounds \ColoringCabals extends $\col$ such that all cabals are colored.
\end{restatable}

\begin{figure}[H]
    \centering
    \begin{tikzpicture}[scale=0.8, every node/.style={scale=0.8}]
        \node (thm) [box] {\cref{thm:high-degree}};
        \node (cabal) [box, below of = thm] {\ColoringCabals \\ \cref{prop:coloring-cabals}};
        \node (noncabal) [box, right of = cabal, xshift = 2.5cm] {\ColoringNonCabals \\ \cref{prop:coloring-non-cabal}};
        \node (acd) [box, left of = cabal, xshift = -2cm] {\computeACD\\ \cref{prop:ACD}};
    
        \node (colorfulmatching) [box, below of = cabal, xshift = 2cm] {\alg{ColorfulMatching} \\ \cref{sec:colorfulmatching}};
        \node (putaside) [box, below of = cabal, xshift = -2cm] {\colorPutAside \\ \cref{sec:put-aside-sets}};
        \node (slackgen) [box, below of = noncabal, xshift = 1.6cm] {\slackgeneration \\ \cref{prop:slack-generation}};
        \node (prepmct) [box, below of = noncabal, xshift = 5cm] {\alg{PrepMCT} \\ \cref{sec:prep-MCT}};
    
        \draw [arrow] (prepmct) -- (noncabal);
        \draw [arrow] (colorfulmatching) -- (cabal);
        \draw [arrow] (colorfulmatching) -- (noncabal);
        \draw [arrow] (putaside) -- (cabal);
        \draw [arrow] (slackgen) -- (noncabal);
    
        \draw [arrow] (acd) -- (thm);
        \draw [arrow] (cabal) -- (thm);
        \draw [arrow] (noncabal) -- (thm);
    \end{tikzpicture}
    \caption{Flowchart of the dependencies.\label{fig:flowchart}}
\end{figure}

Together, these propositions imply our main theorem.
\cref{fig:flowchart} keeps track of the dependencies between propositions with their subroutines.

\begin{proof}[Proof of \cref{thm:high-degree}]
    By \cref{prop:ACD}, \computeACD returns an $\eps$-almost-clique decomposition in $O(1/\eps^2) = O(1)$ rounds. Each $v\in \Vdense$ computes $\tilde{e}_v \in (1 \pm \delta) e_v$ in $O(1/\delta^2) = O(1)$ rounds using the fingerprinting technique (\cref{lem:fingerprint} with $P_v(u)= 1$ iff $u\notin K_v$). Through aggregation on a BFS tree spanning $K$, vertices compute $|K|$ exactly and approximate the average external degree $\tilde{e}_K \in (1\pm \delta)e_K$. In particular, each $v\in \Vdense$ knows if $K_v \in \Kcabal$.
    
    After slack generation, w.h.p., all sparse nodes have $\CSlack \Delta$ slack by \cref{prop:slack-generation}. After $T = \frac{64}{\CSlack^{4}}\ln(\frac{4}{\CSlack}) = O(1)$ rounds of trying random colors in $[\Delta+1]$, w.h.p., the maximum uncolored degree in $\Vsparse$ is decreased to $(1-\CSlack^4/64)^T \Delta \leq \CSlack/4 \cdot \Delta$ (\cref{lem:try-color} with $\calC(v) = [\Delta+1]$, $S = \Vsparse$ and $\gamma = \CSlack$). Sparse nodes are then colored w.h.p.\ in $O(\log^* n)$ rounds with \multitrial (\cref{lem:mct} with $\calC(v) = [\Delta+1]$ and $\gamma = \CSlack/4$).
    
    Preconditions for \ColoringNonCabals are verified because we ran \slackgeneration in $V \setminus \Vcabal$ (\cref{part:non-cabal-slackgen}), \alg{ColoringSparse} and \slackgeneration do not color vertices of $\Vcabal$ (\cref{part:non-cabal-cabal-uncol}) and \slackgeneration does not use reserved colors (\cref{part:non-cabal-reserved}). With high probability, \ColoringNonCabals colors $\Vdense \setminus \Vcabal$ in $O(\log^* n)$ rounds. None of the vertices in $\Vcabal$ have been colored thus far. By \cref{prop:coloring-cabals}, w.h.p., \ColoringCabals colors $\Vcabal$ in $O(\log^* n)$ rounds.
\end{proof}

\subsection{Coloring Non-Cabals}
\label{sec:non-cabals}
This section aims at proving \cref{prop:coloring-non-cabal} by arguing the correctness of \cref{alg:coloring-non-cabals}. The overall structure follows the one of \cite{FGHKN23} with important internal changes. Inliers are defined differently because vertices cannot approximate anti-degrees (hence $a_K$) accurately. \cref{line:non-cabal-complete} in \cref{alg:coloring-non-cabals} requires careful approximation of palette sizes, which is technical and deferred to \cref{sec:prep-MCT} to preserve the flow of the paper.
\PropColorNonCabal*

\begin{algorithm}
    \caption{\alg{ColoringNonCabals}\label{alg:coloring-non-cabals}}

    \nonl\Input{A coloring $\col$ such as given in \cref{prop:coloring-non-cabal}}

    \nonl\Output{A total coloring of $V \setminus \Vcabal$}

    \alg{ColorfulMatching}
    \label{line:non-cabal-matching}

    \alg{ColoringOutliers}
    \label{line:non-cabal-sparse-outliers}

    \sct \label{line:non-cabal-sct}

    \alg{Complete}
    \label{line:non-cabal-complete}
\end{algorithm}

\paragraph{Clique Palette.}
Our algorithm relies heavily on the use of the \emphdef{clique palette}, which for an almost-clique $K$ and coloring $\col$ is the set of colors
$L_\col(K) \eqdef [\Delta+1] \setminus \col(K)$.
In cluster graphs, a node does not know (and cannot, in general, learn) its palette. Nonetheless, vertices 
can \emph{query} $L_\col(K)$ as a distributed data structure. \cref{lem:query} is an adaptation of \cite{FHN23}, hence its proof is deferred to \cref{sec:proof-query}.

\begin{restatable}{lemma}{LemQuery}
    \label{lem:query}
    Assume $\Delta \gg \log n$.
    Let $\col$ be any (partial) coloring of almost-clique $K$
    and $\calC(v) \in \set{\col(K_v), L_\col(K_v)}$. 
    If each $v\in \Vdense$ holds $1\le a_v \le b_v \le \Delta+1$, then, w.h.p., each $v$ can either
    \begin{enumerate}
        \item learn $|\calC(v) \cap [a_v, b_v]|$; or 
        \item if it has $1 \leq i_v \leq b_v$, learn the $i_v\supth$ color in $\calC(v) \cap [a_v, b_v]$.
    \end{enumerate}
    The algorithm runs in $O(1)$ rounds.
\end{restatable}

\paragraph{Colorful Matching.}
The problem with the clique palette is that we might use all the colors of $L_\col(K)$ before all nodes in $K$ are colored, as $K$ can contain as many as $(1+\eps)\Delta$ vertices.
For this purpose, we compute a \emph{colorful matching} \cite{ACK19}. That is we use $\Omega(a_K/\eps)$ colors to color \emph{twice as many nodes in $K$}, thereby creating reuse slack. Informally, the assumption on $\col$ in \cref{lem:colorful-matching-high} means that the algorithm works when no vertices of $K$ are colored \emph{or} when the coloring was produced \slackgeneration. This follows from previous work, thus implementation details are deferred to \cref{sec:appendix-colorful-matching}.

\begin{restatable}{lemma}{LemMatchingHigh}
    \label{lem:colorful-matching-high}
    Let $\calF$ be a set of almost-cliques with $a_K \in \Omega(\log n)$ and $\col$ a coloring such that for each $K\in\calF$, either $\col$ colors no vertices in $K$ or the restriction of $\col$ to $K$ coincides with $\colsg$.
    There exists a $O(1/\epsilon)$-round algorithm (\cref{alg:colorful-matching}) that outputs $\colcm \succeq \col$ such that, w.h.p., in each $K\in \calF$ the coloring $\colcm$ uses $M_K \geq \Omega(a_K/\eps)$ colors to color at least $2M_K$ vertices of $K$ with them. Moreover, it does not use reserved colors (i.e., $\colcm(V_H) \cap [300\eps\Delta] = \emptyset$) and colors a vertex iff at least one other vertex in $K$ also uses that color (i.e., if it provides reuse slack).
\end{restatable}

\paragraph{Inliers \& Outliers.}
Nodes that differ significantly from the average may not receive enough slack from slack generation and colorful matching to be colored later in the algorithm. Those nodes are called \emphdef{outliers} $O_K \subseteq K$ while their complement in $K$ is called \emphdef{inliers} $I_K = K \setminus O_K$. It would suffice to guarantee $e_v \leq O(e_K)$ and $a_v \leq O(a_K)$. While external degrees can be approximated (allowing the first condition in \cref{eq:def-inliers} to be verified), approximating anti-degrees is more challenging. 
We exploit the following relation (derived from counting neighbors of $v$ inside and outside $K$)
\[
    \Delta+1 = (\Delta - \deg(v)) + \deg(v) + 1 = (\Delta - \deg(v)) + |K| + e_v - a_v \ .
\]
Hence, nodes can approximate their anti-degree as
\begin{equation}
    \label{eq:x}
    x_v \eqdef |K| - (\Delta + 1) + \tilde{e}_v \quad \in \quad a_v - (\Delta - \deg(v)) \pm \delta e_v \ .
\end{equation}
Intuitively, the error made in \cref{eq:x} is compensated for by the slack provided to $v$.
Inliers are then defined as
\begin{equation}
    \label{eq:def-inliers}
    I_K \eqdef \set*{u \in K: \tilde{e}_v \leq 20 \tilde{e}_K \text{ and } x_v \leq \frac{M_K}{2} + \frac{\CSlack}{8} \tilde{e}_K } \ .
\end{equation}
Henceforth, we focus on coloring inliers and assume all outliers have been colored. Outliers are colored after the colorful matching, while they have $\Omega(\Delta)$ temporary slack from adjacent uncolored inliers. We refer readers to the proof of \cref{prop:coloring-non-cabal} at the end of this subsection for more details. Nonetheless, inliers must represent a large enough constant fraction of each almost-clique.

\begin{lemma}
    \label{lem:size-inliers}
    For $K \notin \Kcabal$, the number of inliers is
    $|I_K| \geq 0.85|K| \geq 0.8\Delta$.
\end{lemma}

\begin{proof}
Let $Z$ be the set of nodes $v$ in $K$ with $a_v \le 20a_K$ and $e_v \le 15 e_K$. We claim that all nodes in $Z$ are inliers. The lemma follows then, since by Markov at most $(1/15 + 1/20)|K| \le 0.15 |K|\leq 0.15(1+\eps)\Delta$ nodes are outside $Z$.
We first derive a useful bound (letting $M_K = 0$ when $a_K = O(\log n)$, as no colorful matching is computed in that case):
\begin{equation}
  80 a_K \le M_K + \CSlack e_K/8\ .
    \label{claim:ak}
\end{equation}
\cref{claim:ak} holds when 
$a_K \gg \log n$, because $M_K \ge \Omega(a_K/\eps) \ge 80a_K$; while when $a_K = O(\log n)$, then $a_K \ll e_K$, since $e_K = \Omega(\log^{1.1} n)$ in non-cabals.
Consider $v \in Z$. 
Setting $\delta \le \CSlack/300$, by the definition of $x_v$ and $Z$ and \cref{claim:ak}, 
\[ x_v \le a_v + \delta e_v \le 20 a_K + \frac{15\delta}{1-\delta} \tilde{e}_K 
 \le \parens*{ \frac{M_K}{2} + \frac{\CSlack}{16}\tilde{e}_K }  + \frac{\CSlack}{16}\tilde{e}_K = \frac{M_K}{2} + \frac{\CSlack}{8}\tilde{e}_K \ . \]
Hence, $v$ is an inlier, as claimed. 
\end{proof}

As we argue in the next lemma, all vertices classified as inliers received sufficient slack \emph{even when restricted to colors of the clique palette}. \cref{eq:reuse-slack} will be crucial in coloring the inliers remaining after the synchronized color trial.

\begin{lemma}
    \label{lem:reuse-slack}
    There exists a universal constant $\CCSlack = \CCSlack(\epsilon) \in (0,1)$ such that the following holds.
    Let $\col$ be the coloring produced by running slack generation and colorful matching. Then, w.h.p., for all inliers $v\in I_{K_v}$ in non-cabals $K_v \notin \Kcabal$, the reuse slack of $v$ \emph{even when restrained to non-reserved colors available in the clique palette} is at least:
    \begin{equation}
        \label{eq:reuse-slack}
        |\set{u\in K_v \cup E_v: \col(u) > r_v}| 
        - |\set{c\in[\Delta+1] \setminus [r_v]: c \in \col(K_v\cup E_v)}| \geq 
        \CCSlack e_K + 40 a_K + x_v \ .
    \end{equation}
\end{lemma}

\begin{proof}
    Slack generation creates $\CSlack \cdot e_v$ reuse slack in $N(v)$ \emph{when} $e_v \gg \CSlack^{-1}\log n$ (\cref{prop:slack-generation}) and the colorful matching creates $M_K = \Omega(a_K/\eps) \ge 80 a_K$ reuse slack \emph{when} $a_K \gg\log n$, w.h.p. Neither algorithm uses reserved colors. Recall that $e_K \ge \lmin/2$ in non-cabals and by definition of inliers $x_v \leq \CSlack/8 \cdot e_K + M_K/2$.
    \cref{eq:reuse-slack} follows from case analysis. 

    \begin{itemize}
        \item First, suppose that $a_K \geq e_K/2$. Since $K\not\in \Kcabal$, we have $a_K \geq e_K/2  \geq \lmin/4 \gg \CSlack^{-1}\log n$. Hence, the colorful matching provides enough reuse slack: $M_K \geq  e_K + 40 a_K + x_v$.

        \item Next, suppose $e_v \geq e_K/4$ and $a_K \le e_K/2$. Then the reuse slack is $\CSlack \cdot e_v + M_K \geq \CSlack/4 \cdot e_K + M_K \ge \CSlack/8 \cdot e_K + M_K/2 + x_v$. This is at least $\CSlack/16 \cdot e_K + 40 a_K + x_v$, by \cref{claim:ak}.  

        \item Finally, suppose $e_v \leq e_K/4$ and $a_K \le e_K/2$. Then, the clique must be smaller than $\Delta$ because $(\Delta+1) - |K| \geq e_K - a_K \geq e_K/2$. In particular $x_v \leq \tilde{e}_v - e_K/2 \leq -e_K/5$ (for $\delta < 1/5$). Then the reuse slack from colorful matching is at least
        $M_K \ge (x_v + e_K/5) + M_K \ge x_v + \CSlack e_K + 40 a_K$, by \cref{claim:ak} and that $\CSlack \le 1/10$.
    \end{itemize}
    Combined, \cref{eq:reuse-slack} holds with $\CCSlack \eqdef \CSlack / 16$.
\end{proof}

\begin{lemma}
    \label{lem:clique-slack-non-cabals}
    For any $\col$ extending the coloring produced by slack generation and the colorful matching, w.h.p., in all non-cabals $K\notin \Kcabal$, there are at least
    $|L_\col(v) \cap L_\col(K)| \geq |(N(v)\cup K) \setminus \dom \col|$ colors available to $v\in I_K$ in the clique palette.
\end{lemma}

\begin{proof}
    Let $R = |(N(v)\cup K_v) \cap \dom\col| - |\col(N(v) \cup K_v)|$ be the reuse slack in $N(v) \cup K_v$. The number of colors in the clique palette available to $v$ is
    \[
        |L_\col(v) \cap L_\col(K)| \geq 
        \Delta+1 - |K \cap \dom\col| - |E_v \cap \dom\col| +  R \ .
    \]
    Then, using $|K \cap \dom\col| = |K| - |K \setminus \dom\col|$, $|E_v \cap \dom\col| = e_v - |E_v \setminus \dom\col|$, and $\Delta+1 - |K| = (\Delta - \deg(v)) + e_v - a_v$, this becomes
    \begin{align*}
        |L_\col(v) \cap L_\col(K)| - |(N(v) \cup K) \setminus \dom\col| 
        &\geq (\Delta - \deg(v)) + R - a_v \\ 
        &\geq (\Delta - \deg(v)) + R - x_v - (a_v - x_v) \\
        &\geq R - x_v - \delta e_v  \tag{by definition of $x_v$} \ .
    \end{align*}
    \cref{eq:reuse-slack} implies the reuse slack $R \geq \CCSlack \cdot e_K + x_v$ is large. Since $e_v \leq 25e_K$ and $\delta < \CCSlack/25$ is small, $R \geq x_v + \delta e_v$ which concludes the proof.
\end{proof}

\paragraph{SynchronizedColorTrial.}
We use the version of the synchronized color trial from \cite{FGHKN23}.
Sampling a truly uniform permutation in cluster graphs is challenging. Instead, we sample a permutation from a set of pseudo-random permutations, which only affects the success probability by a constant factor. We defer details of the implementation to \cref{sec:sct-proof}.

\begin{restatable}{lemma}{LemSCT}
\label{lem:sct}
For almost-clique $K$, let $S \subseteq K$ and $\alpha \in (0,1]$ such that $\alpha|K| \leq |S_K| \leq |L_\col(K)| - r_K$.
Suppose $\pi$ is a uniform permutation of $[|S|]$ and the $i\supth$ vertex of $S$ tries the $i\supth$ color in $L_\col(K) \setminus [r_K]$. 
Then, w.h.p., $|S \setminus \dom\colsct| \leq \frac{24}{\alpha}\max\set{e_K, \lmin}$. This holds even if random bits outside $K$ are adversarial.
\end{restatable}

\cref{lem:clique-slack-non-cabals} implies that $|L_\col(K)| \geq |K \setminus \dom\col|$ for any coloring after the colorful matching. In particular, if we let $S \subseteq K \setminus \dom\col$ be all but $r_K$ of the uncolored inliers, the conditions of \cref{lem:sct} are verified. We refer readers to the proof of \cref{prop:coloring-non-cabal} at the end of this subsection for more details.

\paragraph{Preparing MultiColorTrial (\cref{sec:prep-MCT})}
The prior steps of the algorithm produced a coloring where uncolored vertices have slack $\Omega(e_K)$ for a small constant while uncolored degrees are $O(e_K)$ for a large hidden constant. Before we can apply \multitrial, we must reduce uncolored degrees to a small constant factor of the slack (\cref{part:mct-slack} of \cref{lem:mct}). Moreover, we must do so without using too many reserved colors.
This necessitates detecting vertices with enough slack in reserved colors, which is challenging in cluster graphs because vertices do not have access to their palettes. Hence, some complications ensue that are deferred to \cref{sec:prep-MCT}. We emphasize that this problem does not occur in cabals because we can easily adjust the size of put-aside sets and colorful matching (contrary to the slack received during slack generation).

\begin{restatable}{proposition}{PrepMCT}
    \label{prop:complete}
    Let $\col$ be a coloring such that 
reserved colors are unused  ($[r_K] \cap \col(K) = \emptyset$ in all $K\not\in\Kcabal$), \cref{eq:reuse-slack} holds, and vertices have uncolored degree $O(e_K)$. There is an algorithm that extends $\col$ to $\Vdense \setminus \Vcabal$ in $O(\log^* n)$ rounds with high probability.
\end{restatable}

We can now prove that \cref{alg:coloring-non-cabals} indeed colors non-cabal vertices in $O(\log^*n)$ rounds with high probability.

\begin{proof}[Proof of \cref{prop:coloring-non-cabal}]
    We argue that dense non-cabal vertices are colored in $O(\log^* n)$ rounds.
    We go over each step of \cref{alg:coloring-non-cabals}.

    \paragraph{Colorful Matching (\cref{line:non-cabal-matching}).}
    We run the colorful matching algorithm in all almost-cliques. Using the query algorithm (\cref{lem:query}) to compare the number of colors in $L(K)$ before and after computing the colorful matching, vertices learn $M_K$. (A vertex is colored in \cref{line:non-cabal-matching} iff it provides slack). If $M_K \geq 2\eps\Delta$, then all nodes of $K$ have $M_K - a_v \geq \eps\Delta$ slack (because $|L(v)| \geq \Delta+1 - |(K_v\cup N(v))\cap\dom\col| + M_K \geq \deg_\col(v) + M_K - a_v$.) If this occurs, we can color \emph{all nodes} of $K$ using $\calC(v) = [\Delta+1]$ (running \trycolor for $O(\eps^{-4}\log\eps^{-1}) = O(1)$ rounds and then \multitrial). We henceforth assume $M_K \leq 2\eps\Delta$. 
    
    \paragraph{Coloring Outliers (\cref{line:non-cabal-sparse-outliers}).}
    Since slack generation colored at most $0.01|K| \leq 0.02\Delta$ nodes in $K$, there are at least $(0.8 - 0.02 - 2\eps)\Delta \geq 0.75\Delta$ uncolored inliers for $\eps < 3/200$. Moreover, each outlier is adjacent to at least $(0.75 - \eps)\Delta \geq 0.5\Delta$ uncolored inliers. After removing the at most $300\eps\Delta$ reserved colors, outliers still have $0.25\Delta$ slack. Hence outliers are colored in $O(\log^* n)$ rounds \emph{without using reserved colors} (setting $\calC(v) = [\Delta+1]\setminus[r_K]$, using \trycolor for $O(1)$ rounds and then \multitrial). We henceforth assume outliers are colored and focus on inliers.

    \paragraph{Synchronized Color Trial (\cref{line:non-cabal-sct}).}
    In each non-cabal, define $S_K \subseteq K\setminus\dom\col$ as an arbitrary set of $|K \setminus\dom\col| - r_K$ uncolored inliers that participate in the synchronized color trial. 
    There are at least $0.75\Delta$ uncolored inliers and $r_K \leq 300\eps\Delta$; hence, the number of vertices participating in the synchronized color trial is $|S_K| \geq 0.75\Delta - r_K \geq (0.75 - 300\eps) \Delta \geq 0.5\Delta$ for $\eps < 1/900$. On the other hand, \cref{lem:clique-slack-non-cabals} implies that $|L_\col(K)| \geq |K\setminus\dom\col| = |S_K| + r_K$. Hence, both conditions of \cref{lem:sct} are verified.
    By \cref{lem:rep-permutation}, we implement the synchronized color trial in $O(1)$ rounds. 
    
    \paragraph{Finishing the Coloring (\cref{line:non-cabal-complete}).}
    After the synchronized color trial, by \cref{lem:sct}, each $K$ contains at most $r_K + 50 e_K \leq 300 e_K$ uncolored vertices, with high probability. Adding external neighbors, the maximum uncolored degree is $300 e_K + e_v \leq 350 e_K$ (as $e_v \leq 20\frac{1+\delta}{1-\delta}e_K \leq 50e_K$). Recall slack generation and colorful matching do not use reserved colors. We were careful not to use reserved colors during the synchronized color trial. By \cref{lem:reuse-slack}, w.h.p., \cref{eq:reuse-slack} holds. All conditions of \cref{prop:complete} are therefore verified and we complete the coloring of $\Vdense \setminus \Vcabal$ in $O(\log^* n)$ rounds.
\end{proof}

\subsection{Coloring Cabals}
\label{sec:cabals}
Once non-cabal vertices are colored, we color cabals. We emphasize that we make no assumption about the coloring computed in \ColoringNonCabals besides that it does not color any vertex in $\Vcabal$.
The task of this subsection is to prove \cref{prop:coloring-cabals}, by arguing the correctness of \cref{alg:coloring-cabals}. Since there is a significant overlap between \cref{alg:coloring-non-cabals,alg:coloring-cabals}, the exposition focuses on the differences (put-aside sets and colorful matching). A proof going over all steps of \cref{alg:coloring-cabals} can be found at the end of the subsection.

\PropColorCabal*

\begin{algorithm}[ht]
    \caption{\ColoringCabals \label{alg:coloring-cabals}}

    \nonl\Input{A total coloring $\col_0$ of $V_H \setminus \Vcabal$}
    
    \nonl\Output{A total coloring of $\Vcabal$}

    \alg{ColorfulMatching} \label{line:cabal-matching} \hfill (\cref{sec:colorfulmatching})

    \alg{ColoringOutliers} \label{line:cabal-outliers}

    \computePutAside \label{line:cabal-compute-put-aside}

    \sct \label{line:cabal-sct}

    \multitrial \label{line:cabal-mct}

    \colorPutAside \label{line:cabal-put-aside} \hfill (\cref{sec:put-aside-sets})
\end{algorithm}

\paragraph{Reserved Colors.} 
Recall cabals are almost-cliques where $\tilde{e}_K \leq \lmin$. All cabals have the same number of reserved colors $r \eqdef r_K = 250 \lmin$ (see \cref{eq:reserved}).

\paragraph{Finding a Colorful Matching in Cabals (\cref{sec:colorfulmatching}).}
To color put-aside sets in \cref{line:cabal-put-aside}, we must have a colorful matching \emph{even when $a_K \leq O(\log n)$}.
We introduce a novel algorithm based on the fingerprinting techniques to compute a colorful matching in cabals where $a_K \in O(\log n)$; see \cref{sec:colorfulmatching} for more details. We run first the algorithm of \cref{lem:colorful-matching-high} and if it results in a matching of size $O(\log n)$, we cancel the coloring and run our new algorithm. This is the distributed algorithm that works in these extremely dense almost-cliques. We emphasize that we do not necessarily find a matching of size $\Omega(a_K/\eps)$. However, it suffices to find a matching of size $M_K$ such that $M_K \geq a_v$ for almost all nodes $v$ of $K$.

\begin{restatable}{proposition}{PropMatchingDenseCabals}
    \label{lem:satisfied-nodes}
    Assume $\Delta \gg \log^2 n$.
    Suppose all vertices in cabals $K$ with $a_K \in O(\log n)$ are uncolored.
    Let $\col$ be the coloring produced by \colorfulmatchingcabal in $O(\log^* n)$ rounds. With high probability, in each cabal $K$ such that $a_K \in O(\log n)$, for at least $(1-10\eps)\Delta$ vertices $v\in K$ the size of the colorful matching exceeds their anti-degrees $a_v \leq M_K \eqdef|K\cap\dom\col| - |\col(K)|$.
    Moreover, the algorithm does not use reserved colors. 
\end{restatable}

\paragraph{Inliers \& Outliers.}
In cabals, it suffices that inliers have external degree $O(e_K)$ because we create slack using put-aside sets. Formally, in each $K\in \Kcabal$, inliers are $I_K \eqdef \set{u\in K: \tilde{e}_v \leq 20\tilde{e}_K}$. The following lemma is clear from Markov inequality.

\begin{lemma}
    \label{lem:size-inliers-cabals}
    For $K\in \Kcabal$, the number of inliers $|I_K| \geq 0.9\Delta$.
\end{lemma}

In cabals, it suffices that for almost all nodes there are as many available colors in the clique palette as uncolored vertices in $K$. \cref{lem:clique-slack-cabal} shows that $a_v \leq M_K$ suffices for this to hold. Note that vertices cannot check if $a_v \leq M_K$. Part of the error in the approximation we used in non-cabals (\cref{eq:x}) was balanced by slack from slack generation, which we cannot run in cabals. Fortunately, in cabals, we will not need to detect when $a_v \leq M_K$.

\begin{lemma}
    \label{lem:clique-slack-cabal}
    For each $v\in K$ such that $K$ has $M_K$ repeated colors, we have that
    \[ 
    |L_\col(v) \cap L_\col(K)| \geq |(K\cup N(v)) \setminus\dom\col| + M_k - a_v\ . 
    \]
\end{lemma}

\begin{proof}
    The clique palette contains at least
    \[
        |L_\col(v) \cap L_\col(K)| \geq 
        \Delta+1 - |K \cap \dom\col| - |E_v \cap \dom\col| + M_K 
    \]
    colors.
    Then, using $|K \cap \dom\col| = |K| - |K \setminus \dom\col|$, $|E_v \cap \dom\col| = e_v - |E_v\setminus\dom\col|$ and $\Delta+1 = (\Delta - \deg(v)) + |K| + e_v - a_v$, this becomes the claimed inequality.
\end{proof}

\paragraph{Computing Put-Aside Sets.} 
Recall that the put-aside sets $P_K$ 
should have two properties: 
  i) they have size $r$ each (where $r$ is the number of reserved colors), and 
  ii) no two such sets have an edge between them.
The aim is to color $P_K$ only at the very end so that we can avoid using colors $\set{1, 2, \ldots, r}$ before calling \multitrial.
Contrary to previous work \cite{HKNT22}, we introduce the additional guarantee that each cabal contains only a few nodes adjacent to nodes in put-aside sets from other cabals, which will be necessary for coloring put-aside sets at the end.

\begin{restatable}{lemma}{LemConstructPutAside}
    \label{lem:compute-put-aside}
    Let $\col$ be a coloring such that $|I_K \setminus \dom\col| \geq 0.75\Delta$. 
    There is a $O(1)$-round algorithm (\cref{alg:compute-put-aside}) computing sets $P_K \subseteq I_K \setminus \dom\col$ such that, w.h.p., for each cabal $K$,
    \begin{enumerate}
        \item\label[part]{part:put-aside-size} $|P_K| = r$,
        \item\label[part]{part:put-aside-put-aside} there are no edges between $P_K$ and $P_{K'}$ for $K'\neq K$,
        \item\label[part]{part:put-aside-inliers} at most $|K|/100$ nodes of $K$ have neighbors in $\bigcup_{K'\in \Kcabal\setminus \set{K}} P_{K'}$.
    \end{enumerate}
\end{restatable}

As previous work \cite{HKNT22} guaranteed \cref{part:put-aside-size,part:put-aside-put-aside} and that \cref{part:put-aside-inliers} is a straightforward analysis, we defer the proof (and pseudo-code) to \cref{sec:omitted-proofs}.

\paragraph{Coloring Put Aside Sets (\cref{sec:put-aside-sets}).}
Only put-aside sets remain to color. The following lemma states it can be done in $O(1)$ rounds, thereby concluding the proof of \cref{prop:coloring-cabals}. 

\begin{restatable}{proposition}{PropColorPutAside}
    \label{prop:color-put-aside}
    Suppose $\col$ is a coloring such that only put-aside sets are uncolored and at least $0.9\Delta$ nodes in each cabal verify $a_v \leq M_K$. 
Then there is a $O(1)$-round algorithm that computes a total coloring of $H$, with high probability. 
\end{restatable}

Note that it \emph{does not} extend the coloring produced by earlier steps of the algorithm, but instead \emph{exchanges} colors with a few already colored inliers in $K$. 
Recoloring vertices must be done carefully, as it must happen in all cabals in parallel without creating monochromatic edges. As \alg{ColorPutAsideSets} is quite involved, we defer its description and analysis to \cref{sec:put-aside-sets}.

\begin{proof}[Proof of \cref{prop:coloring-cabals}]
    We go over the steps of \cref{alg:coloring-cabals}.

    \paragraph{Colorful Matching (\cref{line:cabal-matching}).}
    We say we found a sufficiently large colorful matching $M_K$ if all but $0.1\Delta$ nodes in $K$ have $a_v \leq M_K$. Observe that $M_K \geq \Omega(a_K/\eps)$ is sufficiently large because, by Markov, at most $O(\eps\Delta)$ can have $a_v > \Omega(a_K/\eps)$.
    Run in each cabal the colorful matching algorithm of \cref{lem:colorful-matching-high}. If $a_K \geq C\log n$ for some constant $C > 0$, we find a sufficiently large colorful matching with high probability. Assume $a_K \leq C\log n$. If the algorithm produces a matching of size $\Omega(C/\eps\cdot \log n) \geq \Omega(a_K/\eps)$, we found a large enough matching. Otherwise, nodes can compute $M_K$ using the query algorithm (\cref{lem:query}) and detect that the matching might be too small. We emphasize that vertices do not know $a_K$ but it suffices to compare $M_K$ to $\Omega(C/\eps \cdot \log n)$. In that case, all vertices of $K$ drop their colors and run the algorithm of \cref{lem:satisfied-nodes}. With high probability, we find a sufficiently large matching in those cabals. 
    As in \cref{prop:coloring-non-cabal}, if $M_K \geq 2\eps\Delta$, we can colors all vertices of $K$ in $O(\log^*n)$ rounds with high probability. We henceforth assume \emph{all} cabals contain a sufficiently large colorful matching such that $M_K \leq 2\eps\Delta$.
    
    \paragraph{Coloring Outliers (\cref{line:cabal-outliers}).}
    Recall that, in cabals, inliers are defined as nodes of low-external degree (verifying only $\tilde{e}_v \leq 20\tilde{e}_K$); hence, $|I_K| \geq 0.9\Delta$ (\cref{lem:size-inliers-cabals}). Therefore, as in \cref{prop:coloring-non-cabal}, w.h.p., all outliers get colored in $O(\log^*n)$ rounds and we henceforth focus on inliers.

    \paragraph{Computing Put-Aside Sets (\cref{line:cabal-compute-put-aside}).}
    By \cref{lem:compute-put-aside}, we find put-aside sets $P_K$ (which are all uncolored inliers) in $O(1)$ rounds with high probability.

    \paragraph{Synchronized Color Trial (\cref{line:cabal-sct}).}
    Let $S_K = K \setminus (P_K \cup \dom\col)$ be the uncolored inliers that are not put-aside nodes. Since $|I_K| \geq 0.8\Delta$ and $|P_K| = r \leq O(\lmin) \ll \Delta$, this set contains at least $|S_K| \geq 0.5\Delta$ vertices. On the other hand, by \cref{lem:clique-slack-cabal}, the clique palette contains enough colors: $|L_\col(K)| \geq |K\setminus\dom\col| = |S_K| + |P_K| = |S_K| + r$. Both conditions of the synchronized color trial are verified and afterward at most $50(1+\delta)\lmin \leq 52\lmin$ nodes in $S_K$ remain uncolored in each cabal, by \cref{lem:sct}, w.h.p.

    \paragraph{MultiColorTrial (\cref{line:cabal-mct}).}
    Let $H'$ be the graph induced by $\bigcup_{K\in\Kcabal} K \setminus (P_K \cup \dom\colsct)$, the uncolored vertices not in put-aside sets.
    The maximum degree \emph{in $H'$} is at most $74\lmin$ (at most $52\lmin$ uncolored neighbors in $K$ and at most $e_v\le 20\frac{1+\delta}{1-\delta}\lmin \leq 21\lmin$ external neighbors). 
    Since $\colsct(K) \cap [r] = \emptyset$, each $v$ loses reserved colors only because of external neighbors.
    Hence, the number of reserved colors available to a vertex $v$ in $H'$ is at least
    \[
        |[r] \cap L_{\col}(v)| \ge r - e_v \ge 3\cdot 75\lmin \geq 3\deg_\col(v; H') + 3\lmin  \ .
    \]
    By \cref{lem:mct}, \multitrial with $\calC(v) = [r]$ (recall $r \geq \Omega(\ell)$, \cref{eq:reserved}) colors all vertices of $H'$ in $O(\log^*n)$ rounds with high probability. Observe that all vertices know $r$, hence $\calC(v)$, because it is fixed in advance.
    
    \paragraph{Coloring Put-Aside Sets (\cref{line:cabal-put-aside}).}
    The only vertices left to color are put-aside sets. Recall that, in all cabals, we computed a sufficiently large colorful matching. By \cref{prop:color-put-aside}, they can be colored in $O(1)$ rounds. 
\end{proof}

\section{Fingerprinting, Approximate Counting \& Almost-Clique Decomposition}
\label{sec:fingerprints}
The plan for this section is as follows: in \cref{sec:geom-rv}, we analyze the behavior of maxima of geometric variables; in \cref{sec:fingerprint-encoding}, we give an efficient encoding scheme for sets of maxima of geometric variables; in \cref{sec:distributed-fingerprints}, we apply those results to derive an efficient approximate counting algorithm in cluster graphs; and finally, in \cref{sec:fingerprint-acd}, we use this technique to compute almost-clique decompositions fast.

\subsection{Maxima of Geometric Random Variables}
\label{sec:geom-rv}

We begin with a few simple properties of geometrically distributed random variables. For any $\lambda \in (0,1)$, we say $X$ is a geometric random variable of parameter $\lambda$ when
\[\text{for all } k \in \N_0,\quad \Pr[X = k] = \lambda^{k} - \lambda^{k+1}\ .\]
Note that $\Pr[X \geq k] = \lambda^{k}$. That is $X$ counts the the number of trials needed until the first success when each trial is independent and \emph{fails} with probability $\lambda$.

\begin{claim}
\label{claim:max-geom-prob}
    Let $d$ be an integer and $X_1,\ldots, X_d$ be independent geometric random variables of parameter $1/2$. Let $Y = \max_{i\in [d]}X_i$. For all $k\in \N_0$, we have
    \[
\Pr[Y < k] = (1-2^{-k})^d\ .
    \]
\end{claim}
\begin{proof}
    The event $\set{Y < k}$ is the intersection of the $d$ events $\set{X_i < k}$ with $i \in [d]$, which are independent and have probability $1-2^{-k}$.
\end{proof}

\newcommand{\kstar}{{K^{\star}}}

Standard analysis shows that the expected maximum over $d$ geometric random variables is about $\log_{1/\lambda} d$. There has been work on asymptotic behavior of such variables (e.g., \cite{E08,BO90}). \cref{lem:concentration-fingerprint} shows concentration of measure for maxima of independent geometric variables suited to our use. We will later use this phenomenon to approximate an unknown $d$ from the aggregated maxima.

\begin{lemma}[Concentration of values]
\label{lem:concentration-fingerprint}
    Consider $t, d\geq 1$ integers and $t \times d$ independent geometric random variables $(X_{i,j})_{i \in [t], j \in [d]}$ of parameter $1/2$. For each $i \in [t]$, let $Y_i = \max_{j \in [d]}(X_{i,j})$. For each integer $k$, let $Z_k = \card{\set{i \in [t]: Y_i < k }}$.
Let $\kstar = \min \set{k : Z_k \geq (27/40)t}$ and define:
    \[\hat{d} \eqdef \frac
    {\ln\parens*{Z_{\kstar} / t}}
    {\ln(1-2^{-\kstar})}
    \ .\]
    Then, for any $\xi \in (0,1)$, 
    \begin{equation}
        \label{eq:apx-fingerprint}
        \abs{ d - \hat{d} } \leq \xi d
        \quad\text{w.p.\ at least}\quad
        1-6\exp(-{\xi^{2}t}/{200}) \ .
    \end{equation}
\end{lemma}
\begin{proof}
    For each integer $k$, let $p_k \eqdef (1-2^{-k})^d$, such that  $p_k = \Pr[Y_i <k]$ for all $i \in [t]$, as shown in \cref{claim:max-geom-prob}.
Let us bound the value of $\kstar$ and $p_\kstar$. For any $x\in [0,1]$, we have $1-d\cdot x \leq (1-x)^d \leq e^{-d \cdot x}$. Applied to $p_k$,
    \begin{equation}
    \label{eq:pk-bound}
    1-d\cdot 2^{-k} 
    \quad\leq\quad p_k
    \quad\leq\quad \exp(-d \cdot 2^{-k})\ .
    \end{equation}

    Let us analyze $p_k$ when $k \approx \log d$. 
By definition, the probabilities $p_k$ are strictly increasing as a function of $k$. 
    We also have that $Z_{k+1} \geq Z_k$ structurally, since $\set{i : Y_i < k} \subseteq \set{i: Y_i < k+1}$. \cref{eq:pk-bound} implies that $p_k \geq 3/4$ for each $k \geq \log d + 2$, and $p_k \leq e^{-1/2} < 0.607$ for each $k \leq \log d + 1$. 
    We argue that $\kstar \in \set{\ceil{\log d}+1,\ceil{\log d} + 2}$. 
For each $k > 0$, we have $\Exp[Z_k] = p_k \cdot t$. 
    By the additive Chernoff Bound, \cref{lem:basicchernoff}, \cref{eq:chernoffadditive}, we have 
    \[ 
        \Pr[\abs{Z_k - p_k t} > (\xi/20) t] \leq 2\exp\parens*{ -\frac{\xi^2 t}{200} } \ .
    \] 
    Thus, by union bound, $\abs{Z_k - p_k t} \leq (\xi/20)t$ holds for all $k\in \set{\ceil{\log d}, \ceil{\log d}+1, \ceil{\log d}+2}$, 
    w.p.\ at least $1-6\exp(-(\xi^2/200)t)$.
    In particular, we have $Z_{\ceil{\log d}+2} \geq (3/4 - \xi/20) t > (27/40)t$ and thus $\kstar \leq \ceil{\log d} + 2$. 
    Additionally, it holds that $\kstar \ge \ceil{\log d}+1$ 
    because $Z_{\ceil{\log d}} \leq (e^{-1/2} + \xi/20)t < (27/40)t$.

    To summarize, the selected $\kstar$ belongs to $\set{ \ceil{\log d}+1,\ceil{\log d}+2 }$ and verifies that 
    \begin{equation}
        \label{eq:concentration-kstar}
        \abs{ Z_{\kstar}/t - p_{\kstar} } \leq (\xi/20) \ .
    \end{equation} 
    \Cref{eq:apx-fingerprint} follows from the following computations. Plugging \cref{eq:concentration-kstar} into the definition of $\hat{d}$, we obtain
    \begin{align*}
        \hat{d} \eqdef \frac{\ln(Z_{\kstar}/t)}{\ln(1-2^{-\kstar})} \in
        \frac{\ln(p_{\kstar} \pm \frac{\xi}{20})}{\ln(1-2^{-\kstar})} 
        &=
        \frac{\ln(p_{\kstar}) + \ln(1 \pm \frac{\xi}{20p_{\kstar}})}{\ln(1-2^{-\kstar})} \\
        &= 
        d + \frac{\ln(1\pm\frac{\xi}{20p_{\kstar}})}{\ln(1-2^{-\kstar})}
    \end{align*}
    where the last equality follows from  $p_{\kstar} = \exp\parens*{ d \cdot \ln(1-2^{-\kstar}) }$. 
    
    Recall that $x/(1+x) \leq \ln(1+x) \leq x$ for all $x\in (-1,1)$, and that the three functions are increasing over this interval.
    Let us bound the denominator of the error term first. 

    Since $\kstar \leq \ceil{\log d} + 2$, we have that $2^{-\kstar} \geq 1/(8d)$, and so
    \begin{align*}
    \ln\parens*{ 1 - 2^{-\kstar} }
    & \leq \ln\parens*{ 1 - \frac{1}{8d} }
    \leq -\frac{1}{8d} \ .
    \end{align*}

    Let us now bound the numerator, which is equal to $\ln(1+x)$ for some value $x \in [-\xi/(20p_{\kstar}),\xi/(20p_{\kstar})]$.
    Remark that $\xi/(20p_{\kstar}) \leq \xi/10$ since
    $p_{\kstar} \geq 1/2$ (by \cref{eq:pk-bound}).
As $\ln(1+x)$ is increasing over the interval $(-1,1)$ and the bound of $\ln(1+x) \leq x$, we get an upper bound for the numerator of 
    \[
    \ln\parens*{ 1 + \frac{\xi}{10 p_{\kstar}} }
    \leq \xi/(20 p_{\kstar})
    \leq \xi/10 \ .\]
    On the other hand, that $\ln(1+x)$ is increasing over the interval $(-1,1)$ and the bound $\ln(1+x) \geq x / (1+x)$ lowers bound the numerator by
    \[
    \ln\parens*{ 1 - \frac{\xi}{20 p_{\kstar}} }
    \geq \ln\parens*{ 1 - \frac{\xi}{10} }
\geq - \frac{\xi}{10} / \parens*{1 - \frac{\xi}{10}} \geq - \frac{\xi}{9} \ .
    \] Together, these inequalities yield \cref{eq:apx-fingerprint}.
\end{proof}

\begin{lemma}[Unique maximum]
\label{lem:unique-maximum}
    Consider an integer $d\geq 2$ and $d$ independent geometric random variables $(X_{j})_{j \in [d]}$ of parameter $\lambda < 1$. Let $Y = \max_{j \in [d]} X_{j}$. Then there exist $i\neq j\in[d]$ 
    such that $X_i = X_j = Y$ with probability at most $\frac{(1-\lambda)^2}{1-\lambda^2}$.
\end{lemma}

\begin{proof}
    Call $\calE$ the event that the max is not unique. Then, for each pair $i < j$, let
    \[
        \calF_{i,j} = 
        \set{\forall k \in [d], X_i \geq X_k} 
        \cap
        \set{\forall k \in [d]\setminus\set{i}, X_j \geq X_k}
        \cap
        \set{\forall k \in [j-1] \setminus \set{i}, X_j > X_k} 
    \]
    be the event that the $X_i$ is a maximum, $X_j$ is a maximum in $X_{\neg i}$ and all $X_{<j-1}$ (except $X_i$) are \emph{stricly less} than $X_j( \leq X_i)$. Intuitively, when $\calF_{i,j}$ occurs, $X_i$ is the $1^{st}$ maximum and $X_j$ the $2^{nd}$ one when we order first according to $X$'s and then to indices. 

    Note that $\calF_{i,j}$ are pairwise disjoint events and that $\calE \subseteq \bigcup_{i<j} \calF_{i,j}$. We emphasize however that events $\calF_{i,j}$ do not cover the whole universe (meaning the set $[d] \times \mathbb{N}$ of all outcomes), as these events preclude that the second largest value occurs before the largest one. By disjoint union and Bayes' rule
    \[
        \Pr(\calE) = 
        \sum_{m\geq 0}\sum_{i < j} \Pr(\calE, \calF_{i,j}, Y=m) = 
        \sum_{m\geq 0} \sum_{i < j} \Pr(\calE, Y=m \mid \calF_{i,j})\Pr(\calF_{i,j}) \ .
    \]
    We bound this conditional probability through the following observation. Since $X_i$ and $X_j$ are the largest values, the maximum is not unique iff both are equal to $m$. So,
    \[
        \Pr(\calE, Y=m \mid \calF_{i,j})
        \leq
        \lambda^{2m}(1-\lambda)^2 \ ,
    \]
    and since $\sum_{i<j} \Pr(\calF_{i,j}) \leq 1$, we get
    \[
        \Pr(\calE) =
        \sum_{m\geq 0} \lambda^{2m}(1-\lambda)^2 \parens*{ \sum_{i < j} \Pr(\calF_{i,j}) } \leq 
        \sum_{m\geq 0} \lambda^{2m}(1-\lambda)^2  =
        \frac{(1-\lambda)^2}{1-\lambda^2} \ .
    \]
    which concludes the proof.
\end{proof}

In particular, with a set of geometric random variables of parameter $1/2$, their maximum occurs uniquely with probability at least $2/3$, regardless of the number of random variables. Also, note that the distribution of where the unique maximum occurs is the uniform distribution over the $d$ trials:

\begin{lemma}
    \label{prop:uniform-maximum}
    Let $d$ be a positive integer and $(X_{j})_{j \in [d]}$ be a family of independent geometric random variables with the same parameter $\lambda \in (0,1)$. Let $Y = \max_{j \in [d]}X_{j}$. Then:
    \begin{align}
\forall i \in [d], \qquad \Pr[X_i = Y \mid \exists ! j: X_j = Y] &= \frac 1 d\ .\label{eq:uniform-maximum}
    \end{align}
\end{lemma}
\begin{proof}
    Since variables $X_1, \ldots, X_d$ are i.i.d., permuting their order does not change their joint distribution. Let $\calE$ be the event that the maximum is unique. Observe that $Y$ and $\calE$ are invariant under permutations of $X_i$'s. Hence, then $\Pr(\calE \wedge X_i = Y) = \Pr(\calE \wedge X_j = Y)$ for any pair $i,j \in [d]$. By Bayes' rule, we infer \cref{eq:uniform-maximum}.
\end{proof}

\subsection{Efficient Encoding of Maxima}
\label{sec:fingerprint-encoding}

In this section, we show that a set of $t$ maxima of $d \leq n$ independent geometric random variables of parameter $1/2$ can be encoded in $\Theta(t + \log\log d)$-bits, with high probability.
As we compute such sets of random variables to estimate neighborhood similarities, this ensures our algorithms are bandwidth-efficient.

The intuition for this result is that while each maxima can reach values as high as $\Theta(\log(d) + \log(n))$, requiring $\Theta(\log\log(d) + \log\log(n))$ bits to encode on their own, the values taken together are mostly concentrated around $\Theta(\log(d))$. This high concentration allows for more efficient encoding, by only storing a number $k \approx \log d$ around which the values are concentrated, along with the deviations from $k$.

\begin{lemma}
\label{lem:concentrated-maximums}
    Consider a set of $t \times d$ independent geometric random variables $(X_{i,j})_{i \in [t],j\in[d]}$ of parameter $1/2$ and their associated maxima, the $t$ random variables $Y_i = \max_{j \in [d]} X_{i,j}$ for each $i\in[t]$. We have
    \[
    \sum_{i=1}^t \abs*{Y_i - \ceil{\log d}} \leq 8t \quad w.p.\quad 1-2^{-t/10+1}\ .
    \]
\end{lemma}
\begin{proof}
    Let 
$k = \ceil{\log d}$. For each $i \in [t]$, let $Y_i^+ = \max( 0 , Y_i - k)$ and $Y_i^- = \max( 0 , k - Y_i)$, such that $\abs{Y_i - k} = Y_i^+ + Y_i^-$. Let $Y^+ = \sum_{i=1}^t Y_i^+$ and $Y^- = \sum_{i=1}^t Y_i^-$. We prove the lemma by showing that $Y_i^+$ and $Y_i^-$ do not exceed $O(t)$ w.p.\ $1-\exp(-\Omega(t))$.

    Using that $1-yd \le (1-y)^d$ for any $y \ge 0$, it holds for any $x \ge 0$ that 
    \[ \Pr[Y_i^+ \geq x] = \Pr[Y_i \ge x+k] = 1 - (1-2^{-x-k})^d
\le d \cdot 2^{-x-k} \le 2^{-x}\ . \]
If the sum $Y^+$ reaches $4t$ or more, then there are numbers $x_1,\ldots,x_t$ s.t.\ $\sum_{i=1}^t x_i =4t$ and for each $i \in [t]$, $Y^+_i \geq x_i$. For a fixed combination of $x_i$, since $Y_i$'s are independent, the probability that this occurs is upper bounded by
    \[
    \prod_{i=1}^t \Pr[Y_i^+ \geq x_i]
    \leq \prod_{i=1}^t 2^{-x_i}
    = 2^{-\sum_{i=1}^t x_i}
    \leq 2^{-4t}\ .
    \]
    There are $\binom{4t + t}{t-1} \le \parens*{\frac{5t \cdot e}{t}}^{t} = 2^{t\log(5e)}$ ways 
    to choose numbers $x_1, x_2, \ldots, x_t \ge 0$ such that they sum to $4t$.
    Thus, the probability that $Y^+$ reaches $4t$ is bounded as follows: 
    \[
    \Pr[Y^+ \geq 4t]
    \leq \binom{5t}{t-1} \cdot 2^{-4t}
    \leq 2^{-(4 - \log(5e)) t}
    = 2^{- t/10 }
    \ .\]

    The bound on $Y^-$ is obtained in the same way, using that $\Pr(Y^-_i \geq x_i) = \Pr(Y_i \leq k - x_i + 1) = p_{k -x_i + 1} \leq \exp(-2^{x_i-2}) \leq 2^{-x_i/2}$. Hence $\Pr(Y^- \geq 8t) \leq 2^{-t/10}$.
\end{proof}

\begin{lemma}
    \label{lem:fingerprint-encoding}
    Let $Y_i = \max_{j\in[d]} X_{i,j}$ where $(X_{i,j})_{i\in[t], j\in[d]}$ are independent geometric variables of parameter $1/2$.
    With probability $1-2^{-t/10+1}$,  the sequence of values $(Y_i)_{i\in[t]}$ 
can be described in $O(t + \log \log d)$ bits.
\end{lemma}
\begin{proof}
    To encode the set of maxima efficiently, compute an integer $k \in O(\log d)$ such that 
$\sum_{i=1}^t \abs*{Y_i - k} \leq O(t)$. 
    The existence of such a $k$ is guaranteed with probability $1-2^{-t/10+1}$ by \cref{lem:concentrated-maximums}. We can take the minimal one, or the one that in general minimizes the size of our encoding. Writing the binary representation of $k$ in the encoding takes $O(\log \log d)$ bits.
    
    To encode the values $(Y_i)_{i \in [t]}$, we write $\abs{Y_i - k}$ in unary, prefix it by the sign bit $\sign(Y_i-k)$, and use 0 as a separator. 
In total, the encoding takes $O(\log\log d) + \sum_{i\in [t]} (|Y_i - k|+2) \leq O(t + \log \log d)$ bits.
\end{proof}

\subsection{Approximate Counting from Fingerprinting}
\label{sec:distributed-fingerprints}

In \cref{lem:fingerprint}, each vertex $v$ approximates the number of $u\in N(v)$ such that $P_v(u)=1$, for a binary predicate $P_v$. When $P_v$ is the trivial predicate (i.e., $P_v(u)=1$ for all $u\in N(v)$), the algorithm approximates $|N(v)|$. Other examples of predicates we use are ``neighbors outside of $K_v$'' (when approximating external degrees) or ``neighbors $u\in N(v)$ colored with $\col(u) > r_v$'' (when estimating palette sizes in \cref{claim:apx-z}). Importantly, $P_v$ must be known to the machines of a cluster $V(v)$, as well as being efficiently computable by them. In the case of predicates related to the ACD, it suffices that each vertex informs its cluster of the ID of its almost-clique.

\begin{restatable}{lemma}{LemFingerprint}
    \label{lem:fingerprint}
    Let $\xi \in (0,1)$ and, for each $v\in V_H$, predicates $P_v: N(v) \to \set{0,1}$ such that if $P_v(u)=1$, there exists $w\in V(v) \cap V(u)$ that knows it. There is a $O(\xi^{-2})$-round algorithm for all nodes to estimate $\card{N_H(v) \cap P_v^{-1}(1)}$ with high probability within a multiplicative factor $(1\pm \xi)$.
\end{restatable}
\begin{proof}
    Each vertex $v$ samples $t=\Theta(\xi^{-2}\log n)$ independent geometric variables 
    $X_{v,1}, \ldots, X_{v,t}$ of parameter $1/2$ and broadcasts them within its support tree $T(v)$ and along edges to neighbors.
    Through aggregation, each vertex $v$ computes
    $Y_{v,i} = \max\set{X_{u,i}: u \in N(v)\text{ and } P_v(u)=1}$ for each $i \in [t]$.
    Let $d=|N_H(v) \cap P_v^{-1}(1)|$. 
    By \cref{lem:concentration-fingerprint}, w.h.p., each vertex $v$ deduces from $\set{Y_{v,i}}_{i\in[t]}$ an estimate $\hat{d}_v \in (1\pm\xi)d$.

    We now argue that the maxima $\set{Y_{v,i}}_{i\in[t]}$ can be aggregated efficiently.
    Let $Z_{w,v,i} = \max\set{X_{u,i}: \exists e \in E_H(u,v), P_v(u) = 1 \text{ and } w = m(e)}$ 
    be the geometric variables that machine $w \in V(v)$ receives from neighbors of $v$. 
    We aggregate variables $Z_{w,v,i}$ by order of depth in the support trees. 
    More precisely, we have $\dilation$ phases of $O(\xi^{-2})$ rounds each, such that at the end of phase $i \in [\dilation]$, each node $w$ at depth $\dilation - i$ in $T(v)$ has computed the coordinate-wise maximum of the variables $(Z_{w,v,j})_{j \in [t]}$ from its subtree. Note that each partially aggregated set $\set{Z_{w,v,j}}$ is a set of $t$ maxima of independent geometric variables. Thus, to send the maxima to their parents in the support tree, vertices use the encoding scheme in \cref{lem:fingerprint-encoding}. In total, $\dilation n\leq n^2$ aggregates are computed, each one a vector of $\Theta(\xi^{-2} \log n)$ random variables fixed in advance. By union bound over all those aggregates, w.h.p., the root of $T(v)$ learns $Y_{v,i}$ for each $i\in[t]$ in $O(\xi^{-2})$.
\end{proof}

\subsection{Almost-Clique Decomposition from Fingerprinting}
\label{sec:fingerprint-acd}\label{sec:ACD}

To compute an $\eps$-almost-clique decomposition (\cref{def:ACD} for $\eps$ as in \cref{eq:params}), we need to detect $\Theta(\eps)$-friendly edges. For $\xi\in(0,1)$, we say that the edge $\set{u,v}$ is $\xi$-friendly when $|N(u) \cap N(v)| \geq (1-\xi)\Delta$. Authors of \cite{ACK19} showed that it suffices to tell apart very friendly edges from those that are not friendly enough. More precisely, we need to solve the $\xi$-\emph{buddy} predicate: 
\begin{itemize}
    \item if $|N(v) \cap N(u)| \geq (1-\xi)\Delta$, the algorithm answer $\mathsf{Yes}$; and
    \item if $|N(v) \cap N(v)| < (1-2\xi)\Delta$, the algorithm answer $\mathsf{No}$.
\end{itemize}
The algorithm can answer arbitrarily when the edge is in neither cases. Using the fingerprinting technique to approximate degrees and the size of joint neighborhoods $|N(u) \cup N(v)|$, we solve the buddy predicate.

\begin{lemma}
    \label{lem:buddy}
    There exists a $O(\xi^{-2})$ round algorithm such that, w.h.p., for each $w_u\in V(u), w_v\in V(v)$ s.t. $w_uw_v \in E_G$, $w_u$ and $w_v$ solve the $\xi$-buddy predicate.
\end{lemma}

\begin{proof}
    Let $\xi' \eqdef 2\xi/c < \xi$ for some constant $c>2$ to be fixed later. To solve the $\xi$-buddy predicate, it suffices to discriminate between $\xi'$-friendly edges from those that are not $c\xi'$-friendly.
    By \cref{lem:fingerprint} vertices approximate their degrees $\hat{d}(v) \in (1 \pm 0.5\xi')\deg(v)$ using the fingerprinting in $O(\xi'^{-2}) = O(\xi^{-2})$ rounds. All vertices with $\hat{d}(v) < (1-1.5\xi')\Delta$ have machines in their cluster answer $\mathsf{No}$ for all their edges. Such vertices have degree $< (1-\xi')\Delta$, hence cannot have any incident $\xi$-friendly edges. We henceforth assume that all vertices $v$ have a large degree $\deg(v) \geq (1-2\xi')\Delta$.

    Then, each vertex samples independent geometric random variables $X_{v, 1}, \ldots, X_{v,t}$ for $t = \Theta(\xi^{-2}\log n)$. By \cref{lem:fingerprint-encoding}, vectors $Y_{v,i} = \max\set{X_{u', i},~u\in N(v)}$ are disseminated to all machines in $V(v)$ in $O(\xi^{-2})$ rounds.
    On each edge $w_uw_v$ connecting two clusters of $u$ and $v$, $w_u$ and $w_v$ exchange vectors $Y_{u,i}$ and $Y_{v,i}$ and compute $Y_{uv, i} = \max\set{Y_{v,i}, Y_{u,i}} = \max\set{X_{u', i},~u'\in N(u) \cup N(v)}$ for each $i\in [t]$. By \cref{lem:concentration-fingerprint}, w.h.p., machines $w_u$ and $w_v$ approximate $F \in |N(v) \cup N(u)| \pm 0.5\xi'\Delta$, the size of their clusters' joint neighborhoods. If $\set{u,v}$ is a $\xi'$-friendly edge, then
    \[
        F \leq |N(u) \cup N(v)| + 0.5\xi'\Delta \leq (1+1.5\xi')\Delta
    \]
    Otherwise, when the edge is not $c\xi$-friendly, because remaining vertices have a high degree, the symmetric difference is large
    \[ 
        |N(v) \setminus N(u)| = |N(v)| - |N(u)\cap N(v)| > (1-2\xi')\Delta - (1 - c\xi')\Delta = (c - 2)\xi'\Delta \ ,
    \] 
    which in turn implies the joint neighborhood of $u$ and $v$ must be large
    \begin{align*}
        F &= |N(u) \cup N(v)| - 0.5\xi'\Delta \geq 
        |N(u)| + |N(v) \setminus N(u)| - 0.5\xi'\Delta \\ &\geq
        (1-2.5\xi')\Delta + (c - 2)\xi'\Delta \geq
        (1 + (c - 4.5)\xi')\Delta \ .
    \end{align*}
    Hence, for $c \geq 6$, we are able to tell apart $\xi'$-friendly edges from the not $c\xi'$-friendly. 
\end{proof}

We can now compute the almost-clique decomposition.

\begin{proof}[Proof of \cref{prop:ACD}]
    Let $\xi = \Theta(\eps)$. 
    By \cref{lem:buddy}, w.h.p., each machine responsible for an edge computes the $\delta$-buddy predicate in $O(\xi^{-2}) = O(\eps^{-2})$ rounds. An edge for which the predicate answer is $\mathsf{Yes}$ is called a buddy edge. Then using the fingerprinting algorithm of \cref{lem:fingerprint} with the predicate $P_v(u)$ ``edge $\set{u,v}$ is a buddy edge'', we can discriminate between vertices incident to $\geq (1-\xi)\Delta$ buddy edges from those incident $<(1-2\xi)\Delta$ buddy edges.
    By \cite[Lemma 4.8]{ACK19}, the $\eps$-almost-clique of the decomposition are connected components in the graph where we retain only buddy edges. Since they have diameter two, a $O(1)$ round BFS suffices to elect a leader in each almost-clique and propagate its identifier. This is efficiently implementable in our settings, as the BFS are performed on node-disjoint subgraphs (\cref{fact:bfs}). By comparing the ID of their almost-cliques' leaders, machines with edges to other clusters can check if they connect vertices from the same almost-clique.
\end{proof}

\section{Colorful Matching in Densest Cabals}
\label{sec:colorfulmatching}
Recall that a colorful matching is a partial coloring that uses each color twice (within a given almost-clique $K$). We wish to find a set of anti-edges and to same-color the nodes of each pair.

There are two regimes for computing a colorful matching: when $a_K$ is $\Omega(\log n)$ and when $a_K$ is $O(\log n)$.
Computing a colorful matching in almost-cliques of average anti-degree $a_K \in \Omega(\log n)$ can be done following a sampling technique from prior work \cite{FGHKN24}. 
In non-cabals, this suffices because vertices also get slack from slack generation. 
In cabals, we begin by running the algorithm of \cite{FGHKN24}, and if it returns too small a colorful matching $M_K \leq O(\eps^{-1}\log n)$, we cancel the coloring in $K$ to run the new algorithm presented in this section. (See proof of \cref{prop:coloring-cabals} in \cref{sec:cabals} for more details.)
This section focuses exclusively on \emph{low anti-degree cabals}, with $a_K \in O(\log n)$. Its main algorithm  (\colorfulmatchingcabal) never produces a colorful matching with more than $\Theta(\log n)$ edges.

\PropMatchingDenseCabals*

Throughout this section, we denote by $C$ the constant such that $a_K \leq C\log n$, as assumed in \cref{lem:satisfied-nodes}.

The heart of the algorithm is to find a large matching of anti-edges in $K$. Using basic routing and \multitrial, they can then be colored in $O(\log^* n)$ rounds. Contrary to \cite{FGHKN24}, we do not find $\Omega(a_K/\eps)$ anti-edges, but rather enough anti-edges to ``satisfy almost all nodes''.
To make this formal, let us introduce some quantities and notations:
\begin{itemize}
    \item let $\tau \eqdef 4\eps$ be the constant fraction at which we divide nodes for the analysis,
    \item let $\hat{a}_K$ be the largest number s.t.\ at least $\tau \card{K}$ nodes in $K$ have anti-degree $\hat{a}_K$ or more, 
    \item call a vertex \emphdef{low} when $a_v < \hat{a}_K$ and \emphdef{high} otherwise.
\end{itemize}
We split vertices according to their position in 
the distribution of anti-degrees within the almost-clique:
high vertices represent the top $\tau$-fraction of anti-degrees in the cabal, and low vertices the bottom $(1-\tau)$-fraction. 
As anti-degrees are non-negative, $\hat{a}_K$ cannot be much larger than the mean.

\begin{fact}
    \label{fact:med-avg}
    $\hat{a}_K \le \frac{1}{\tau} a_K \le (C/\tau)\log n $.
\end{fact}
\begin{proof}
    $C\log n \geq a_K = \frac{1}{\card{K}}\sum_{v \in K} a_v
         \geq \frac{1}{\card{K}}\sum_{v \in K: a_v \geq \hat{a}_K} a_v
        \geq \tau\cdot \hat{a}_K$.
\end{proof}

The bulk of the analysis is the following lemma, which we prove in \cref{sec:fingerprintmatching}. It states that we can find (at least) $\hat{a}_K$ anti-edge in each cabal where we run the algorithm. 

\begin{restatable}{lemma}{LemFingerprintMatching}
    \label{lem:fingerprint-matching}
    Assume $\Delta \gg \eps^{-3}\log n$. With high probability, \cref{alg:fingerprint-matching} outputs a matching of $\tau \hat{a}_K / (4\eps)$ anti-edges.
\end{restatable}

Before 
proving \cref{lem:fingerprint-matching}, we show in the next section that it
implies \cref{lem:satisfied-nodes}.

\subsection{Coloring the Matching}
\label{sec:matchingcoloring}

We show that \cref{alg:colorfulmatching-cabal} has the effects claimed in \cref{lem:satisfied-nodes} and can be implemented as laid out here.

\begin{algorithm}
    \caption{\colorfulmatchingcabal}
    \label{alg:colorfulmatching-cabal}
    
    \nonl\Input{A cabal $K$ of low anti-degree $a_K \leq C\log n$ for some (large) constant $C > 0$.}

    \nonl\Output{Cabal $K$ has $2 M_K \geq 2\hat{a}_K$ of its nodes properly colored by $M_K$ non-reserved colors.}
    
    Apply \alg{FingerprintMatching}, let $F_K$ be the resulting matching and $M_K = \card{F_K}$.
    
    Each $v\in K$ joins a random group $j \in [M_K]$. 
    
    \nonl Group $j$ assists the $j\supth$ non-edge in $F_K$ with performing \multitrial.

    Each $(u,v) \in F_K$ runs \multitrial\ with color space $\calC = [\Delta+1] \setminus [300\eps\Delta]$.
\end{algorithm}

\begin{proof}[Proof of \cref{lem:satisfied-nodes}]
    Let $K_1, K_2, \ldots, K_k$ be the cabals where we run the algorithm.
    The algorithm has two phases: first, \alg{FingerprintMatching} (\cref{alg:fingerprint-matching}) finds an anti-matching $F_i$ of $M_i$ anti-edges in each $K_i$; second, we colors them.
    By \cref{lem:fingerprint-matching-implementation,lem:fingerprint-matching}, in $O(1/\eps)$ rounds, we find an anti-matching of size $M_i = \tau\hat{a}_K/(4\eps) = \hat{a}_K$ anti-edges in each $K_i$ with high probability. For each $K$, at least $(1-\tau)|K| \geq (1-10\eps)\Delta$ vertices $v\in K$ have $a_v \leq \hat{a}_{K_i} = M_i$, by definition of $\hat{a}_{K_i}$ and choice of $\tau = 4\eps$. 

    It remains to explain how we color the anti-edges. We split each cabal $K_i$ into $M_i$ random groups. Recall that $\Delta \gg\log^2 n$: let us assume that $\Delta \geq (C^2/\tau)\log^2 n$, i.e., $\Delta$ dominates $\log^2 n$ by a sufficiently large constant factor $(C^2/\tau)$. This implies that $\Delta \geq \hat{a}_K \cdot C\log n$, and $M_i = \hat{a}_K$ for each cabal $K_i$. By \cref{fact:random-groups}, w.h.p., both endpoints of the $j\supth$ anti-edge are adjacent to the $j\supth$ group, and the $j\supth$ group has diameter $2$. This enables efficient communication and coordination between the endpoints of each anti-edge.

    Consider \multitrial. The procedure relies on a \trymulticolor (\cref{alg:try-multi-color}) subroutine for randomly sampling increasing numbers of colors and testing whether they can be adopted. This is done by having nodes pick pseudorandom sets of colors, which only take $O(\log n)$ bits to describe even as they contain up to $\Theta(\log n)$ colors.
    This is easily adapted to our setting where it is anti-edges that are trying multiple colors. For each anti-edge, for every call to \trymulticolor, we let the endpoint of highest ID pick the (pseudo)random set of colors to try. This set is then sent to the other endpoint of the anti-edge using the anti-edge's designated random group. From there, all anti-edge endpoints try the colors from their sets: each endpoint learns which colors from its set are already taken by a neighbor, or are tried in this round by neighbor. Each anti-edge then uses its group to see which colors are available at its two endpoints, and colors itself with such a color if it exists. Thus all operations of \multitrial can be performed by a cabal's discovered anti-edges.

    Now, let $\calC = [\Delta+1] \setminus [300\eps\Delta]$. Note that $\calC$ excludes reserved colors as $r_K \leq 300\eps\Delta$ (\cref{eq:reserved}), so the colorful matching will not use any reserved color.
    Each anti-edge is adjacent to $\leq 2\eps\Delta + O(\log n) \leq 3\eps\Delta$ other anti-edges in $\bigcup_i F_i$; on the other hand, they have at least $\Delta+1 - 300 \eps\Delta - 2\eps\Delta \geq 0.5\Delta$ available colors. Hence, they have slack $0.1\Delta$, and they get colored by $O(1)$ rounds of \trycolor (\cref{lem:try-color}) plus $O(\log^* n)$ of \multitrial (\cref{lem:mct}).
\end{proof}

Using random groups for communication between the endpoints of the anti-edges strongly relies on the assumption that $\Delta \gg \log^2 n$.
Later in the paper, when dealing with the low-degree setting, we instead assign each anti-edge a single node in the clique as dedicated relay. 
We do so by computing a matching between anti-edges and potential relays, which increases the runtime of the procedure to $\poly \log \log n$.
Modifications to \cref{alg:colorfulmatching-cabal} in the low-degree setting are detailed in \cref{sec:lowdeg-inliers}, in particular in
\cref{lem:antiedge-relays}.

\subsection{Finding the Matching}
\label{sec:fingerprintmatching}

The suitable matching is found using fingerprinting (\cref{sec:fingerprints}).
Each node $v$ in the almost-clique samples $k=\Theta(\log n)$ independent geometric random variables $(X_{v,1}, X_{v,2}, \ldots, X_{v,k})$. For each $i\in [k]$, nodes compute the point-wise maximum of the random variables in their neighborhood $Y_i^v = \max\set{X_{u,i}, u\in N(v)}$, and the almost-clique computes $Y_i^K = \max\set{X_{u,i}: u \in K}$ the point-wise maximum of all the random variables in contains.

Let us focus on the first random variable sampled by each node in $K$, and the relevant maxima. Suppose the maximum value of the $\card{K}$ random variables is unique and let $v \in K$ be the unique node in $K$ at which it occurs. All neighbors of $v$ have the same maximum value in their in-clique neighborhood as in the whole almost-clique. Meanwhile, any anti-neighbor of $v$ has a different value for the two, so they all learn that they have an anti-edge with the node that sampled the unique maximum. 
This observation is the main idea behind how \cref{alg:fingerprint-matching} finds anti-edges.

\begin{algorithm}
\caption{\alg{FingerprintMatching}}
\label{alg:fingerprint-matching}
    \nonl\Input{A cabal $K$ of low anti-degree $a_K \leq C\log n$ for some (large) constant $C > 0$.}

    \nonl\Output{A matching of size $M_K \in \Omega(\hat{a}_K/\eps) \cap O(\log n)$ in the complement of $K$.}

Let $k = \frac{6C \log n}{\eps\tau} \geq 6 \hat{a}_K/\eps$  \hfill(gives a success probability $\geq 1 - n^{-\Omega(C)}$)

Each $v$ computes fingerprints $(Y^v_i)_{i \in [k]}$ and $(Y^{K_v}_i)_{i \in [k]}$ of $N_H(v)\cap K_v$ and $K_v$.\label{step:matching-fingerprint}

By BFS in $K$, give each node $v \in K$ with a local identifier $\ID_{v} \in \set{1,\dots,\card{K}}$. \label{step:matching-renaming}

{By BFS in $K$, identify the set $I$ of indices $i\in [k]$:\label{step:matching-first-filter}
\begin{itemize}
    \item the maximum in $K$ is reached at a unique vertex $u_i$ ($\exists ! u_i \in K, X_{u_i,i} = Y^K_{i}$),
    \item a non-edge $\set{u_i,v}$ incident to $u_i$ was detected ($\exists v \in K, Y^v_i \neq Y^K_i$),
    \item the maximum-value vertex $u_i$ was not a unique maximum in a previous trial ($\forall j < i, \exists u_j \in K \setminus \set{u_i}, X_{u_i,j} \leq X_{u_j,j}$).
\end{itemize}}

\nonl For each $i \in I \subseteq [k]$, let $A_i = \set{v \in K: Y^v_i \neq Y^K_i}$, and $u_i$ the one node s.t.\ $X_{u_i,i} = Y^K_{i}$.

Each $v\in K$ samples $i_v \in [k]$ uniformly at random and joins the $i_v\supth$ random group. \label{step:matching-random-group}

\nonl Group $i$ does computation and communication regarding trial $i \in [k]$.

Each node $v \in K$ informs its neighbors whether $v \in A_i$ for each $i \in I \subseteq [k]$.\label{step:matching-announce-group}

Group $i$ picks a random $(1/2, |K|)$-min-wise independent hash function $h_i$.\label{step:matching-min-pick}

Group $i$ 
computes the minimum $h_i$-hash value of nodes in $A_i$, $\min_{w \in A_i} h_i(\ID_w)$.\label{step:matching-min-compute}

\nonl
Let $w_i\in A_i$ be the node of minimum $\ID$ s.t.\ $h_i(\ID_{w_i}) = \min_{w \in A_i} h_i(\ID_w)$.

Group $i$ learns and broadcasts $\ID_{w_i}$ to nodes in $A_i$ . \label{step:matching-sampled-anti}

Discard trials $i$ s.t.\ $\exists j: u_i = w_j$. \label{step:selection-trumps-maximum}

For each $w$ s.t.\ $\card{\set{i: w=w_i}}>1$, $w$ picks an arbitrary such trial $i$, discards others. \label{step:matching-pick-trial}

\nonl Let $I' \subseteq I$ be the set of remaining trial indices.

Output as matching the anti-edges $(u_iw_i)_{i \in I'}$.\label{step:matching-output}
\end{algorithm}

Let us refer to the $\Theta(\log n)$ independent cells of our fingerprints as trials. There is a constant probability that the maximum value in each trial is unique (\cref{lem:unique-maximum}), and since all vertices have the same distribution, independent from each other, the unique maximum occurs at a uniform vertex in $K$ (\cref{prop:uniform-maximum}). Notably, it occurs at a high vertex with probability $\tau$. Each trial with a unique maximum is an opportunity to sample a new anti-edge between the unique maximum and its anti-neighbors. 
As earlier trials have already found anti-edges, later trials run the risk of sampling anti-edges sharing endpoint(s) with a previously sampled anti-edge. We cannot add such anti-edges to our set of anti-edges since it would not result in a matching.
We show that until enough anti-edges have been discovered, each trial has a probability $\Omega(\tau)$ of discovering an anti-edge with both endpoints unmatched in ealier trials.
Key to our argument is that until enough anti-edges have been sampled, a majority of the high nodes must remain non-sampled and have a majority of their anti-neighbors non-sampled.

In each trial, we need the unique maximum and its anti-neighbors to communicate. We use random groups for this, aided by the fact that $\Delta \gg k \cdot \log n \in \Theta(\eps^{-2} \log^2 n)$. One subtlety is how we sample the anti-neighbor $w_i$ from the set of anti-neighbors $A_i$ for each trial $i\in[k]$.
Each node is possibly the anti-neighbor of multiple unique maxima, i.e., a node can be part of many trials as a samplable anti-neighbor.
The random group for the $i$-th trial performs the sampling of $w_i$ by using a min-wise hash function (\cref{def:min-wise}). Min-wise hash functions have the property that, when applying a random such function to a set, which element of the set has the minimum hash value roughly follows a uniform distribution over the set.

First, we argue that the algorithm has the claimed runtime. As we use $k=\Theta(\log n)$ random groups, we need that $\Delta \gg \eps^{-2}\log^2 n$. We remark nonetheless that the probabilistic argument behind \cref{lem:fingerprint-matching} only requires that $\Delta \gg \eps^{-3}\log n$.

\begin{lemma}
    \label{lem:fingerprint-matching-implementation}
    Let $k$ be as in \cref{alg:fingerprint-matching} and assume $\Delta \gg k\log n$.
    \Cref{alg:fingerprint-matching} 
runs in $O(1/\eps^2)$ rounds with high probability.
\end{lemma}
\begin{proof}
    From \cref{lem:fingerprint-encoding}, the maxima $Y^v_i$ and $Y_i^K$ need $O(k + \log\log \Delta) = O(\eps^{-2}\log n)$ bits to describe in \cref{step:matching-fingerprint} with probability $1-2^{-\Theta(k)}$. Hence, w.h.p., each vertex learns their values in $O(1/\eps^2)$ rounds by aggregation on support trees.

    \Cref{step:matching-first-filter} is performed in $O({1/\eps^2})$ rounds by simple aggregations over a BFS tree in $K$. Since there are $O(\eps^{-2}\log n)$ trials, as long as only $O(1)$ bits are needed per trial for the aggregation, the aggregation is possible. Aggregating whether the maximum is unique is done by counting how many nodes satisfy $X_{v,i} = Y_i^K$, but with a sum capped at $2$. Whether a non-edge was detected is an OR of Boolean values at the nodes. Eliminating trials in which a node is a unique maximum for the second time is done by aggregating the OR of $O(\log n)$-bitmaps in which the nodes indicate which trial they want to throw away. The same idea is used again in \cref{step:selection-trumps-maximum,step:matching-pick-trial}.

    In \cref{step:matching-random-group}, by \cref{fact:random-groups}, w.h.p., each node is adjacent to one vertex of each group. Thus, as every node announces to its neighbors in which trials $i$ it is an eligible anti-neighbor using a $k=O(\eps^{-2}\log n)$ bitmap in \cref{step:matching-announce-group}, it is heard by at least one node of every group.

    In \cref{step:matching-min-pick}, we let the maximum ID node of each group sample a $(1/2,\card{K})$-min-wise hash function and send it to its group.
    By \cref{thm:min-wise}, such hash functions with input and output space of size $\Theta(\card{K})$ exist that can be described in only $O(\log \card{K})$ bits. Here, we use the fact that we equipped nodes in $K$ with identifiers between $1$ to $\card{K}$ in \cref{step:matching-renaming}.
    Group $i$ computes the ID of the node participating in trial $i$ that hashes to the smallest value by having all nodes in its support trees compute the hashes of their neighbors in $A_i$, and aggregating the minimum (\cref{step:matching-min-compute}). That minimum is then disseminated in all the support trees of group $i$ during \cref{step:matching-sampled-anti}. Each potential anti-neighbor aggregates on its support tree an $O(\log n)$ bitmap of the trials in which it was sampled.

    Like \cref{step:matching-sampled-anti}, \cref{step:selection-trumps-maximum,step:matching-pick-trial} are performed by each node using a $k=O(\eps^{-{2}}\log n)$ bitmap. In \cref{step:matching-sampled-anti}, the bitmap indicates in which trials the node was selected. In \cref{step:selection-trumps-maximum,step:matching-pick-trial}, it indicates which trials it opts out of. Aggregating those bitmaps over the whole cabal means every node knows the set of trials that produced an edge for the matching, and every matched node knows to be so and learns the ID of its relevant anti-neighbor from the relevant random group.

    The matching computed as output in \cref{step:matching-output} is only distributively known, but at least for each edge of the matching, the random group that was in charge of the trial that produced the edge knows the edge and can inform both endpoints of the other endpoint.
\end{proof}

\LemFingerprintMatching*
\begin{proof}
    Let us analyze \cref{alg:fingerprint-matching} trial by trial. As in the algorithm's pseudocode, let $u_i = \bot$ if $i\notin I$ in \cref{step:matching-first-filter} and let otherwise $u_i$ denote the unique maximum in trial $i \in [k]$ and identified by the algorithm in that step. Similarly, when $u_i \neq \bot$, let $w_i$ be the random anti-neighbor of $u_i$ sampled in \cref{step:matching-min-compute}. We define the following (random) sets $U_i$ and $W_i$ for the analysis:
\begin{align*}
    U_i 
    &\eqdef \set{v \in K: \exists j \leq i, u_j = v} \setminus \set{v \in K: \exists j \leq i, w_j = v}\ ,
    \\
    W_i &\eqdef \set{v \in K: \exists j \leq i, u \in U_i, w_j = v \wedge u_j = u}\ .
    \end{align*}
    
    Intuitively, $U_i$ is the set of unique maxima in the first $i$ trials that have not been removed from the set of useful trials in \cref{step:selection-trumps-maximum} by also being a sampled anti-neighbor. The set $W_i$ contains the anti-neighbors randomly selected by nodes in $U_i$.
    Note that $\card{W_i}$ never shrinks: in the case where the randomly selected anti-neighbor $w_i$ was already in $W_{i-1}$, then $|W_{i}| = |W_{i-1}|$ (the sets are without multiplicity); if $w_i$ happens to be a previous unique maximum $u_j$, then $w_j$ might exit the set $W_{i-1}$, but $w_i = u_j$ joins it; if $w_i$ is not in $W_{i-1}$ and was not a previously sampled unique maximum, $W_i = W_{i-1} \cup \set{w_i}$ and $\card{W_i} = \card{W_{i-1}} +1$.

    Each node in $W_k$ is a guaranteed anti-edge for the matching.
    Indeed, consider the bipartite graph between nodes of $U_k$ and $W_k$, with an edge between $u \in U_k$ and $w\in W_k$ if there exists $i \in [k]$ s.t.\ $u_i = u$ and $w_i = w$. Edges of this graph corresponds to anti-edges in $K$. It has maximum degree $1$ on $U_k$'s side, and minimum degree $1$ on $W_k$'s side. The former is true because each $u\in K$ may be a $u_i$ for at most one index in $I$, by the third condition in \cref{step:matching-first-filter}. The latter holds because each $w = w_j \in W_k$ is an anti-neighbor of $u_j \in U_k$, by definition of $W_k$. Taking one edge incident on every node of $W_k$ thus gives a matching of size $\card{W_k}$.
    
    We now show that $\card{W_k} \geq \Omega(\tau\hat{a}_K / \eps)$ w.h.p., i.e., \cref{alg:fingerprint-matching} finds a matching of this size.
    For completeness, set $W_{0} = \emptyset$ and define $Z_i$ for all $i\in[k]$ as follows:
    \begin{itemize}
        \item if $\card{W_{i-1}} \leq (\tau/4)\cdot \hat{a}_K/\eps$, then $Z_i \eqdef \card{W_i} - \card{W_{i-1}}$,
        \item if $\card{W_{i-1}} > (\tau/4)\cdot \hat{a}_K/\eps$, then $Z_i \eqdef 1$.
    \end{itemize}

    Observe that $U_i$, $W_i$ and thus $Z_i$ are functions of the random variables $X_{v, \leq i}$.
    We claim that the lower bound 
    \[ 
    \Pr[Z_i=1~|~X_{v,<i} = x_{v,<i}] \geq \tau/12
    \] holds for any conditioning $\set{x_{v, <i}}_{v\in K}$ on previous trials. Before proving it, we explain how it implies our result. By \cref{lem:chernoff}, w.h.p., we have $\sum_{i=1}^k Z_i > \tau/24 \cdot k$. 
    Since we chose $k$ such that $k\geq 6\hat{a}_K/\eps$ (by \cref{fact:med-avg}), we have that $\sum_{i=1}^k Z_i \geq \tau/24 \cdot k > (\tau/4)\cdot \hat{a}_K /\eps$. 
    Let $i\in[k]$ be the first index such that $\sum_{j \leq i} Z_j \geq (\tau/4)\cdot \hat{a}_K /\eps$. It must exist because the total sum is strictly larger. By definition, $Z_j = |W_j| - |W_{j-1}|$ for all $j \leq i$ and $\sum_{j \leq i} Z_j = |W_i|$ thus $\sum_{j \leq i} Z_j = |W_{i}| \geq (\tau/4)\cdot  \hat{a}_K /\eps$. Thus we found a large matching $|W_k| \geq |W_{i}| \geq (\tau/4)\cdot  \hat{a}_K /\eps$.
    
    When $\set{x_{v,<i}}$ is a conditioning such that $\card{W_{i-1}} > (\tau/4)\cdot \hat{a}_K/\eps$, we get $\Pr[Z_i=1~|~X_{v,<i} = x_{v,<i}] = 1 \geq \tau/12$. We assume henceforth that $\card{W_{i-1}} \leq (\tau/4)\cdot \hat{a}_K/\eps$.
    
    Consider the anti-edges incident on the nodes in $\card{W_{i-1}}$, the total of their incidences is $\sum_{v \in W_{i-1}} a_v \leq \card{W_{i-1}} \cdot \eps \Delta \leq (\tau / 4)\cdot \hat{a}_K \Delta$.
From the number of anti-edges incident to the nodes of $W_{i-1}$, we get that at most $(\tau / 2) \Delta$ nodes have $\hat{a}_K/2$ or more anti-neighbors in $W_{i-1}$. Since high nodes have anti-degree at least $\hat{a}_K$, this means that at most $(\tau / 2) \Delta$ high nodes can have half or more of their neighborhood in $W_{i-1}$.
    Since there are at least $\tau \card{K}$ high nodes, and at most $2k$ of them were sampled in $U_{i-1}$ or $W_{i-1}$ by previous trials, at least    
    $\tau|K| - (\tau/2)\Delta - 2k\geq \tau\card{K}/3$ non-previously sampled high nodes have more than half their anti-neighbors outside $W_{i-1}$. If the $i$-th trial sampled such a $u_i$, it would join $U_i$, thereby increasing its size by one.
    
    Conditioned on the existence of a unique maximum, this maximum occurs at a non-previously sampled high node with a majority of anti-neighbors outside $W_{i-1}$ with probability at least $\tau/3$. The probability of the maximum being unique is at least $2/3$ (\cref{lem:unique-maximum}), and the probability of the random anti-neighbor being outside $W_{i-1}$ is at least $1/4$, where the probability is smaller than $1/2$ from the parameter chosen for min-hashing (\cref{thm:min-wise}). Selecting a random anti-neighbor outside $W_{i-1}$ grows $W_{i-1}$ by a new node, meaning $Z_i = \card{W_i}-\card{W_{i-1}}$ equals one with probability at least $\tau/3 \cdot 2/3 \cdot 1/4 = \tau/12$. 
\end{proof}

As for coloring the matching, the parts of \alg{FingerprintMatching} that rely on random groups have to be modified to achieve the same results in the low-degree setting ($\Delta \in O(\log^2 n)$). \Cref{sec:lowdeg-inliers} covers how we adapt the algorithm to this setting, at the cost of a higher $O(\log \log n)$ runtime.

\section{Coloring Put-Aside Sets}
\label{sec:put-aside-sets}
\PropColorPutAside*

This section focuses on coloring put-aside sets in cabals, which we recall are extremely dense almost-cliques where $\tilde{e}_K < \lmin$. 
In \cref{sec:proof-put-aside}, we describe the full algorithm as well as the formal guarantees given by each intermediate step. Details for intermediate steps are given in subsequent subsections.

\subsection{Proof of \cref{prop:color-put-aside}}
\label{sec:proof-put-aside}

Let $\col$ be the partial coloring produced by the algorithm before we color put-aside sets. Crucially, it is not adversarial in cabals. 
The only uncolored vertices are the put-aside sets, i.e., $V\setminus \dom \col = \cup_{K} P_K$.
\cref{lem:compute-put-aside} ensures that each $P_K \subseteq I_K$ has size $r = \Theta(\lmin)$ and all inliers verify $e_v\le 25\lmin$. From \cref{lem:colorful-matching-high,lem:satisfied-nodes}, we know there exist $0.9\Delta$ nodes in $K$ such that $a_v\leq M_K$ (although vertices do not know if they verify it).
Henceforth, we refer to uncolored nodes of $K$ as $u_{K,1},$ $\ldots,$ $u_{K,r}$. 
Each uncolored vertex learns its index by computing prefix sums on a BFS tree spanning $K$ (\cref{lem:prefix-sum}). Note that each vertex knows $r$.

\paragraph{Parameters.}
We split the color space into contiguous \emphdef{blocks} of $b$ colors each, defining:
\begin{equation}
    \label{eq:put-aside-params}
    \ls \eqdef \Theta(\lmin^3) \ , \quad
    b \eqdef 256\ls^{6} \ , \quad 
    \text{and for each }i\in \range*{\frac{\Delta+1}{b}}, \quad
    B_i \eqdef \set{(i-1)b+1, \ldots, ib} \ .
\end{equation}

\begin{algorithm}[H]
    \caption{\colorPutAside\label{alg:coloring-put-aside}}

    \nonl\Input{The coloring $\col$ as described in \cref{prop:color-put-aside}}

    \nonl\Output{A total coloring $\coltotal$}

    \uIf{$|L_\col(K)| \geq \ls$}{
        \alg{TryFreeColors} \label{line:try-free-colors}
    }\Else{
        \FindCandidateDonors \hfill  (\cref{sec:put-aside-possible-recolorings})\label{line:put-aside-recolorings}

        \FindSafeDonors \hfill (\cref{sec:internally-safe-swaps})\label{line:find-internally-safe-swaps}

        \DonateColors \hfill \label{line:ext-safe-swaps}
    }
\end{algorithm}

\paragraph{Use Free Colors If Any (\cref{line:try-free-colors}).}
We say a color is 
\begin{itemize}
    \item \emphdef{unique} in $K$ if exactly one vertex in $K$ uses it and
    \item \emphdef{free} in $K$ if it is not used by colored nodes in $K$ (i.e., it is in the clique palette). 
\end{itemize}
Although there should be almost no free colors at this point of the algorithm, we cannot ensure there are none. 
It simplifies subsequent steps to assume there are few free colors.
Hence, before running the donation algorithm, vertices count the number of colors $|L_\col(K)|$ in the clique palette using the query algorithm.
If there are fewer than $\ls$ free colors, they go to \cref{line:put-aside-recolorings} of \cref{alg:coloring-put-aside}. Otherwise, they can all get colored in $O(1)$ rounds with high probability as we now explain.

Suppose $K$ has at least $\ls$ free colors. Vertices of $K$ find a hash function $h_K : [\Delta+1] \to [\poly(\log n)]$ that does not collide on the $\ls$ smallest colors of $L_\col(K)$. This can be done using, say, random groups. We defer details to \cref{lem:hash-clique-palette-col-free} to preserve the flow of the chapter. Such a hash function can be represented in $O(\log \log n)$ bits and each hash takes $O(\log\log n)$ bits.
Each $u_{K, i}$ samples $k = \Theta\parens*{ \frac{\log n}{\log\log n} }$ indices in $\set{1, 2, \ldots, \ls}$ with replacement. They learn the hashes of the corresponding colors in $L_\col(K)$ with a straightforward adaption of the query algorithm (\cref{lem:query}). Overall, the message with the hashes of the sampled colors takes $k\cdot O(\log\log n) = O(\log n)$ bits. A hash is safe iff it does not collide with $h(\col(E_{u_{K,i}}))$ the hashes of external neighbors nor with those of other put-aside vertices.
Since $h_K$ is collision free on the $\ls$ smallest colors of $L_\col(K)$, each hash sampled by $u_{K,i}$ is uniform in a set of $|h(L_\col(K))| = |L(K)| = \ls$ colors.
As $P_K \subseteq I_K$, external neighbors block $\leq 25\lmin$ hashes (recall external neighbors are not put-aside nodes). Put-aside nodes $u_{K, \neg i}$ from the same cabal block $kr \le \lmin^2$ hashes. By union bound, a hash is unsafe with probability $\Theta(\lmin^2) / \ls \le 1/\log n$ (because $\ls = \Theta(\lmin^3)$). Thus, the probability that none of the hashes sampled by $u_{K, i}$ are safe is $\parens*{ 1/\log n }^k = 2^{-k\log\log n} \le 1/\poly(n)$. By union bound, all uncolored nodes of $K$ find a safe hash with high probability.
Using the query algorithm, they then learn which color of (the smallest $\ls$ colors of) $L(K)$ hashes to that value in $O(1)$ rounds and get colored.

\begin{remark}
    \label{rmk:bound-free-repeated}
Henceforth, we assume that the number of free or repeated colors in $K$ is at most $3\ls$. Indeed, we may assume that $|K| < \Delta+1 + \ls$ as otherwise the colorful matching has size $\Omega(a(K)/\eps) \gg \ls$, and the clique palette $L_\col(K)$ has at least $\ls$ free colors (by \cref{lem:clique-slack-cabal} for $v$ with $a(v) \leq a(K)$). If the number of repeated colors is at least $2\ls$, there are at least $\ls$ free colors since $|L_\col(K)| \geq \Delta + 1 - |K| + |P_K| + 2\ls \geq \ls$ from the assumption that $|\Delta| < \Delta+1 - \ls$.
\end{remark}

\paragraph{Finding Candidate Donors (\cref{line:put-aside-recolorings}).} 
Henceforth, we assume $|L_\col(K)| < \ls$, i.e., there are few free colors.
Being unique is necessary for a color to be donated, but it is not sufficient.
\FindCandidateDonors (\cref{alg:Q-find}) computes in each cabal a set $Q_{K}$ of candidate donors. 
A candidate donor holds a unique color in its cabal and is not adjacent to candidate donors or put-aside nodes of other cabals. 
Naturally, we look for many candidates as it increases the chances of finding a donation later.
\begin{restatable}{lemma}{LemFindQ}
    \label{lem:find-Q}
After \FindCandidateDonors (\cref{alg:Q-find}), w.h.p., 
    for each cabal $K$ where $|L_\col(K)| < \ls$, the algorithm computes sets $Q_{K} \subseteq I_K$ of inliers such that 
    \begin{enumerate}
        \item\label[part]{part:unique}
        for each $v \in Q_{K}$, the color $\col(v)$ is \emph{unique} in $K$ and at least $\frac{\Delta+1}{b} \cdot 8\ls^3$ of them have $a_v \leq M_K$;
        \item\label[part]{part:ext} no edge connects $Q_{K}$ to $P_{\neg K} \cup Q_{\neg K}$.
    \end{enumerate}
    The algorithm ends after $O(1)$ rounds and each vertex $v\in K$ knows if $v\in Q_K$.
\end{restatable}

\paragraph{Finding Safe Donations (\cref{line:find-internally-safe-swaps}).} 
\FindSafeDonors (\cref{alg:find-S}) provides each $u_{K, i}$ with a large set $S_{K, i}$ of \emph{safe donors}, all of which can be recolored using the \emph{same} replacement color $\crecol_{K, i}$, which is different from all $\crecol_{K,\neg i}$.
We emphasize that the donation is safe for the donor (the vertex in $S_{K, i}$) but that $u_{K, i}$ might reject it because the donated color is used by one of its external neighbors.
Donors are safe as the replacement color is in their palette (\cref{part:palette}) and not used as a replacement by donors from another group (\cref{part:different}). All donations for $u_{K,i}$ come from the same block (\cref{part:S-range}) so that they can be described with few bits.
Finding many safe donors (\cref{part:large}) ensures that a random donation is likely to work (see \cref{line:ext-safe-swaps}).

\begin{restatable}{lemma}{LemFindIntSafe}
    \label{lem:find-S}
    After \FindSafeDonors (\cref{alg:find-S}), w.h.p.,
    for each cabal $K$ where $|L_\col(K)| < \ls$, for each $i\in[r]$, there exists a triplet $(\crecol_{K, i}, j_{K,i}, S_{K, i})$ where $\crecol_{K, i} \in L_\col(K)$, $j_{K,i}\in \range*{\frac{\Delta+1}{b}}$ and $S_{K, i} \subseteq Q_{K}$ such that
    \begin{enumerate}
        \item\label[part]{part:different} each $\crecol_{K, i} \neq \crecol_{K,i'}$ for any $i' \in [r]\setminus\set{i}$ and sets $S_{K,i}$ are pairwise disjoint,
        \item\label[part]{part:palette} each $v\in S_{K, i}$ has $\crecol_{K, i}\in L_\col(v)$, 
        \item\label[part]{part:S-range} each $v\in S_{K, i}$ has $\col(v) \in B_{j_{K,i}}$, and
        \item\label[part]{part:large}
        $|S_{K, i}| = \ls$. 
    \end{enumerate}
    The algorithm ends in $O(1)$ rounds. Each vertex knows if $v\in S_{K, i}$ for some $i\in[r]$ and, if so, knows $\crecol_{K, i}$. Each uncolored vertex $u_{K,i}$ knows $j_{K,i}$. 
\end{restatable}

\paragraph{Donating Colors (\cref{line:ext-safe-swaps}).}
Each uncolored vertex $u_{K, i}$ samples $k=\Theta(\frac{\log n}{\log\log n})$ colors with replacement in its set $\col(S_{K, i})$ of safe donations.
We claim it samples some $\cdon_{K, i} \notin \col(E_{u_{K, i}})$, i.e., a color unused by external neighbors. 
Indeed, recall that $u_{K, i}$ has at most $e_v \le 25\lmin$ external neighbors (by definition of inliers). As each color in $\col(S_{K, i})$ is unique, it contains $|\col(S_{K, i})| = |S_{K, i}| \geq \ls$ colors. By union bound, a uniform color in $\col(S_{K, i})$ conflicts with an external neighbor of $u_{K,i}$ with probability at most $25\lmin/\ls \le 1/\log n$ (because $\ls \gg \lmin^3$, \cref{eq:put-aside-params}). Thus, the probability that all $k$ samples hit a bad color is at most
$\parens*{1/\log n}^k = 2^{-k\log\log n} \le 1/\poly(n)$. By union bound, w.h.p., each uncolored $u_{K, i}$ finds $\cdon_{K, i}$ used by a donor in $S_{K, i}$.
We argue that the following coloring is total and proper
\[
    \coltotal (w) \eqdef \begin{cases}
        \cdon_{K, i} &\text{when } w = u_{K, i}\\
        \crecol_{K, i}  &\text{when } w \in S_{K, i} \text{ and } \cdon_{K, i} = \col(w)\\
        \col(w)  &\text{otherwise} 
    \end{cases} \ .
\]
Each uncolored vertex $u_{K, i}$ gets colored with $\cdon_{K, i}$.
They are properly colored because 
1) no external neighbors gets colored nor recolored (\cref{lem:compute-put-aside}, \cref{part:put-aside-put-aside} and \cref{lem:find-Q}, \cref{part:ext}), 
2) no put-aside neighbors in $K$ gets the same color $\cdon_{K, i} \in \col(S_{K, i})$ since sets $S_{K,*}$ are disjoint (\cref{lem:find-S}, \cref{part:different}) and contain only unique colors ($S_{K,i} \subseteq Q_K$ and \cref{lem:find-Q}, \cref{part:unique}), 
3) color $\cdon_{K, i}$ is not used in any recoloring in $K$ (because it is not free), and
4) at most one vertex in $K$ was colored $\cdon_{K, i}$, which recolor itself using some other color (\cref{lem:find-Q}, \cref{part:unique}).

Let $v \in S_{K, i}$ be the vertex donating its color to $u_{K, i}$; recall it is recolored with $\crecol_{K, i}$.
It is properly colored because 
1) none of its external neighbors gets colored nor recolored (\cref{lem:find-Q}, \cref{part:ext}),
2) as explained before, it does not conflict with any $u_{K,*}$ getting colored,
3) it does not conflict with any recoloring in $K$ as $\crecol_{K,*}$ are all different (\cref{lem:find-S}, \cref{part:different}), and
4) it does not conflict with other colored neighbors (in $K$ or outside) as $\crecol_{K, i}\in L_\col(v)$ (\cref{lem:find-S}, \cref{part:palette}).

\paragraph{Implementing \cref{line:ext-safe-swaps}.}
Implementations of \FindCandidateDonors and \FindSafeDonors guarantee that a vertex knows when it belongs to sets $Q_{K, i}$ and $S_{K, i}$ for some $i\in [r]$. Moreover each $u_{K,i}$ knows the block $j_{K,i}$. 

To sample colors in $\col(S_{K,i})$, uncolored vertices sample $k$ indices in $\set{1, 2, \ldots, |\col(S_{K,i})|}$. To learn $|\col(S_{K,i})|$, we partition $K$ into $r$ random groups. Compute BFS trees $T_1, \ldots, T_r$ where each $T_i$ spans the $i\supth$ random group and $S_{K,i}$. Observe that a vertex belongs to at most two trees (one for its random group, and possibly an other if it belongs to $S_{K,*}$). Hence, congestion on trees $T_i$ causes only $O(1)$ overhead. By aggregation on trees $T_i$, for all $i\in[r]$ in parallel, we count $|S_{K,i}| = |\col(S_{K,i})|$ exactly in $O(1)$ rounds.

After sampling indices in $\set{1, 2, \ldots, |\col(S_{K,i})|}$, each $u_{K,i}$ needs to learn the colors they correspond to. We claim the list of sampled donations can be represented by a $O(\log n)$-bit message. Indeed, since $\col(S_{K, i}) \subseteq B_{j_{K,i}}$ for each $i$, to encode $k$ colors, it suffices to represent each sampled color by its offset in $B_{j_{K,i}}$. Overall, it uses $k\cdot O(\log b) = O(\log n)$ bits.

To check for collisions with external neighbors, $u_{K,i}$ uses the same representation but also includes the index of the block $j_{K.i}$ (which is not known to external neighbors). Then it broadcasts this $O(\log n)$-bit message to external neighbors. They respond, say, with a bitmap describing which sampled colors conflict with $\col(E_{u_{K,i}})$.

\subsection{Finding Candidate Donors}
\label{sec:put-aside-possible-recolorings}
The algorithm first filters out nodes with external neighbors in put-aside sets, the remaining vertices join $\Qpre_K$.
It then \emph{activates} nodes independently with probability $\Theta(\ls^3/b) = \Theta(1/\ls^3)$, in a set $\Qactive_K$.
It retains the active nodes that have unique colors and have no external active neighbors, resulting in $Q_K$.

As \cref{part:ext} of \cref{lem:find-Q}
is direct from \cref{alg:Q-find}, the focus of this section is on proving \cref{part:unique}.
We first prove the algorithm is correct and conclude with its implementation on cluster graphs.

\begin{algorithm}
    \caption{\FindCandidateDonors\label{alg:Q-find}}

    \nonl\Input{The coloring $\col$ and the integer $N$}
    
    \nonl\Output{Sets $Q_{K} \subseteq I_K$}

    Let $\Qpre_{K}$ be uncolored inliers of $K$ with no external neighbor in a put-aside set.
    $$
    \Qpre_{K} = \set*{u\in I_K: E_v \cap \parens*{\bigcup_{K'\neq K} P_{K'}} = \emptyset} \ .
    $$
    \label{line:Q-pre}

    Each $v\in \Qpre_{K}$ joins $\Qactive_{K}$ w.p.\ $p = \frac{50\ls^3}{b}$. 
    \label{line:Q-active}
    
    Let $Q_{K}$ be the vertices of $\Qactive_{K}$ with no active external neighbors.
    \[ Q_{K} = \set*{v\in \Qactive_{K}: \col(u)\text{ is unique}\quad\text{and}\quad N(v) \cap \parens*{\bigcup_{K'\neq K} \Qactive_{K'}} = \emptyset} \]
    \label{line:find-Q}
\end{algorithm}

\begin{lemma}
    \label{lem:unique}
    At the end of \cref{alg:Q-find}, w.h.p, \cref{part:unique} of \cref{lem:find-Q} holds for each cabal.
\end{lemma}

\begin{proof}
    We call a color \emphdef{good} if and only if $i)$ it is unique, and $ii)$ used by a vertex in $I_K$ with $a_v \leq M_K$ and no external neighbors in a put-aside set. By extension, a vertex $v$ is good if and only if $\col(v)$ is good. Observe that a good vertex is in $\Qpre_K$ and it joins $Q_K$ iff it becomes active in \cref{line:Q-active} while none of its $O(\lmin)$ external neighbors do.

    First, we claim that each cabal contains at least $0.75\Delta$ good vertices/colors. There are at most $3\ls$ colors that are not unique in $K$ (\cref{rmk:bound-free-repeated}). By assumption on $I_K$, we drop at most $(0.1 + \eps)\Delta$ colors because they are used by nodes in $K\setminus I_K$ (\cref{lem:size-inliers-cabals}). By \cref{lem:compute-put-aside}, \cref{part:put-aside-inliers}, at most $0.01|K| \leq (0.01 + \eps)\Delta$ inliers have an external neighbor in some $P_{K'}$ for $K'\neq K$. Overall, the number of good colors in $[\Delta+1]$ is at least $(\Delta+1) - (0.1+\eps)\Delta - (0.01+\eps)|K| - 3\ls \ge 3/4 \cdot \Delta$.

    Each vertex becomes active in \cref{line:Q-active} independently with probability $p$. Hence, the classic Chernoff Bound implies that, w.p.\ $1-e^{\Theta(-p\Delta)} = 1-e^{-\Theta(-\Delta \ls^3/b)} \ge 1-1/\poly(n)$, at least $p\cdot 3/8 \cdot \Delta$ \emph{good} vertices in $\Qpre_{K}$ become active. 
    Fix some $v\in \Qactive_{K}$. The expected number of active external neighbors of $v$ is $pe_v \le p\cdot 25\lmin \le 1/10$. The first inequality holds because $\Qactive_{K, i} \subseteq I_K$ and the second because $b \gg \ls^3 \cdot \lmin$ (\cref{eq:put-aside-params}). By Markov inequality, vertex $v$ fails to join $Q_{K}$ in \cref{line:find-Q} with probability at most $1/10$. Thus, in expectation, at least $9/10 \cdot p\cdot 3/8 \cdot \Delta \ge \frac{\Delta+1}{b} \cdot 16\ls^3$ \emph{good} vertices of $\Qactive_{K}$ join $Q_{K}$. 

    For each good vertex $v$, let $Y_v$ and $Z_v$ be the random variables indicating, respectively, if $v\in \Qactive_{K_v}$ and $v \in Q_K$.
    Variables $Z_v$ are \emph{not} independent, but all depend on random independent activations of external neighbors with \emph{low external degree}. 
The activation of $w\in \Qpre_{\neg K} \subseteq I_{\neg K}$ influences at most $e_w \le 25\lmin$ external neighbors of $w$ in $K$. Hence, we can apply the read-$k$ bound with $k=25\lmin$ (\cref{lem:k-read}) to show concentration on $\sum_{v} Z_v$. With probability at least $1-e^{-\Theta(p\Delta/k)} = 1- e^{-\Theta(\Delta\ls^3/b\lmin)} \ge 1-\poly(n)$, the number of \emph{good} vertices in $Q_{K}$ is at least $\frac{\Delta+1}{b} \cdot 8\ls^3$.
\end{proof}

\LemFindQ*

\begin{proof}
\cref{part:ext} holds by construction of $Q_{K}$ (\cref{line:Q-active,line:find-Q}). \cref{part:unique} holds with high probability by \cref{lem:unique}. 

Implementation-wise, the only non-trivial step in \cref{alg:Q-find} is testing if a color is unique in \cref{line:find-Q}.
Crucially, we only need to test the uniqueness of colors for active vertices. Since each vertex becomes active independently, w.h.p., each cabal contains at most $2(1+\eps)\Delta p \ll \Delta/\Theta(\log^2 n)$ active vertices (recall $b \gg \ls^3 \log n$, \cref{eq:put-aside-params}).
Hence, we can partition $K$ into $|\Qactive_K|$ random groups, where the $i\supth$ group tests the uniqueness of $\col(v)$ where $v$ is the $i\supth$ vertex in $\Qactive_K$.
Observe we can count the size of $\Qactive_K$ and order active vertices using the prefix sum algorithm (\cref{lem:prefix-sum}).
A group tests if its attributed color $c$ is unique in two aggregations: first, they compute the smallest identifier of a node colored $c$ in $K$; then, they compute the next smallest identifier for a node colored $c$. If they can find two different identifiers, then the color is not unique. They never misclassify a color since each random group is adjacent to all vertices of $K$ (\cref{fact:random-groups}). Each group broadcasts its findings and vertices learn if their color is unique.
\end{proof}

\subsection{Finding Safe Donors}
\label{sec:internally-safe-swaps}
Focus on cabal $K$, let $u_1, \ldots, u_r$ be the uncolored vertices, and let $Q$ be the set of inliers returned by \FindCandidateDonors. We show that if each vertex in $Q$ samples a uniform color in the clique palette, at least $|P_K| = r$ colors are sampled by at least $\ls$ vertices from the same block.
Throughout this section, for block $j$, let 
\[ C_j \eqdef Q \cap \col^{-1}(B_j) = \set{v \in Q: \col(v) \in B_j} \]
be the vertices of $Q$ colored with colors from the $j\supth$ block.

\begin{lemma}
    \label{lem:many-happy}
    If each vertex in $Q$ samples a uniform color $c(v) \in L_\col(K)$, then w.h.p., there exist $r$ colors $c_1, c_2, \ldots, c_r \in L_\col(K)$ and $r$ block indices $j_1, j_2, \ldots, j_r$ such that the number of vertices in $C_{j_i}$ that sampled $c_i$ is
    \begin{equation}
        \label{eq:many-happy}
        \card{\set{v\in C_{j_i}: c(v) = c_i \in L_\col(v) }} \geq 4\ls \ .
    \end{equation}
\end{lemma}

\begin{proof}
Call a vertex $v\in Q$ \emphdef{good} iff $a_v \leq M_K$ and denote by $\calG = \set{v\in Q: a_v \leq M_K}$ the set of good vertices in $K$.
We define a color $c\in L_\col(K)$ in the clique palette as \emphdef{happy} iff $c$ is available to at least a $1/\ls$-fraction of the \emph{good} vertices
\[
    \card{\set{v\in \calG: c\in L_\col(v)}} \geq 
    \frac{|\calG|}{\ls} \ .
\]
As we shall see, happy colors are likely to be sampled by many vertices. 

First, we claim there exist at least $r$ happy colors in the clique palette. Recall, by \cref{lem:clique-slack-cabal}, that good vertices always have at least $|L_\col(v) \cap L_\col(K)| \geq |K\setminus\dom\col| = |P_K| = r$ colors available in the clique palette. 
Hence, if $h$ is the number of happy colors, then
\begin{equation}
    \label{eq:avg-happy}
    r \leq 
    \frac{1}{|\calG|}\sum_{v\in \calG} |L_\col(v) \cap L_\col(K)| = 
    \frac{1}{|\calG|}\sum_{c\in L_\col(K)} \card{\set{v\in \calG: c\in L_\col(v)}} < 
    h + \frac{|L_\col(K)|}{\ls} \ ,
\end{equation}
where the last inequality comes from the fact that each happy color can contribute to all palettes, while the $\leq |L_\col(K)|$ unhappy colors contribute to fewer than $|\calG|/\ls$ palettes, by definition.
Since we assumed that $|L_\col(K)| < \ls$, unhappy colors contribute very little and \cref{eq:avg-happy} simplifies to $r < h+1$. Because $r$ and $h$ are integral, the number of happy colors is $h \geq r$.

We now argue that for each happy color $c_i = c$, we can find a block $j_i = j$ such that \cref{eq:many-happy} holds.
By \cref{part:unique} in \cref{lem:find-Q}, there are at least $|\calG| \geq \frac{\Delta+1}{b} \cdot 8\ls^3$ good vertices in $K$. This implies that there exists a block containing $\geq 8\ls^2$ vertices with $c$ in their palettes, as otherwise the total number of good vertices with $c$ in their palette would be $\leq \frac {\Delta+1}{b} \cdot 8\ls^2 < |\calG|/\ls$ contradicting $c$ being happy. Let $j \in \range*{\frac{\Delta+1}{b}}$ be the index of a block such that $\geq 8\ls^2$ vertices in $C_j$ have $c$ in their palette.

Since each vertex $v \in C_j$ samples $c(v) = c$ independently with probability $1/|L_\col(K)| \geq 1/\ls$, we expect at least $8\ls$ of them to sample $c(v) = c$. By Chernoff, w.h.p., at least $4\ls$ vertices in $C_j$ sample $c$, which shows \cref{eq:many-happy} for happy color $c$. By union bound, it holds with high probability for all happy colors in all cabals.
\end{proof}

\cref{lem:many-happy} shows that after vertices sample a random color in the clique palette, it suffices to \emph{detect} which colors $c$ have a block $j$ such that $C_j$ contains many vertices that sampled $c$. We achieve this through random groups and fingerprinting. \cref{alg:find-S} gives the main steps and implementation details are in the proof of \cref{lem:find-S}.

\begin{algorithm}
    \caption{\FindSafeDonors\label{alg:find-S}}
    \nonl\Input{The coloring $\col$ and sets $Q_K \subseteq K$ as given by \cref{line:put-aside-recolorings}.}

    \nonl\Output{Sets $S_{K, i} \subseteq Q_{K}$, blocks $j_{K,i}$ and colors $\crecol_{i} \in L_\col(K)$ as described in \cref{lem:find-S}}

    Each $v\in Q_K$ samples a uniform color $c(v) \in L_\col(K)$  \emph{and drop it if $c(v) \notin L_\col(v)$.}
    \label{line:S-sample-color}

    Partition $K$ into $|L_\col(K)|\cdot \frac{\Delta+1}{b}$ random groups indexed as $(c, j) \in L_\col(K) \times \range*{\frac{\Delta+1}{b}}$. \\
\nonl Group $(c, j)$ estimates the number of vertices in the $j\supth$ block that sampled $c$ as replacement 
    \[ \beta_{c,j} \in (1\pm 0.5)\card{\set{ v \in C_j: c(v) = c \in L_\col(v) }} \ . \]
    \label{line:beta}

    For each $c\in L_\col(K)$, find $j(c) \in \range*{\frac{\Delta+1}{b}}$ for which $\beta_{c,j} > 2\ls$. If no such $j(c)$ exist, set $j(c) = \bot$.
    \label{line:select-block}

    Select $r$ colors $c_1, \ldots, c_r$ for which $j(c) \neq \bot$\\
    \nonl Return $\crecol_{K,i} = c_i$ with $j_{K,i} = j({c_i})$ and $S_{K, i} = \set{ v \in C_j: c(v) = c_i \in L_\col(v) }$
    \label{line:output}
\end{algorithm}

\LemFindIntSafe*

\begin{proof}    
    We go over steps of \cref{alg:find-S}.

    \paragraph{Sampling (\cref{line:S-sample-color}).} Sampling a uniform color in the clique palette takes $O(1)$ rounds using the query algorithm (\cref{lem:query}). Each vertex checks if $c(v) \in L_\col(v)$ with a $O(1)$ round broadcast and bit-wise OR.
    
    \paragraph{Fingerprinting (\cref{line:beta}).}
    By \cref{fact:random-groups}, w.h.p., all vertices are adjacent to at least one vertex from each group.
    Using the fingerprinting algorithm (\cref{lem:fingerprint}), vertices of group $(j,c)$ aggregate fingerprints from $C_j$ and compute $\beta_{i,c}$ in $O(1)$ rounds. 
    
\paragraph{Selecting Blocks (\cref{line:select-block}).}
    For each color $c\in L_\col(K)$, build a BFS trees $T_c$ spanning random group for that colors $\set{(c, j), j\in\range*{\frac{\Delta+1}{b}}}$. Each group sets a bit $x_{c,j}$ to one if $\beta_{c,j}> 2\ls$ and zero otherwise. Use the prefix sum algorithm over $T_c$ (\cref{lem:prefix-sum}) to count for each group $(c,j)$ the prefix sum $\sum_{j' \leq j} x_{c,j}$. Let $j(c)$ be the index for which this prefix sum is exactly one, i.e., the block of smallest index with enough vertices to sample $c(v) = c$ in \cref{line:S-sample-color}. If none exists, then $j(c) = \bot$. This information is disseminated to all vertices in groups $\set{(c, j), j\in\range*{\frac{\Delta+1}{b}}}$ by converge/broadcast on $T_c$.
    Note that since random groups are disjoint, the trees $T_c$ for different colors are disjoint as well. In particular, the algorithm runs independently for each color without creating congestion, thus computing blocks $j(c)$ for all $c\in L_\col(K)$ in parallel. 
    
    \paragraph{Output (\cref{line:output}).}
    Similarly, using the prefix sum algorithm on one tree spanning $K$, we find the first $r$ colors $c_1, \ldots, c_r$ for which such a block exist. 
    For each $i\in[r]$, the algorithm returns $\crecol_{K,i} \eqdef c_i$ and $S_{K,i}$ the vertices in $C_{j(c_i)}$ that sampled $c_i$. 
    Each random group broadcasts a message if it was selected, thus each vertex knows if it belongs to $S_{K,i}$ for some $i \in[r]$.
    
    \paragraph{Correctness.}
    By \cref{lem:many-happy}, w.h.p., there exist at least $r$ colors for which \cref{eq:many-happy} holds. As $\beta_{c,j}$ is a $2$-approximation of the left-hand side in \cref{eq:many-happy}, there exist at least $r$ pairs $(c,j)$ such that $\beta_{c,j} > 2\ls$. Hence, the algorithm always finds $r$ colors in \cref{line:output}.

    Clearly, the colors $c_1, \ldots, c_r$ selected by the algorithm are different and the sets $S_{K,i}$ are disjoint (since each vertex samples exactly one color). Since a vertex retains the sampled color only if it belongs to its palette, \cref{part:palette} holds. \cref{part:S-range} is clear from the definition of $S_{K,i}$ as the subset of one block. 
    Finally, \cref{part:large} holds because when $\beta_{i,c} > 2\ls$, then at least $\ls$ vertices from block $i$ sampled color $c$.
\end{proof}

\section{Preparing MultiColorTrial in Non-Cabals}
\label{sec:prep-MCT}
\PrepMCT*

If a vertex does not have slack in $[r_v]$, it must have many available colors in $L(K_v) \setminus [r_v]$. In particular, if it tries a random color in $L(K_v) \setminus [r_v]$, it gets colored with constant probability. Trying random colors can be used to reduce the number of nodes without slack in $[r_v]$ to a small fraction of $r_v$ so as to color them using reserved colors, for only $e_K \ll r_K$ external neighbors can block colors.

The key technical challenge when implementing this on cluster graphs is to determine \emph{when} a vertex has slack in reserved colors. 
Vertices cannot count the number of available colors (i.e., in $|L(v) \cap L(K_v) \setminus [r_v]|$) exactly. We rather (approximately) count the number of \emph{nodes} using colors in $[\Delta+1] \setminus [r_v]$. 

Define $\mu_v^K(c) = |\col^{-1}(c) \cap K_v|$ as the number of nodes in $K_v$ with color $c$, and similarly $\mu^e_v(c) = |\col^{-1}(c) \cap E_v|$.
For each vertex $v$, we estimate $|L(v) \cap L(K) \setminus [r_v]|$ by counting the difference between the number of colors (i.e., the $\Delta+1 - r_v$ term) and the number of nodes with non-reserved colors (i.e., the two sums)
\begin{equation}
    \label{eq:z}
    z_v \eqdef 
    \parens*{
        (\Delta+1 - r_v) - \sum_{c = r_v+1}^{\Delta+1}\mu^{K}_v(c) - \sum_{c=r_v+1}^{\Delta+1} \mu^e_v(c)
    } 
    + \CCSlack \cdot e_K + 40a_K + x_v\ .
\end{equation}
The last three terms account for the reuse slack we expect for $v$ (\cref{lem:reuse-slack}).
We emphasize that $z_v$ depends implicitly on the current coloring.
\cref{lem:mu-count-high} shows that $z_v$ lower bounds the number of non-reserved colors in the clique palette that are available for $v$.

\begin{lemma}
    \label{lem:mu-count-high}
    The number of non-reserved colors available to a dense non-cabal inlier $v\in I_{K_v} \subseteq K_v \notin\Kcabal$ is at least 
    $|L_\col(v) \cap L_\col(K_v) \setminus [r_v]| \ge z_v$.
\end{lemma}

\begin{proof}
    Fix a node $v\in I_{K_v}$ for $K_v\notin \Kcabal$. 
    Since $\mu^K_v(c) + \mu^e_v(c)$ is the number of nodes in $K_v\cup E_v$ using color $c$, when that color is used by at least one vertex, $\mu^K_v(c) + \mu^e_v(c) - 1$ count the contribution of $c$-colored nodes to the reuse slack.
Hence, \cref{eq:reuse-slack} can be rewritten in terms of $\mu^K$ and $\mu^e$ as
    \begin{equation}
        \label{eq:bound-eK}
        \sum_{\substack{c \in [\Delta+1] \setminus [r_v]: \\ \mu^K_v(c) + \mu^e_v(c) \ge 1}} 
            (\mu_v^K(c) + \mu_v^e(c) - 1)
        \geq \CCSlack \cdot e_K + 40 a_K + x_v \ ,
    \end{equation}
    where we emphasize that the sum is only over non-reserved colors.
    Plugging \cref{eq:bound-eK} into \cref{eq:z}, we upper bound $z_v$ by the following complicated sum
    \[ 
        z_v \leq
        \parens*{ 
            \sum_{\substack{c \in [\Delta+1] \setminus [r_v]: \\ \mu^K_v(c) + \mu^e_v(c) \ge 1}} 
            (\mu_v^K(c) + \mu_v^e(c) - 1) 
        }
        +
        \parens*{
            (\Delta+1 - r_v) 
            - \sum_{c = r_v+1}^{\Delta+1}\mu^{K}_v(c) 
            - \sum_{c=r_v+1}^{\Delta+1}\mu^e_v(c) 
        }
    \]
    where the $\mu^K_v(c)$ and $\mu^e_v(c)$ cancel out leaving only the terms $\Delta+1-r_v$ and a sum of $-1$ over all colors used in $K\cup N(v)$, which is exactly the number of available non-reserved colors
    \begin{align*}
        z_v &\leq 
        (\Delta+1-r_v) - |\set{c\in[\Delta+1]\setminus [r_v]: \mu_v^K(c)+\mu_v^e(c) \geq 1 }| 
        = |L(v) \cap L(K) \setminus [r_v]| \ . \qedhere
    \end{align*}
\end{proof}

This second lemma shows that when $z_v$ is too small, then $v$ has slack in the reserved colors because many neighbors are colored using the same non-reserved colors.

\begin{lemma}\label{lem:mu-count-low}
    Let $\col$ be a coloring such that $\col(K) \cap [r_K] = \emptyset$. 
Then, the number of reserved colors available to each inlier $v \in I_{K_v}$ is
\[ |[r_v] \cap L_\col(v)| \ge \deg_\col(v) + 0.5\CCSlack \cdot e_K - z_v \ . \]
\end{lemma}
\begin{proof}
    Using that $\Delta+1 = (\Delta - \deg(v)) + |K| + e_v - a_v \geq |K| + e_v - x_v - \delta e_v$ (\cref{eq:x}), we lower bound $z_v$ as
    \[
        z_v \geq
        (|K| + e_v - \delta e_v) - r_v 
        - \parens*{ 
            \sum_{c=r_v+1}^{\Delta+1}\mu^K_v(c) + \sum_{c=r_v+1}^{\Delta+1}\mu^e_v(c)
        }
        + \CCSlack \cdot e_K
    \]
    Recall that $e_v \leq {25} e_K$ and $\delta \le 1/100$, and hence 
    $
        \CCSlack \cdot e_K - {25}\delta e_K \geq 
        0.5 \CCSlack \cdot e_K,
    $.
Because nodes of $K$ do not used reserved colors, we have $|K| = |K \setminus \dom\col| + \sum_{c=r_v+1}^{\Delta+1}\mu^K_v(c)$.
    Partition external neighbors of $v$ as 
    \[
        e_v =
        \underbrace{|E_v \setminus \dom\col|}_{\text{uncolored}} + 
        \underbrace{|E_v \cap \col^{-1}([r_v])|}_{\text{colored with $c\in[r_v]$}} + 
        \underbrace{\sum_{c=r_v+1}^{\Delta+1}\mu^e_v(c)}_{\text{colored with $c > r_v$}} \ .
    \]
     Since uncolored neighbors $N(v) \setminus \dom\col \subseteq (K_v \cup E_v) \setminus \dom\col$, the lower bound on $z_v$ becomes
    \[
        z_v \geq
        \deg_\col(v) +
        |E_v \cap \col^{-1}([r_v])| - r_v  + 0.5\CCSlack \cdot e_K\ .
    \]
    Finally, as $[r_v] \cap \col(K_v) = \emptyset$, the only colors lost in $|[r_v] \cap L(v)|$ must be used by external neighbors. That is $|[r_v] \cap L(v)| = r_v - |E_v \cap \col^{-1}([r_v])|$. 
    So the lemma follows from the lower bound on $z_v$.
\end{proof}

\subsection{Algorithm Analysis}
We have now all the ingredients of \cref{alg:reduce}. The first phase runs $O(1)$ random color trials. To decide if they should take part in the color trial, vertices approximate $z_v$. Indeed, with aggregation on a BFS tree spanning $K$, vertices of $K$ count $\sum_{c=r_v+1}^{\Delta+1}\mu^K_v(c)$ exactly. Using the fingerprinting technique, they approximate $\sum_{c=r_v+1}^{\Delta+1} \mu^e_v(c)$ up to an error $\Theta(\delta) \cdot e_v$. Since they also know $\tilde{e}_K \in (1\pm \Theta(\delta)) e_K$, they appropriate $z_v$ up to a small error.
\begin{claim}
    \label{claim:apx-z}
    Each $v$ can compute $\tilde{z}_v \in z_v \pm \delta e_K$ in $O(1/\delta^2)$ rounds.
\end{claim}

\begin{algorithm}
    \caption{\alg{Complete}\label{alg:reduce}}
    \nonl\Input{The coloring $\colsct$ produced by \sct}

    \nonl\Output{A coloring $\col$ such that $V \setminus \Vcabal = \dom\col$.}

    \tcp{Phase I}

    Let $\col_0 =\colsct$.

    \For{$i=1, 2, \ldots, t = O(1)$}{
        Each $v$ computes $\tilde{z}_{v,i} \in z_{v,i} \pm \delta e_K$ (w.r.t.\ $\col_{i-1}$)     

        Let $v$ join $S_i$ if $\tilde{z}_{v,i} \geq 0.25\CCSlack \cdot \tilde{e}_K$

        Each $v\in S_i$ runs \trycolor with $\calC(v) = |L_{\col_{i}}(K) \setminus [r_v]|$.
        Call $\col_{i}$ the resulting coloring.
    }

    Each $v$ computes $\tilde{z}_{v,t+1} \in z_{v,t+1} \pm \delta e_K$ (w.r.t.\ $\col_{t}$)
    
    Let $S_{t+1} = S_t \setminus \dom \col_{t}$ 
    be the nodes with 
    $\tilde{z}_{v,t+1} > 0.25\CCSlack \cdot \tilde{e}_K$.
    \label{line:S-t1}

    Each $v\in S_{t+1}$ runs \alg{MultiColorTrial} with $\calC(v) = [r_{v}]$. 
    Call $\col_{t+1}$ the resulting coloring.
    \label{line:reduce-mct}

    \tcp{Phase II}

    Each $v\in H \setminus \dom \col_{t+1}$ runs $O(1)$ rounds of \trycolor with $\calC(v) = [r_v]$.

    Run \multitrial with $\calC(v) = [r_v]$ for the remaining uncolored nodes.
\end{algorithm}

\cref{alg:reduce} begins by reducing the number of nodes with too few reserved colors available by trying random colors $O(1)$ times. Observe that before \cref{line:reduce-mct}, we do not use any reserved color in $S$.

\begin{lemma}
    After \cref{line:S-t1} of \cref{alg:reduce}, w.h.p., 
    for each almost-clique $|S_{t+1} \cap K| \leq e_K$.
    Moreover, prior steps did not use colors of $[r_K]$ in $K$.
\end{lemma}

\begin{proof}
    We show that when $v$ joins $S_i$, it verifies the conditions of \cref{lem:try-color} for small universal constants $\gamma = \Theta(\CCSlack)$ independent of $i$. That is, vertex $v$ has enough non-reserved colors available to get colored with constant probability by \trycolor without using reserved colors, i.e., with $\calC(v) = |L_{\col_i}(K_v) \setminus [r_v]|$. This implies the lemma because, by \cref{lem:try-color}, w.h.p., the size of sets $S_i$ shrinks by a $(1-\gamma^4/64)$-factor each iteration. Since there are only $O(e_K)$ uncolored nodes in each $K$ initially, after $O(\gamma_1^{-4}) = O(1)$ iterations, the number of nodes left is $|S_{t+1}| \leq e_K$.

    Vertices can sample a uniform color in $L(K_v) \setminus [r_v]$ using the query algorithm (\cref{lem:query}). So the first assumption of \cref{lem:try-color} holds.

    If $v\in S_i$ then $\tilde{z}_{v,i} \geq 0.25\CCSlack \cdot \tilde{e}_K$, by definition. By \cref{claim:apx-z}, the error of approximation is small, thus $z_{v,i} \geq \tilde{z}_{v,i} - \delta e_K \geq 0.01 \CCSlack \cdot e_K$ by choosing $\delta = \Theta(\CCSlack)$ small enough (\cref{eq:params}).
    \cref{lem:mu-count-high} implies there are $|L(v)\cap L(K_v) \setminus[r_v]| \geq z_v \geq 0.01 \CCSlack \cdot e_K$ available non-reserved colors.
    As $e_K \geq \lmin/2$ in non-cabals, this implies the second condition of \cref{lem:try-color}, $|\calC(v)| \gg \CCSlack^{-1} \log n$.
    Since uncolored degrees are $O(e_K)$, this also implies the forth condition of \cref{lem:try-color}, i.e., $|L(v) \cap L(K) \setminus [r_v]| \geq \Omega(\CCSlack) \cdot \deg_{\col_i}(v)$.
    
    If $|L(K)| \ge 2(r_K + 25e_K)$, 
    since $e_v \leq 25e_K$, half of the non-reserved colors in the clique palette are available because $|L(v) \cap L(K) \setminus [r_K]| \ge |L(K)| - 25e_K - r_K \ge |L(K)|/2$.
    Otherwise, if $|L(K)| < 2(r_K + 25\tilde{e}_K) = O(e_K)$, 
    a constant fraction of the non-reserved colors are available because
    $|L(v) \cap L(K) \setminus [r_K]| \ge 0.01\CCSlack \cdot e_K \ge \Omega(\CCSlack) \cdot |L(K)|$.
    This implies the third condition of \cref{lem:try-color} and concludes the proof.
\end{proof}

After \cref{line:S-t1}, we can color all nodes with $z_v \geq \Omega(e_K)$ using \multitrial with reserved colors in \cref{line:reduce-mct}. Meanwhile, the only uncolored nodes left have $z_v \leq O(e_K)$, and thus have $\Omega(e_K)$ available reserved colors and can be colored later.

\begin{lemma}
    \label{lem:color-star}
    After \cref{line:reduce-mct} in \cref{alg:reduce}, w.h.p., 
    each uncolored $v$ has $|L(v) \cap [r_v]| \ge \deg_{\col_{t+1}}(v) + 0.2\CCSlack \cdot e_K$.
\end{lemma}

\begin{proof}
    Before \cref{line:reduce-mct}, none of the colors in $[r_K]$ has been used by nodes of $K$,
    i.e., $[r_K] \cap \col_t(K) = \emptyset$.
    Uncolored vertices $v\notin S_{t+1}$ have 
    $z_{v,t+1} \leq \tilde{z}_{v,t+1}/(1-\delta) \le 0.3 \CCSlack \cdot e_K$.
    \cref{lem:mu-count-low} implies
    \begin{align*}
        |[r_v] \cap L_{\col_{t}}(v)| &\ge
        \deg_{\col_{t}}(v) + 0.5\CCSlack \cdot e_K - z_{v,t+1} \ge
        \deg_{\col_{t}}(v) + 0.2\CCSlack \cdot e_K \ .
    \end{align*}
    Note that each time a neighbor of $v$ is colored with some $c\in [r_v] \cap L_\col(v)$, 
    then both sides of this inequality decrease by one. Hence, the inequality
    holds for any extension of $\col_t$.
    
    We now argue that \alg{MultiColorTrial} with $\calC(v) = [r_v]$ colors all $v\in S_{t+1}$ 
    with high probability.
    Observe that $v$ can describe $\calC(v)$ to all its neighbors simply by broadcasting $r_v$.
    Since $\col_t(K) \cap [r_K]=\emptyset$, 
    if a reserved color $c\notin L(v)$ for some $v\in S_{t+1}$, 
    it must be because of an external neighbor.
    Hence, $|[r_v] \cap L(v)| \ge r_v - e_v \ge r_v - 21e_K \geq 3\cdot 75e_K$ 
    by our choice of $r_v$ (\cref{eq:reserved}).
    On the other hand, each $v\in S_{t+1}$ has active uncolored degree
    $\deg_{\col_{t}}(v; S_{t+1}) \le 30e_K$ (at most
    $e_K$ in $S_{t+1} \cap K_v$
    and at most 
    $25e_K$ external neighbors).
    Since $e_K \geq \lmin/2$ in non-cabals, \cref{lem:mct} applies and \multitrial colors all nodes of $S_{t+1}$ in $O(\log^*n)$ rounds
    with high probability.
\end{proof}

\begin{lemma}
    \label{lem:reduce-end}
    At the end of \alg{Complete}, w.h.p., all nodes in $V \setminus \Vcabal$ are colored.
\end{lemma}

\begin{proof}
    By \cref{lem:color-star}, after \cref{line:reduce-mct},
    the only nodes remaining to color have slack 
    $0.2\CSlack\cdot e_K$. Also recall $e_K \geq \lmin/2$ because $K\notin\Kcabal$.
    Conditions of \cref{lem:try-color} with $\calC(v) = [r_v]$
    (i.e., using reserved colors) are verified.
    Therefore, after $O(\CCSlack^{-4}\log \CCSlack^{-1}) = O(1)$ rounds of \trycolor, w.h.p., 
    uncolored nodes have uncolored degree $\le 0.1\CSlack \cdot e_K$.
    Hence, \multitrial colors all remaining nodes in $O(\log^* n)$ rounds.
    Vertices describe sets $\calC(v) = [r_v]$ to their neighbors by broadcasting $r_v$.
\end{proof}

\section{Coloring Low Degree Cluster Graphs}
\label{sec:low-deg}

In this section, we complete the result of \cref{sec:coloring-alg} with an algorithm for coloring low-degree graphs in $\poly(\log\log n)$ round. More formally, we prove
\ThmRandGeneral*

We assume in this section that $\Delta \leq \Deltalow = \poly(\log n)$. 
Our approach is slightly different depending of whether $\Delta \gg  \log n$ or not, we refer to the two cases as 
\begin{itemize}
    \item the sub-logarithmic regime when $\Delta \leq O(\log n)$, and
    \item the poly-logarithmic regime when $O(\log n) \leq \Delta \leq \Deltalow$.
\end{itemize}
In the poly-logarithmic regime, the high-level algorithm is the same as our algorithm when $\Delta \geq \Deltalow$. Importantly, we color vertices in the same order, i.e., sparse vertices first, then non-cabals, and finally cabals, but the way we color those vertices is different. In the sub-logarithmic regime, the algorithm is much simpler and deals with all vertices at the same time. The polylogarithmic-regime algorithm can be viewed as successive applications of the sub-logarithmic-regime algorithm for $\Delta \leq O(\log n)$ over different subsets of vertices.

In both cases, the main ingredients of our coloring algorithm are:
reduction of uncolored degrees down to $O(\log n)$, vertices learning at least $\deg+1$ colors in their palette, reducing uncolored parts of the graph to connected components of size $\poly \log n$ (\emph{shattering}), and an algorithm adapted from the deterministic \CONGEST model to $(\deg+1)$-list color $\poly \log n$-sized subgraphs of the virtual graph in $\poly \log \log n$ rounds. The degree reduction and shattering parts of the algorithm are minor adaptations of a classic result by Barenboim, Elkin, Pettie, and Schneider~\cite{BEPSv3}. 
The last part of this section is technically most novel:
an adaptation of an algorithm by Ghaffari and Kuhn~\cite{GK21} to color $\poly \log n$-sized subgraphs in $\poly \log \log n$ rounds, for which we once more rely on fingerprinting. We state here this result, with details \cref{sec:lowdeg-gk}.

\begin{restatable}[Ghaffari-Kuhn Algorithm in virtual graphs]{lemma}{LemLowDegGK}
    \label{lem:lowdeg-gk}
    Let $\Delta \leq \poly(\log n)$.
    Let $F$ be an $N$-vertex virtual graph with maximum degree $\Delta_F \leq O(\log n)$ on a network of bandwidth $\Theta(\log n)$. Suppose each vertex knows a list $L_v \subseteq [\Delta+1]$ of $\deg_F(v)+1$ colors. There exists a randomized algorithm that $L_v$-list-colors $F$ with probability $1 - 1/\poly(n)$ in $O(\log N \cdot \log^6 \log n)$.
\end{restatable}

The plan for this section is as follows. In \cref{sec:lowdeg-log-regime}, we explain how we perform the coloring when $\Delta \leq O(\log n)$. This will set the stage for the polylogarithmic regime that we describe next. Finally, we analyze the implementation of the Ghaffari-Kuhn algorithm behind \cref{lem:lowdeg-gk}.

\subsection{Logarithmic Regime}
\label{sec:lowdeg-log-regime}

The $\Delta \in O(\log n)$ regime is greatly simplified by the fact that in $O(1)$ rounds, each cluster can learn the colors used by its neighbors, and conversely, which colors it still has available. To do so, we simply aggregate a $O(\log n)$-bitmap in each cluster, where each bit encodes whether a given color is used. The same idea is used when $\Delta$ is larger, once we are dealing with nodes of uncolored degree $O(\log n)$.

\begin{algorithm}[H]
    \caption{High-Level Coloring Algorithm when $\Delta \in O(\log n)$\label{alg:lowdeg-log}}

    \nonl\Input{A cluster graph $H$ on $G$ such that $\Delta \le O(\log n)$}

    \nonl\Output{A $\Delta+1$-coloring} 

    \centering \nonl (Degree reduction not needed for that regime)

    \centering \nonl (Learning colors before shattering not needed for that regime)
    
    \raggedright \alg{Shattering}

    \alg{SmallInstanceColoring} \hfill (\cref{sec:lowdeg-gk})
\end{algorithm}

In all subsequent sections, we repeatedly use essentially the same algorithm as \cref{alg:lowdeg-log} to color subsets of nodes, with the small difference that we add two steps before shattering: a degree reduction step, and a step to learn a set of $\deg+1$ colors at each node.

\paragraph{Shattering}
In both this regime and for higher $\Delta$, \alg{Shattering} simply consists of each uncolored node trying a random color from its palette for $O(\log \log n)$ rounds. By an argument of Barenboim, Elkin, Pettie and Schneider~\cite[Lemma 5.3]{BEPSv3}, this results in the uncolored parts of the graph being of size $O(\Delta^2 \log_\Delta n) \leq O(\log^3 n)$. Between two random color trials, nodes update their palettes  in $O(1)$ rounds by learning which colors from their list are used by their neighbors. Since lists have size $\deg+1 \leq \Delta + 1 = O(\log n)$, this can be done by aggregating a $O(\log n)$-bitmap on their cluster as mentioned before. This allows the nodes to sample and try colors from their current palette.
In total, this step takes $O(\log \log n)$ rounds.

\paragraph{Coloring small connected components}
Finally, nodes complete the coloring using \cref{lem:lowdeg-gk}, an adaptation of an algorithm by Ghaffari and Kuhn~\cite{GK21} to our setting. Note that even in our first simple setting of $\Delta \in O(\log n)$, this is not immediate. In particular, even though $O(\log \Delta) = O(\log \log n)$ rounds suffice for each node to broadcast a $O(\log \Delta)$-bit message and receive the set of messages sent by its $\Delta = O(\log n)$ neighbors, our setting is more difficult than the Broadcast \CONGEST setting with $\log \Delta$-bit messages. This is because while each node can receive the set of messages sent by its neighbors, it cannot a priori know how many times each message was received, contrary to the Broadcast \CONGEST setting. As a result, the prior result from~\cite{FGHKN24} does not immediately apply, and we need to spend some effort adapting the Ghaffari-Kuhn algorithm to the cluster graph setting. This is done in \cref{sec:lowdeg-gk}.

\subsection{Polylogarithmic Regime}
\label{sec:lowdeg-polylog-regime}

In the polylogarithmic regime, the idea is essentially to color nodes in an order which allows us to color them exactly as in the logarithmic regime, i.e., by performing $O(\log \log n)$ random color trials and using our adaptation of the Ghaffari-Kuhn algorithm. To reduce the degree to $O(\log n)$, and ensure that each node can learn $\deg+1$ colors from its palette, we color the nodes in the same order as in the $\Delta \in \Omega(\log^{21} n)$ regime, and use the same technique as in this regime to obtain slack. For dense nodes, we also make use of the clique palette.

\begin{algorithm}[H]
    \caption{High-Level Coloring Algorithm when $\Delta \in O(\log^{21} n) \cap \Omega(\log n)$\label{alg:lowdeg-polylog}}

    \nonl\Input{A cluster graph $H$ on $G$ such that $\Delta \in O(\log^{21} n) \cap \Omega(\log n)$}

    \nonl\Output{A $\Delta+1$-coloring} 

    \computeACD  \hfill (\cref{sec:ACD})

    \slackgeneration in $V \setminus \Vcabal$ \hfill (\cref{prop:slack-generation})

    \alg{ColoringSparse}

    \nonl Let $\ell = \Theta(\log n)$ for the purpose of defining cabals (i.e., cabals have $\tilde{e}_K \in O(\log n)$)
    
    \alg{ColoringDense} in $V\setminus \Vcabal$

    \alg{ColoringDense} in $\Vcabal$
\end{algorithm}

A small difference between this algorithm, for the polylogarithmic regime, and our $O(\log^* n)$ algorithm for $\Delta \in \Omega(\log^{21} n)$, is that we define cabals as almost-cliques with $O(\log n)$ external degree, rather than $O(\log^{1.1} n)$. The reasons for this change are that $e_K \in \Omega(\log n)$ is sufficient for inliers to obtain slack from their external neighbors during \slackgeneration, and changing the threshold avoids having to give a special treatment to almost-cliques with external degree between $\Theta(\log n)$ and $\Theta(\log^{1.1} n)$. This change is possible as we do not make use of any procedure requiring $\omega(\log n)$ slack, like \alg{MultiColorTrial} (\cref{lem:mct}).

\begin{algorithm}[H]
    \caption{\alg{ColoringDense}, $\Delta \in O(\log^{21} n) \cap \Omega(\log n)$\label{alg:lowdeg-dense}}
    \alg{ColorfulMatching} 
        
    \alg{ColoringOutliers}

    \alg{ColoringInliers}     
\end{algorithm}

The three procedures
\alg{ColoringSparse}, \alg{ColoringOutliers}, and \alg{ColoringInliers} are near identical and 
very similar to the algorithm for $\Delta \in O(\log n)$ (\cref{alg:lowdeg-log}). They only differ in the way in which nodes are able to reduce their degree and learn colors prior to shattering, i.e., the implementation of the first two steps, and have the exact same last two steps. The rest of this section explains how each set of nodes implements its own version of each procedure.

\begin{algorithm}[H]
    \caption{\alg{Coloring(Sparse/Outliers/Inliers)}, $\Delta \in O(\log^{21} n) \cap \Omega(\log n)$\label{alg:lowdeg-main-subroutine}}
    \alg{DegreeReduction} through $O(\log \log n)$ uses of \trycolor \hfill (\cref{sec:appendix-try-random-color})

    \alg{LearnColors} 
    
    \alg{Shattering}

    \alg{SmallInstanceColoring} \hfill (\cref{sec:lowdeg-gk})
\end{algorithm}

\subsubsection{Sparse nodes and outliers}
\label{sec:lowdeg-sparse}

Sparse nodes have slack $\Omega(\Delta)$ w.h.p., by \cref{prop:slack-generation}, while outliers have slack $\Omega(\Delta)$ from the fact that inliers in their clique are colored later. This is both useful for reducing the degree and for finding $\deg+1$ free colors. Let us focus on the subgraph induced by sparse nodes, which is the only active part of the graph during \alg{ColoringSparse}, and which gets fully colored by this procedure. The subgraphs induced by outliers in non-cabals and in cabals
are colored similarly by \alg{ColoringOutliers}

\paragraph{Degree reduction.}
A color sampled uniformly at random in $[\Delta+1]$ by an uncolored sparse node always has a constant probability of success. This is also true for outliers. By \cref{lem:try-color}, each use of \trycolor (\cref{alg:try-random-color}) reduces the degree of nodes by a constant fraction, granted they have degree $\Omega(\log n)$. Therefore, after $O(\log \log n)$ random color trials in $[\Delta +1]$, each sparse node has at most $O(\log n)$ uncolored neighbors in the subgraph of sparse nodes (similarly in the subgraph induced by the currently considered outliers, when performing \alg{ColoringOutliers}).

\paragraph{Learning colors.}
Recall that sparse nodes have permanent slack $\Omega(\Delta)$, where $\Delta \in \Omega(\log n)$, w.h.p., and have at most $O(\log n)$ uncolored sparse neighbors after the degree reduction step. To learn colors, each sparse node spends $\Theta(\log \log n)$ rounds sampling $\Theta(\log n / \log \log n)$ colors that it has not previously sampled before, and asking its neighbors which of them are available. This has each node $v$ send $\Theta(\log n)$ colors to its neighbors, each of which has a constant probability to be free unless $(1+\Omega(1)) \deg(v)$ colors have already been discovered. This allows us to argue by Chernoff bound (\cref{lem:chernoff}) that each uncolored sparse node discovers at least $\deg(v)+1$ free colors from its palette, w.h.p. The same argument applies to outliers. Note that for both sparse nodes and outliers, the nodes have to ask their whole colored neighborhood whether some given colors are free, not just nodes of the same type.

\paragraph{Completing the coloring.}
From there, sparse nodes and outliers are colored in the same way we color nodes in the logarithmic regime (when $\Delta \in O(\log n)$). We first have them try $O(\log \log n)$ random colors from their palette, using the colors discovered in the previous step and updating their palette between two random color trials, in $O(1)$ 
rounds every time. Updating the palette in $O(1)$ rounds is possible due to all machines in each cluster knowing the $O(\log n)$ palette from which the cluster tries colors, and thus, updating the palette can be done by aggregation of a $O(\log n)$-sized bitmap indicating which colors from this palette were taken by neighboring clusters. After that, the subgraph induced by uncolored sparse nodes or the currently considered set of outliers has connected components of size at most $\poly \log n$, and we can finish the coloring using \cref{lem:lowdeg-gk}.

\subsection{Inliers}
\label{sec:lowdeg-inliers}

The treatment of inliers differs from that of sparse nodes and outliers, in that we rely on the clique palette to sample colors instead of directly sampling in $[\Delta + 1]$. We do not reserve colors, contrary to the case where $\Delta \geq \Deltalow$, which makes the properties we want from the clique palette simpler than in the main text, and more similar to \cite[Lemma 23]{FHN23}.

\paragraph{Querying The Clique Palette \& Colorful Matching.}
In the following, we will need that vertices can use the query algorithm (\cref{lem:query}) and the assumption that there is a large enough colorful matching. The query algorithm and the colorful matching algorithm when $a(K) \geq \log n$ (\cref{lem:colorful-matching-high}) only requires that $\Delta \gg \log n$. However, the colorful matching algorithm for cabals where $a(K) \leq O(\log n)$ as presented in \cref{sec:colorfulmatching} assumes that $\Delta \gg \log^2 n$, which might not hold here.

The algorithm of \cref{sec:colorfulmatching} requires $\Delta \gg \log^2 n$ only in two places: when it colors the $O(\log n)$ anti-edges and when it selects anti-edges in \cref{alg:fingerprint-matching}. Let us explain how to implement the former first. For each trial, the algoritm samples a min-wise hash function $[O(\Delta)] \to [O(\Delta)]$ that can be represented in $O(\log\Delta) = O(\log\log n)$ bits and computes the anti-neighbor with the smallest hash value. Note that we can label vertices with identifiers $\set{1, 2, \ldots, |K|}$ which take $O(\log\Delta)$ bits to describe instead of their $\Theta(\log n)$ bits identifiers. So each trial aggregates one $O(\log\Delta) = O(\log\log n)$ bit message, which requires $O(\log\log n)$ rounds to perform for all $k = O(\log n)$ trials in parallel. In fact, in this much time, all vertices of $K$ can learn all aggregates, thus computing the anti-edges does not require additional communication.

Let $F$ be the $O(\log n)$ anti-edges selected by the matching algorithm. To color them we use \trycolor and \multitrial, which requires the endpoints of each anti-edge to broadcast the same $O(\log n)$-bit message\footnote{It cannot be reduced to $O(\log \Delta)$ bits because of the use of representative sets in the implementation of \multitrial in virtual graphs.}. We select different relays for each anti-edge in $F$: a vertex $w_i$ incident to both endpoints of the $i\supth$ anti-edge in $F$. To color $F$, each relay runs \trycolor and \multitrial and anti-edges repeat the corresponding messages. We find relays using a \congest algorithm for maximal matching. To ensure it can be emulated efficiently on virtual graph, we run the \congest algorithm on a random subset of the vertices.

\begin{lemma}
\label{lem:antiedge-relays}
    Assume let $k = \Theta(\eps^{-1}\log n)$ as in \cref{alg:fingerprint-matching} and $\Delta \geq 3k$. Let $K$ be some almost-clique and $F$ a set of $\leq k$ vertex-disjoint anti-edges of $G[K]$. There is an $O(\log^4\log n)$-round algorithm that computes relays $v_1, \ldots, v_{|F|}$ such that $v_i \neq v_j$ for all $i\neq j$ and $v_i$ is adjacent to both endpoints of the $i\supth$ anti-edge in $F$.
\end{lemma}

\begin{proof}
    Sample each vertex in $K$ independently with probability $3k/\Delta$. By the Chernoff Bound, w.h.p., both endpoints of the $i\supth$ anti-edge are incident to between $k$ and $6k$ sampled vertices. By union bound, this holds for all anti-edges of $F$. We consider a bipartite virtual graph that has the anti-edges of $F$  on the left and the sampled vertices on the right. For each of the $\Theta(k)$ sampled vertices, we put an edge between it and an anti-edge if the vertex is incident to both endpoints of the anti-edge, i.e., when it can act as a relay for the anti-edge. Since each anti-edge has degree $\geq k$ and there are $\leq k$ anti-edges, each anti-edge is matched in a maximal matching. So we compute a maximal matching in the bipartite graph and let $v_i$ be the sampled vertex matched with the $i\supth$ anti-edge.

    To compute a maximal matching, we run the \congest algorithm of \cite{Fischer20}. It runs in $O(\log^2 \Delta \cdot \log N)$ rounds with $O(\log N)$ bit messages on $N$-vertices graphs. Here, vertices have $O(\log \Delta)$-bits unique identifiers so the algorithm runs in $O(\log^2\Delta \log N) = O(\log^3 \log n)$ rounds. In each round of Fischer's algorithm, a vertex may have to send $O(k\log\Delta) = O(\log n \cdot \log\log n)$ bits, since w.h.p.\ each anti-edge is incident to at most $6k$ sampled vertices. Thus, simulating one round of $O(\log \log n)$-bit communication over the virtual bipartite graph takes $O(\log\log n)$ rounds to simulate on the communication graph by having each cluster broadcast all of its sent messages.
\end{proof}

\paragraph{Non-cabals.}
In non-cabals, and more generally, when $(a_K + e_K) \in \Omega(\log n)$, inliers receive $\Omega(a_K + e_K)$ slack from \slackgeneration and \alg{ColorfulMatching}, even when using the clique palette (\cref{lem:reuse-slack}). Most importantly, for every uncolored inlier $v$, the intersection of its palette with the clique palette always contains more colors than the uncolored degree of $v$.
Focusing on the subgraph induced by inliers of non-cabals, by \cref{lem:try-color}, each use of \trycolor (\cref{alg:try-random-color}) where each active inlier uses the clique palette of its almost-clique to sample colors reduces the degree of each inlier by a constant fraction, granted it has degree $\Omega(\log n)$. As a result, after $O(\log \log n)$ of trying random colors, the graph induced by the uncolored inliers in non-cabals has maximum degree $O(\log n)$.

Finally, since inliers have slack $\Omega(a_K + e_K)$ even when using the clique palette in that regime, they can find $\deg+1$ free colors by sampling $\Theta(\log n)$ colors in the clique palette. This is similar to the way that sparse nodes and outliers find free colors, the crucial difference being that inliers sample colors from the clique palette instead of $[\Delta+1]$. From there, we have reached the point where the algorithm is the same as for sparse nodes and outliers.

\paragraph{Cabals.}
Contrary to inliers in non-cabals, however, they may not have slack. However, in cabals, every inlier has $e(v) \leq c\log n$ for some large constant $c$ (depending on $\eps$). So, as long as the clique palette contains at least $2c\log n$ colors, by \cref{lem:try-color}, trying a color from the clique palette decreases the number of vertices by a constant factor. When the clique palette contains fewer than $2c\log n$ colors, all vertices may learn the whole clique palette in $O(\log\log n)$ rounds. Since they have $O(\log n)$ external neighbors, these inliers learn exactly which colors are used by their external neighbors, in $O(\log \log n)$ rounds, similarly to the way vertices can efficiently learn all the colors used by their neighbors when $\Delta \leq O(\log n)$. So they learn their exact palette in $O(\log\log n)$ rounds. From then on, they can perform random color trials using only colors from their palette, which colors vertices with constant probability. After $O(\log\log n)$ rounds of this, each cabal contains $O(\log n)$ uncolored vertices with high probability.

Simple accounting shows that the clique palette contains enough colors for each inlier.
So, after learning the clique palette (or the $O(\log n)$ smallest colors of the clique palette) and the color of external neighbors, in $O(\log\log n)$ rounds, the vertices learn a list of $\deg_\col+1$ many colors.
We thus reach a point where inliers from cabals can be treated like all other types of vertices, and be colored by a combination of trying $O(\log \log n)$ random colors so the subgraph of uncolored vertices has connected components of size $\poly \log n$, and we color them in $\poly \log \log n$ rounds.

\subsection{Coloring the Shattered Instances} 
\label{sec:lowdeg-gk}
In this section, we prove \cref{lem:lowdeg-gk}, which gives a $\poly \log \log n$ algorithm for coloring a subgraph of a cluster graph whose connected components are of size $\poly \log n$, and whose clusters each know a list of at least $\deg +1$ free colors.

Contrary to previous approaches for post-shattering, we do not run a deterministic algorithm on shattered instances but a randomized one. Recall that the whole graph has size $n$ and shattered instances have size $N = \poly(\log n)$.
We implement the Ghaffari-Kuhn algorithm for $\deg+1$-list-coloring in \CONGEST \cite[Section 4.2]{GK20arxiv} but use the fingerprinting technique (see \cref{sec:fingerprints}) to implement its subroutines on cluster graphs. The runtime is $O(\log N) \cdot \poly(\log \Delta) = \poly \log\log n$ and the success probability $1 - 1/\poly(n)$. Let us introduce some notations and the general scheme of this algorithm.

\paragraph{Overview of the Ghaffari-Kuhn Algorithm.}
Let $\col$ be the partial coloring of $H$ and suppose each uncolored vertex knows a list $L_v \subseteq L_\col(v) \subseteq [\Delta+1]$ of $\deg_\col(v)+1$ colors (given as input). The Ghaffari-Kuhn algorithm uses local rounding to color a constant fraction of the uncolored vertices. Repeating that $O(\log N)$ times colors the whole graph.

Partition the color space $\calC \eqdef [\Delta+1]$ into $K \eqdef O(\sqrt{\log |\calC|})$ chunks $P_1, P_2, \ldots, P_K$ of equal size and use local rounding to let vertices decide on a subspace $P_i$ they restrict themselves to. Then, we recursively run the algorithm in parallel in each $P_i$. After $Q \eqdef O(\log |\calC|/\log\log|\calC|)$ levels of recursion, each part contains exactly one color and we obtain a coloring.

Let us explain how each vertex chooses its part so that the number of monochromatic edges in the end is $O(N)$. Let $\calL \eqdef [K]$ be the set of labels for the parts $P_1, \ldots, P_K$. The Ghaffari-Kuhn algorithm takes a fractional labelling assignment --- a vector $x \in [0,1]^{V_H \times \calL}$ ---
and labeling-cost\footnote{
    Notice \cref{eq:cost} does not capture the full generality of the approximate rounding lemma of \cite[Section 3]{GK20arxiv}, but this will suffice for the coloring algorithm described in \cite[Section 4.2]{GK20arxiv}.}
$y_{u,\ell} + y_{v,\ell}$ on each edge $\set{u,v} \in E_H$ where the $y \in [0,1]^{V_H \times \calL}$ are penalties on labels defined by each vertex. Each vertex $v$ knows only $x_{v,\ell}$ and $y_{v,\ell}$ for each label. The cost of the $x$ assignment, with respect the $y$-vector, is
\begin{equation}
    \label{eq:cost}
    C(x,y) \eqdef \sum_{uv \in E_H} \sum_{\ell \in \calL} x_{v,\ell} x_{u, \ell} (y_{v,\ell} + y_{u,\ell}) \ .
\end{equation}
The heart of the Ghaffari-Kuhn algorithm is an approximate rounding algorithm which takes as input the vectors $x, y \in [0,1]^{V_H \times \calL}$ and produces a vector $x' \in \set{0,1}^{V_H \times \calL}$ such that $C(x',y) \leq (1 + \eps)C(x,y)$.

\paragraph{Implementing Ghaffari-Kuhn on Cluster Graphs.}
The main technical challenge in cluster graphs is for vertices to estimate how costly it would be to increase the value on some labels. Namely, in cluster graphs, vertices cannot distinguish between being linked to one neighbor with label $\ell$ through many paths or having many neighbors with label $\ell$. To overcome this issue, we use the fingerprinting technique with weights to approximate cost (or weights) of labels over neighborhoods. \cref{lem:apx-sum} describes the types of sums we can approximate efficiently.
We use \cref{lem:apx-sum} to implement the two core routines of the Ghaffari-Kuhn algorithm: the weighted defective coloring (\cref{lem:defective-coloring}) and the approximate rounding lemma (\cref{lem:apx-weights}).
The key aspect of weights we rely on when using \cref{lem:apx-sum} is the separate contributions of $u$ and $v$ to the cost of each edge.

\newcommand{\tW}{\widetilde{W}}

\begin{definition}
    A vector $x \in [0,1]^{I}$ is $2^{-b}$-integral iff $x_i = k_i/2^b$ where $k_i\in \mathbb{N}_{\geq0}$ for all $i \in I$.
\end{definition}
\begin{lemma}
    \label{lem:apx-sum}
    Let $b \geq 0$ and 
    suppose $x \in [0,1]^{V_H}$ is a $2^{-b}$-integral vector such that $v$ knows $b$, $x_v$ and each edge $\set{u,v} \in E_H$ holds $\alpha_{u \to v}, \alpha_{v \to u} \in \set{0,1}$. There exists a randomized $O(\eps^{-2} + \frac{\log b + \log\Delta}{\log n})$-round algorithm at the end of which every vertex knows $\tW_v$ such that with probability $1 - 1/\poly(n)$, for each $v\in V_H$ we have
    \[
        \widetilde{W}_v \in (1 \pm \eps) W_v 
        \quad\text{where}\quad 
        W_v = \sum_{u\in N(v)} \alpha_{u \to v} \cdot x_u \ .
    \]
\end{lemma}

\begin{proof}
    For each $u$, let $k_u$ be the integer such that $x_u = k_u/2^b$.
    At a high level, we duplicate $k_u$ times each $u$ and use the fingerprinting algorithm with edges filtering out neighbors for which $\alpha_{u \to v} = 0$. More formally, let $t = O(\eps^{-2}\log n)$ and let $X_{v, j, i}$ be independent samples from a geometric distribution of parameter $1/2$ where $j \in [k_v]$ and $i \in [t]$. Then, for all $i \in [t]$, let
    \[ Y_i^v = \max_{\substack{u \in N(v) \\ \alpha_{u \to v} = 1}} \max_{j \in [k_u]} X_{u, j, i} \]
    which is the maximum of $\sum_{u \in N(v)} \alpha_{u \to v} k_u = 2^{b}W_v$ independent geometric variables since $\alpha_{u \to v} \in \set{0,1}$. Since all vertices know $b$, by \cref{lem:concentration-fingerprint} (with $d = 2^b W_v$), vertex $v$ computes $\tW_v$ from the values $(Y_i^v)_{i \in [t]}$ such that, with high probability, $\tW_v \in (1 \pm \eps)W_v$. This holds simultaneously at all vertices with high probability by union bound. 
    
    To implement this on a virtual graph, one designated machine in $V(u)$ samples all variables $X_{u,i,j}$ and locally computes the vector $\parens*{ \max_j X_{u,j,i} }_{i \in [t]}$, i.e., the second max in the definition of $Y_i^v$. By \cref{lem:fingerprint-encoding}, encoding these values uses $O(t + \log\log k_u) = O(\eps^{-2}\log n + \log b)$ bits since $k_u \leq 2^b$. Hence, those aggregated maximums can be distributed to all machines of $V(u)$ in $O(\eps^{-2} + \frac{\log b}{\log n})$ rounds. Since both endpoints of links between $V(u)$ and $V(u)$ know $\alpha_{u \to v}$, we can properly aggregate each $Y_i^v$ over all neighbors $u \in N(v)$ such that $\alpha_{u \to v} = 1$. By \cref{lem:fingerprint-encoding}, encoding variables $(Y_i^v)_{i \in [t]}$ (as well as all partial aggregates) can be done in $O(t + \log\log (2^bW_v))$ bits. Observe that since $x_v \leq 1$ it must be that $W_v \leq |N(v)|$; hence, pipelining requires $O(\eps^{-2} + \frac{\log b + \log |N(v)|}{\log n})$ rounds of pipelining, which concludes the proof.
\end{proof}

Let us now define what weighted defective colorings are and explain how we compute them on virtual graphs. We emphasize that \cref{lem:defective-coloring} mostly follows from previous work (e.g., \cite{FGGKR23} and references thererin) so we focus on how we use \cref{lem:apx-sum} to implement it on virtual graphs.

\begin{definition}
    Let $H$ be a graph and $w : E_H \to \mathbb{R}_{\geq 0}$ be non-negative edge weights.
    For $q \geq 1$ and $\delta > 0$, a weighted $\delta$-relative $q$-coloring of $H$ is function $\psi : V_H \to [q]$ such that 
    \[ 
        \text{for all $v\in V_H$,}\quad 
        \sum_{u \in N(v)} \I{ \psi(v) = \psi(u) } \cdot w(uv) \leq \delta \cdot \sum_{u \in N(v)} w(uv) \ . \]
\end{definition}
\begin{lemma}[Weighted Defective Coloring]
    \label{lem:defective-coloring}
    Suppose we are given a $O(\log^2 n)$-proper coloring.
    Let $\calL$ be some set of labels and
    suppose each edge $\set{u,v} \in E_H$ has weight 
    \begin{equation}
        \label{eq:def-sums-defective}
        w(uv) = \sum_{\ell \in \calL} x_{u, \ell} x_{v, \ell} (y_{v, \ell} + y_{u, \ell}) \ .
    \end{equation}
    where $x,y\in [0,1]^{V_H \times \calL}$ are $2^{-\Theta( \log(n) )}$-integral vectors.
    Then, for any $\delta > 0$, there is a randomized $O\parens*{ \frac{1}{\delta} |\calL|\log \log n \cdot \log\log\log n }$-round algorithm that returns a weighted $\delta$-relative-defective $O(1/\delta^2)$-coloring of $H$ with probability $1 - 1/\poly(n)$.
\end{lemma}

\begin{proof}
    Starting from $\psi_0$, the given $q_0$-coloring with $q_0 = O(\log^2 n)$, we compute a sequence of $\delta_i$-weighted-defective $q_i$-colorings $\psi_i$ for $i = 0, 1, \ldots, t = O(\log^* q_0)$.
    In \cite{linial92} (and also later in \cite{Kuhn2009WeakColoring,MT20,BEG18} etc...), it is shown that there is a universal constant $c \geq 1$ such that for any $s, q \in \mathbb{N}_{\geq 1}$ there exists sets $S_1, \ldots, S_q \subseteq [s^2 \tau]$ of \emph{candidate colors} such that 
    \begin{equation}
        \label{eq:linial}
        \begin{split}
        \text{for all $i\in [q]$,}\quad
        |S_i| = s\tau \quad\text{and for all $j \neq i$}\quad |S_i \cap S_j| < \tau \\
        \quad\text{where } \tau = c \cdot \min\set{\log q, \log_s^2 q} \ .
        \end{split}
    \end{equation}
    From $\psi_i$, we compute $\psi_{i+1}$ as follow. Let $S_1, \ldots, S_{q_i}$ be the sets such as described by \cref{eq:linial} with $s_i = 2^{t - i + 2}/\delta$ and $q = q_i$. For succinctness, let $S_v \eqdef S_{\psi_i(v)}$ and $N_{i,b}(v)$ is the set of $u\in N(v)$ with $\psi_i(u) \neq \psi_i(v)$. Vertex $v$ computes for each candidate color $\chi \in S_v$ the weight of its $\psi_i$-bichromatic edges to neighbors sharing candidate color $\chi$:
    \[
        W_{v,\chi} = \sum_{u \in N_{i,b}(v)} w(uv) \I{ \chi \in S_u } \ .
    \]
    Using \cref{eq:linial}, one can show that there is a color $\chi \in S_v$ for which $W_{v,\chi} \leq W_v/s_i$ where $W_v = \sum_{u\in N(v)} w(uv)$.
    By having $v$ choose a color $\psi_{i+1}(v) = \chi \in S_v$ such that $W_{v,\chi}$ is within a factor 2 from the minimum $W_{v,\chi'}$ for $\chi' \in S_v$, we therefore increase the defect at most by an additive $2W_v/s_i = \delta W_v /2^{t-i+1}$ term. After $t$ iterations, we obtain a $O(1/\delta^2)$-coloring $\psi_{t+1}$ with defect of $(2/s_0 + 2/s_1 + \ldots + 2/s_t) W_v \leq \delta W_v$. We refer the readers to \cite[Appendix A]{KS18} or \cite[Theorem 4.9]{Kuhn2009WeakColoring} for more details.
    
    We now explain how each vertex computes $\tW_{v,\chi} \in (1 \pm 0.5)W_{v,\chi}$ for all $\chi \in S_v$. Simple calculation shows that one can decompose each $W_{v,\chi}$ as
    \begin{align*}
        W_{v,\chi} &= 
        \sum_{\ell \in \calL}
            x_{v,\ell}y_{v,\ell} W_{v,\chi,\ell}^{(1)} +
            x_{v,\ell} W_{v,\chi,\ell}^{(2)} 
        \shortintertext{where for each $\ell \in \calL$ we define} 
        W_{v,\chi,\ell}^{(1)} &= \sum_{u \in N_{{i,b}}(v)} \I{\chi \in S_u} x_{u,\ell} 
        \quad\text{and}\quad
        W_{v,\chi,\ell}^{(2)} = \sum_{u \in N_{{i,b}}(v)} \I{\chi \in S_u} x_{u,\ell}y_{u,\ell} \ .
    \end{align*}
    We can assume without generality that all machines of $V(v)$ know all values $x_{v,\ell}$ and $y_{v,\ell}$ as broadcasting them requires $O(|\calL|) = O(|\calL|)$ rounds (they each take $O(\log n)$ bits to describe from our integrality assumption). Also note that a machine in $V(v) \cap V(u)$ knows all candidate colors of $u$ and $v$, hence can locally compute $\I{ \chi \in S_u }$. 
    Hence, to approximate $W_{v,\chi}$, it suffices that $v$ approximates each $W_{v,\chi,\ell}^{(1)}$ and $W_{v,\chi,\ell}^{(2)}$ for all $\ell \in \calL$. 
    Notice that any fixed sum $W_{v,\chi,\ell}^{(1)}$ and $W_{v,\chi,\ell}^{(2)}$ can be computed w.h.p., at all vertices with candidate color $\chi$ in $O(1)$ rounds, by \cref{lem:apx-sum} (with $\eps=1/2$ and $\alpha_{u \to v} = \I{\chi \in S_u}$ in the instance of $W_{v,\chi,\ell}$ where we ignore all $\psi_i$-monochromatic edges). Hence, if vertices run this algorithm in parallel for all $\chi \in S_v$ and $\ell \in \calL$, through basic pipelining, one step of color reduction runs in $O(|\calL|s_i\tau_i)$ rounds as $|S_v| = s_i\tau_i$.
    One can show by induction that $\log_{s_i} q_i \leq 4\log^{(i+1)} q_0$ (see, e.g., \cite[Proof of Theorem 2]{KS18}) so that for all $i \in [0, t]$,
    \[
        s_{i+1} \tau_{i+1} \leq 2^{t+2}/\delta \cdot 16c \parens*{ \log^{(i+1)} q_0 }^2 \ .
    \]
    So computing $\psi_{i+1}$ from $\psi_i$ for $i=1, 2, \ldots, t$ takes $O(|\calL|) \cdot 2^{O(\log^* n)}/\delta \cdot O(\log^2\log\log n)= O( \frac{1}{\delta}|\calL| \cdot \log^3\log\log n)$ rounds. When we compute $\psi_1$ from $\psi_0$, note that $s_0 \tau_0 \leq 2^{t+2}/\delta \cdot c\log q_0$ so it takes $O(\frac{1}{\delta}|\calL| \cdot \log\log n \cdot \log\log\log n)$ rounds.
    Over the $t=O( \log^* n )$ iterations of color reduction, the cost of the first step dominates and we obtain the claimed runtime.
\end{proof}

\begin{lemma}[Approximate Rounding]
    \label{lem:apx-weights}
    Let $\eps \in (0,1)$, $b \in [1, \Theta(\log n)]$ be an integer and $\calL$ a set of labels with $2^{-b}$-integral vectors $x, y \in [0,1]^{V_H \times \calL}$ and we are given a proper $O(\log^2 n)$-coloring. There is a randomized algorithm that computes a $2^{-b+1}$-integral $x' \in [0,1]^{V_H \times \calL}$ such that, w.h.p., $C(x',y) \leq (1 + \eps)C(x, y)$. It ends after 
    \[
      O\parens*{ \frac{|\calL|\log\log n \cdot \log\log\log n}{\eps} + \frac{|\calL|}{\eps^4} }  
      \quad\text{rounds.}
    \]
\end{lemma}

\newcommand{\ov}[1]{\overline{#1}}

\begin{proof}
    Recall the main steps of approximate rounding in \cite[Section 3]{GK21}. First, compute a $\eps/8$-defective $O(1/\eps^2)$-coloring using \cref{lem:defective-coloring}. Then, go sequentially over color classes $i=1, 2, \ldots, O(1/\eps^2)$ of the defective coloring\footnote{
        In \cite{GK21}, the authors quadratically improve the runtime by using an \emph{average} weighted $\eps$-relative defective $O(1/\eps)$-coloring. Here, in the interest of simplicity, we compute only a weighted defective coloring.} 
    and, each time, vertices $v$ of the $i\supth$ color class decide on two disjoint sets $\calL^-_v, \calL^+_v \subseteq \ov{\calL_{v}}$ of equal size where labels $\ell \in \ov{\calL_v}$ are those for which $x_{v,\ell}$ is not a multiple of $2^{-b+1}$. Then, all $x_{v,\ell}$ for $\ell \in \calL^-_v$ are updated to $x_{v,\ell} - 2^{-b}$ and all $x_{v,\ell}$ for $\ell \in \calL^+_v$ are increased to $x_{v,\ell} + 2^{-b}$. It is easy to verify that the resulting $x$-vector is $2^{-b+1}$-integral.

    To decide on $\calL_v^{-},\calL_v^+$, for each $\ell \in \ov{\calL_v}$, vertices compute the sums
\[
        W_{v,\ell} = \sum_{u \in N(v)} x_{u,\ell} (y_{u,\ell} + y_{v,\ell})
    \]
    As shown in \cite[Lemmas 3.4 and 3.5]{GK21}, 
    \[
        C(x',y) - C(x,y) \leq 
        \sum_v \parens*{ \sum_{\ell \in \calL_v^+} W_{v,\ell} - \sum_{\ell \in \calL^-_v } W_{v,\ell} } 
        + \eps/2 \cdot C(x,y)\ .
    \]
    We compute approximations $\tW_{v,\ell} \in (1 \pm \eps)W_{v,\ell}$ for each $\ell \in \ov{\calL_v}$ as follow. Using the definition of weights, simple calculation shows that $W_{v,\ell}$ can be decomposed as
    \[
        W_{v,\ell} = y_{v,\ell} W_{v,\ell}^{(1)} + W_{v,\ell}^{(2)}
        \quad\text{where}\quad
        W_{v,\ell}^{(1)} = \sum_{u\in N(v)} x_{u,\ell} 
        \quad\text{and}\quad
        W_{v,\ell}^{(2)} = \sum_{u\in N(v)} x_{u,\ell} y_{u,\ell} \ .
    \]
    By \cref{lem:apx-sum}, these sums can be approximated, w.h.p., up to multiplicative error $1+\eps/4$ in $O(|\calL|/\eps^2)$-rounds. Then, put the $|\ov{\calL_v}|/2$ labels with the largest approximate weights in $\calL_v^-$ and the rest in $\calL_v^+$. This implies that
    \[
        \sum_{\ell \in \calL^-} \tW_{v,\ell} \geq \sum_{\ell \in \calL^+} \tW_{v,\ell} 
        \quad\text{and hence}\quad
        \sum_{\ell \in \calL^-} W_{v,\ell} 
        \geq \frac{1 - \eps/4}{1 + \eps/4} \sum_{\ell \in \calL^+} W_{v,\ell} \geq
        (1-\eps/2) \sum_{\ell \in \calL^+} W_{v,\ell} \ .
    \]
    In particular, this implies that $C(x', y) \leq (1 + \eps)C(x, y)$.

    To conclude the proof, we consider the round complexity. We compute the weighted-defective coloring once at the beginning in the runtime given by \cref{lem:defective-coloring}. Then, for $O(1/\eps^2)$ iterations, we need $O(|\calL|/\eps^2)$ rounds to approximate weights, which results in the claimed runtime.
\end{proof}

We are now ready to implement the Ghaffari-Kuhn algorithm on the post-shattered instances.

\LemLowDegGK*

\begin{proof}

    Since $\Delta_F = O(\log n)$, we can properly $O(\log^2 n)$-color $F$ in $O(1)$ rounds with probability $1 - 1/\poly(n)$ \cite[Theorem 6.1]{HN23}.
    When the color space is partitioned into $P_1, \ldots, P_K$, the Ghaffari-Kuhn algorithm turns the fractional labelling $x_{v,\ell} = \frac{|L_v \cap P_\ell|}{|P_\ell|}$ into a $2^{-b}$-integral $x'$ where $b = O(\log(KQ))$. Since $y_{v,\ell} = \frac{1}{|L_v \cap P_\ell|}$ it can be made $2^{-b}$-integral without increasing the cost by more than constant factor. By applying $b$ times \cref{lem:apx-weights} times with $\eps = \Theta(1/Qb)$, we obtain an integral $x''$ for which the cost has increased by a $(1 + O(1/Q))$ factor. Hence, after the $Q$ levels of recursion, the total cost has increased by a constant factor. As explained in \cite{GK21}, it induces a coloring of the vertices with $O(n)$ monochromatic edges. Considering the set of vertices incident to $O(1)$ monochromatic edges and running Linial on this $O(1)$-degree graph, we color a constant fraction of the vertices in $O(\log^* \Delta)$ rounds. The total round complexity is
    \[
        Qb \cdot K \cdot \parens*{ \frac{\log\log n \cdot\log\log\log n }{\eps} +  \frac{1}{\eps^4} }
        = O( \log^{5.5} \log n ) 
    \]
    rounds, using that $Q = O(\log \Delta/\log\log \Delta)$, $K = O(\sqrt{\log\Delta})$, so $Qb = O(\log\Delta)$ and $\Delta \leq \poly(\log n)$. Since we repeat this algorithm $O(\log N)$ times to color all vertices, the total round complexity is $O(\log N \cdot \log^6 \log n)$. Note that since $\Delta_F \leq O(\log n)$, it takes $O(1)$ rounds to update the lists $L_v$ of available colors between each iteration.
\end{proof} 

\newpage
\appendix

\section{Related Work on Cluster Graphs \& Virtual Graphs}
\label{sec:related-work-cluster-graph}

\subsection{Cluster Graphs in the Distributed Graph Literature}

\paragraph{The Laplacian Framework.}
Algorithms for maximum flow are heavily based on sparsification techniques (e.g., j-trees, spectral and ultra sparsifiers, low stretch spanning tree, congestion approximator). When brought to distributed models, e.g., by \cite{GKKLP18,FGLPSY21}, the sparsified graphs are embedded as cluster graphs, with clusters possibly overlapping. Although our framework assumes clusters are disjoint, it is easy to see that by paying overhead proportional to edge congestion, our algorithms can be implemented when clusters overlap.

\paragraph{Network Decomposition.}
The recent network decomposition algorithms of \cite{RG20,GGR20,GHIR23} are based on a pruning approach. They repeat $O(\log n)$ times a sub-routine growing low-diameter clusters covering at least half of the vertices. To that end, and especially when it comes to \congest algorithms \cite{GHIR23}, their algorithms are performing computation on clusterings. In our words, provided some cluster graph $H$, they compute a nicer cluster graph $H'$ while handling bandwidth-efficient communication on the communication network.

\paragraph{Power Graphs.}
Recently, there has been a growing interest in coloring $H = G^k$ \cite{HKM20,HKMN20,FHN23,BG23} and on related problems on power graphs \cite{BCMPP20,MPU23}. Proper coloring means then that any pair of vertices of distance at most $k$ in $G$ must receive different colors. For instance, the distance-2 coloring problem models frequency allocation in wireless models. It has also been used in algorithms for the Lov\'asz Local Lemma \cite{FG17}. The current randomized complexity for distance-2 coloring in \congest is $O(\log^6\log n)$ \cite{FHN23}.
Power graphs do not fit in our definition of cluster graphs since we assume clusters are disjoint, but in the sibling paper \cite{us:partii}, we describe how cluster graphs can be generalized to capture power graphs.

\paragraph{Local Rounding.}
In the general local rounding technique introduced in \cite{FGGKR23}, vertices slowly update a fractional solution (e.g., a coloring, a matching, an independent set) while
giving up
a small factor in the objective. A core subroutine in this process is computing a distance-2 defective coloring (defective means each vertex can be adjacent to few vertices of its color). To that end, they introduce distance-2 multigraphs that are cluster graphs with congestion. Applications of this general technique \cite{GG23,GHIR23} also have to deal with those multigraphs as well.

\subsection{Virtual Graphs}
\label{ssec:virtual-graphs}

In the second part of this paper, we formally define the notion of virtual graph. Informally speaking, a virtual graph is a cluster graph where clusters are allowed to overlap. Results of this paper naturally extend to virtual graphs. We provide definitions for completeness.

\begin{definition}[Virtual Graph]
\label{def:virtual}
Let $G=(V_G, E_G)$ be a \emph{simple} graph. 
A virtual graph on $G$ is a \emph{multi}-graph $H=(V_H, E_H)$ where each vertex $v\in V_H$ is mapped to a set $V(v) \subseteq V_G$ of machines called the \emphdef{support} of $v$. Whenever two nodes are adjacent in $H$ their supports intersect, i.e., if $E_H(u,v) \neq \emptyset$ then $V(u) \cap V(v) \ne \emptyset$.
Each machine $w \in V_G$ knows the set $V^{-1}(w)$ of vertices whose supports contains it.
\end{definition}

When bandwidth is not an issue, we can work directly with the representation of \cref{def:virtual}. We can compute a breadth-first spanning tree $T(v) \subseteq E_G$ on each support $V(v)$ for distributing information, and then simulate a local algorithm on this support structure. 
With bandwidth constraints, we need to be more careful.

\begin{definition}[Embedded Virtual Graph]
\label{def:embedded}
Let $H$ be a virtual graph on $G$ such that $|V_H| \le \poly(|V_G|)$.
Suppose that (1) for each vertex $v\in V_H$, there is a tree $T(v) \subseteq E_G$ spanning $V(v)$; and (2) for each edge $e\in E_H(u,v)$ there is a machine $m(e) \in V(u) \cap V(v)$. 
We call $T(v)$ the \emphdef{support tree} of $v$ and $m(e)$ the machine \emphdef{handling} edge $e$.
Each machine $w$ knows the set of edges $m^{-1}(w)$ it handles as well as, for each incident link $\set{w,w'} \in E_G$, the set $T^{-1}(ww')$ of support trees it belongs to.
\end{definition}

For a concrete example, the distance-$2$ coloring problem is captured by introducing a virtual node $v \in V_H$ of support $V(v) = N_G(v)$ for each original node $v \in V_G$. For each support $V(v)$, we take the natural support tree of a star centered at $v$. An edge between two nodes $v$ and $v'$ at distance $2$ in the original graph is handled by the interior nodes of the paths of length $2$ connecting $v$ and $v'$.

Given support trees, it is convenient to design our algorithms as a sequence of rounds each consisting of a three-step process: broadcast, local computation on edges, followed by converge-cast.
We use two parameters to quantify the overhead cost of aggregation on support trees.
The \emphdef{congestion} $\congestion$ of $H$ is the maximum number of trees using the same link. The \emphdef{dilation} $\dilation$ is the maximum height of a tree $T(v)$ in $G$.
Formally,
\begin{equation}
\label{eq:def-congestion-dilation}
   \congestion \eqdef \max_{e\in E_G} |T^{-1}(e)|
\quad\text{and}\quad
   \dilation \eqdef \max_{v\in V_H} \parens*{\max_{u,u'\in T(v)}\operatorname{dist}_{T(v)}(u, u')} \ .
\end{equation}
Coming back to our distance-$2$ coloring example, congestion and dilation are both $2$ for this particular problem with the embedding specified earlier.
Congestion and dilation are natural bottlenecks for virtual graphs. 
\smallskip

\begin{informalbox}
    \medskip
    
    Informally speaking, our results in \cite{us:partii} are
    \begin{enumerate}
        \item Any coloring algorithm requires $\Omega(\congestion + \dilation \log^*n)$ rounds.
        \item There is an algorithm to $\deg+1$-color the virtual graphs in $O(\congestion\dilation\poly\log\log n)$ rounds.
    \end{enumerate}
\end{informalbox}

We emphasize that the upper bound of \cref{thm:rand-general,thm:high-degree} are incomparable to that of \cite{us:partii} because, in this paper, we use $\Delta+1$ colors where $\Delta$ is the true maximum degree of $H$ while degrees in \cite{us:partii} account for double redundant paths in the communication network.

\section{Concentration Bounds}
\label{app:concentration}
\begin{lemma}[Chernoff bounds]\label{lem:basicchernoff}
Let $\{X_i\}_{i=1}^r$ be a family of independent binary random variables variables, and let $X=\sum_i X_i$.
Suppose $\mu_{\mathsf{L}} \leq \Exp[X] \leq \mu_{\mathsf{H}}$, then
\begin{align}
\forall t > 0,
\qquad
\Pr\range*{X < \mu_{\mathsf{L}} - t},~\Pr\range*{X > \mu_{\mathsf{H}} + t}
& \le \exp\parens*{-2t^2/r}
\ .\label{eq:chernoffadditive}
\end{align}
\end{lemma}

We use the following variants of Chernoff bounds for dependent random variables. The first one is obtained, e.g., as a corollary of Lemma 1.8.7 and Thms.\ 1.10.1 and 1.10.5 in~\cite{Doerr2020}.

\begin{lemma}[Martingales \cite{Doerr2020}]\label{lem:chernoff}
Let $\{X_i\}_{i=1}^r$ be binary random variables, and $X=\sum_i X_i$.
If $\Pr \range*{ X_i=1\mid X_1=x_1,\dots,X_{i-1}=x_{i-1} }\le q_i\le 1$, for all $i\in [r]$ and $x_1,\dots,x_{i-1}\in \{0,1\}$ with $\Pr[X_1=x_1,\dots,X_r=x_{i-1}]>0$, then for any $\delta>0$,
\[\Pr\range*{X\ge(1+\delta)\sum_{i=1}^r q_i}\le \exp\parens*{-\frac{\min(\delta,\delta^2)}{3}\sum_{i=1}^r q_i}\ .\]
If $\Pr[X_i=1\mid X_1=x_1,\dots,X_{i-1}=x_{i-1}]\ge q_i$, $q_i\in (0,1)$, for all $i\in [r]$ and $x_1,\dots,x_{i-1}\in \{0,1\}$ with $\Pr[X_1=x_1,\dots,X_r=x_{i-1}]>0$, then for any $\delta\in [0,1]$,
    \begin{equation}\label{eq:chernoffmore}
    \Pr\range*{ X\le(1-\delta)\sum_{i=1}^r q_i }\le \exp\parens*{ -\frac{\delta^2}{2}\sum_{i=1}^r q_i }\ .
    \end{equation}
\end{lemma}

\begin{lemma}[read-$k$ bound, \cite{kread}]
\label{lem:k-read}
Let $k > 0$ be an integer and $X_1, \ldots, X_r$ be independent binary random variables. We say binary random variables $Y_1, \ldots, Y_m$ are a read-$k$ family if there are sets $P_i \subseteq [r]$ for each $i\in[m]$ such that: (1) for each $i\in[m]$ variable $Y_i$ is a function of $\set{X_j: j\in P_i}$ and (2) each $j\in[r]$ verifies $|\set{i: j\in P_j}| \le k$. In word, each $X_j$ influences at most $k$ variables $Y_i$. 
If $\set{Y_i}$ is a read-$k$ family, then, for any $\delta > 0$,
\[
  \Pr\range*{\abs*{\sum_j Y_j - \Exp\range*{\sum_j Y_j}}} \le 2\exp\parens*{-\frac{2\delta^2r}{k}} \ .
\]
\end{lemma} 
\section{Pseudo-Random Tools}
\label{app:hashing}
\begin{definition}[Min-wise independence]
    \label{def:min-wise}
    A family of functions $\calH$ from $[n]$ to $[n]$ is said to be $(\eps,s)$-min-wise independent if for any $X \subseteq [n]$, $\card{X} \leq s$ and $x \in [n] \setminus X$ we have:
    \[
    \abs*{
    \Pr_{h \in \calH}[h(x) < \min h(X)] - \frac{1}{\card{X}+1} 
    } \leq \frac{\eps}{\card{X}+1}\ .
    \]
\end{definition}

\begin{lemma}[\cite{Indyk01}]
    \label{thm:min-wise}
    There exists some constant $C_{\ref{thm:min-wise}} > 0$ such that for any $N \geq 2$, $\epsilon > 0$ and $s \le \epsilon N/C_{\ref{thm:min-wise}}$ any $O(\log 1/\epsilon)$-wise independent family of functions $[N] \to [N]$ is $(\epsilon,s)$-min-wise independent. In particular, a function in such a family can be described using $O(\log N \cdot \log1/\epsilon)$ bits.
\end{lemma}

\begin{definition}[Almost Pairwise independence]
    Let $\epsilon > 0$. A family of functions $\calH$ from $[N]$ to $[M]$ is said to be $\epsilon$-almost-pairwise independent if for every $x_1\neq x_2 \in [N]$ and $y_1, y_2\in [M]$, we have 
    \[
        \Pr_{h\in \calH}[h(x_1) = y_1\text{ and } h(x_2) = y_2] \le \frac{1+\epsilon}{M^2} \ .
    \]
\end{definition}

\begin{theorem}
    \label{thm:pwi-indep}
    For any $\epsilon > 0$, there exists an explicit $\epsilon$-almost-pairwise independent family $\calH$ of functions $[N] \to [M]$ such that describing $h\in \calH$ requires $O(\log\log N + \log M + \log 1/\epsilon)$ bits.
\end{theorem}

\begin{definition}[Representative Sets, \cite{HN23}]
    \label{def:rep-sets}
    Let $\calU$ be some universe of size $k$. A family $\calF = \set{S_1, S_2, \ldots, S_t}$ of $s$-sized sets of $\calU$ is said to be $(\alpha, \delta, \nu)$-representative iff 
    \begin{equation}
        \label{eq:rep-sets-large}
        \forall T \subseteq \calU, \text{ s.t. } |T| \ge \delta k: 
        \qquad 
        \Pr_{i\in [t]} \parens*{\abs*{\frac{|S_i \cap T|}{|S_i|} - \frac{|T|}{k}} \le \alpha\frac{T}{k}} \ge 1-\nu \ ,
    \end{equation}
    \begin{equation}
        \label{eq:rep-sets-small}
        \forall T \subseteq \calU, \text{ s.t. } |T| < \delta k: 
        \qquad 
        \Pr_{i\in [t]} \parens*{\frac{|S_i \cap T|}{|S_i|} \le (1+\alpha)\delta} \ge 1-\nu \ .
    \end{equation}
\end{definition}

\begin{lemma}[\cite{HN23}]
    \label{lem:rep-set-exists}
    Let $\calU$ be a universe of size $k$. For any $\alpha, \delta, \nu > 0$, there exists an $(\alpha,\delta,\nu)$-representative family $\calF$ containing $t \in \Theta(k/\nu + k\log(k))$ sets, each of size $s \in \Theta(\alpha^{-2}\delta^{-1}\log(1/\nu))$.
\end{lemma}

\section{Deferred Proofs}
\label{sec:omitted-proofs}
\subsection{Breadth-First Search}
\label{sec:proof-bfs}

\LemBFS*
\begin{proof}
    We define timestamps $\tau_v$, initially $+\infty$ for all vertices except for source vertices: $\tau_{s_i} = 0$ for each $i \in [k]$. The algorithm repeats for $t$ iteration: each vertex $u\in V_{H_i}$ such that $\tau_u < +\infty$ broadcasts $(ID_{s_i}, ID_u, \tau_u+1)$. When a vertex $v\in V_{H_i}$ receives a message $(ID_{s_i}, ID_u, \tau_u)$ where $\tau_u + 1 < \tau_v$ from a neighbor $u\in V_{H_i}$, it sets $\tau_v$ to $\tau_u+1$, making the node with identifier $ID_u$ its parent in the tree $T_{H,i}$. Timestamps are introduced to handle the delays incurred by communicating in $G$.
    
    We now detail the implementation on $G$.
    Focus on some $v \in H_i$ and suppose a machine in $w\in T(v)$ receives a message of the form $(ID_{s_i}, ID_u, \tau_u)$ from a machine $w'\in T(u)$ such that $u \in H_i$. We say that vertex $u$ was the emitter and $w$ the receiver. Multiple emitters can reach different machines in $T(v)$ at concurrent times. Each receiver $w\in T(v)$ crafts a message $m_w = (RECEIVE, \tau_u, ID_u, ID_w)$ where $u \in H_i$ is the emitter. We aggregate on $T(v)$ the minimum of these messages according to the lexicographical order. Note that it suffices to learn the message with the minimum timestamp. Moreover, the aggregation produces a \emph{unique receiver identifier $ID_w$}. When vertex $v$ receives a message $m_w$ with $\tau_u + 1 < \tau_v$, it updates $\tau_v = \tau_u + 1$ and broadcasts $(UPDATE, \tau_v, ID_u, ID_w)$ to all machines in $T(v)$. When machines of $T(v)$ receive this message, if they are adjacent to a cluster $u\in V_{H_i} \setminus \set{v}$, they emit the message $(ID_{s_i}, ID_v, \tau_v)$ to the neighboring machine in $V(u)$.

    Since support trees $T(v)$ have diameter $\dilation = O(1)$, after $O(t'\dilation) = O(t')$ rounds in $G$ for all $t' \leq t$, all vertices at distance $\leq t'$ from $s_i$ in $H_i$ identified a unique receiver in their cluster. We orient all edges of $T(v)$ toward that \emph{selected} receiver. Taking all those directed trees together with only the inter-cluster edges connecting the \emph{selected} receivers with their emitters, we obtain the tree $T_G$.
\end{proof}

\subsection{Prefix Sums}
\label{sec:proof-prefix-sum}

\LemPrefixSums*

\begin{proof}
We first describe the algorithm within a single tree $T$ before arguing that it can be executed in parallel in multiple edge-disjoint trees.
    
    Let $r$ be the root of $T$.
    By convergecast, in $O(d)$ rounds,
    each machine computes the sum of the $x_u$ over machines $u$ in its subtree.
    In particular, the root knows the sum $\sum_{u\in S} x_u$, 
    which it can broadcast to all machines in $O(d)$ rounds.
    To compute partial sums, we repeat the following inductive process starting with $w = r$.
    Let $w\in V_T$.
    If its subtree contains no nodes of $S$, it has nothing to do.
    Otherwise, let $u_1 \preceq v_1 \prec u_2 \preceq v_2 \prec \ldots \prec u_d \preceq v_d$ be machines of $S$
    such that all $u \in [u_i, v_i]$ belong to the $i\supth$ subtree of $w$. Call $w_i$ the $i\supth$ children of $w$.
    By induction on the height of $T$, suppose that $w$ knows the sum 
    $S_{\prec w} \eqdef \sum_{u \prec w} x_u$.
For the root, this sum is zero because $r$ is the minimum for $\prec$, thus unique.
    Since, during the converge cast, $w$ received each sum 
    $S_{[u_i, v_i]} \eqdef \sum_{u_i \preceq u \preceq v_i} x_u$, 
    it can send to $w_i$, its $i\supth$ children
    $S_{<w_i} = S_{<w} + \sum_{j < i} S_{[u_j, v_j]} = \sum_{u \prec u_i} x_u$, x
    which concludes the proof.

    Note that an execution of the algorithm in $T$ only has machines communicate over edges of $T$. As a result, when running the algorithm in parallel over multiple edge-disjoint trees, at no point do two distinct executions attempt to send a message over the same edge. Therefore, the algorithm can be performed on multiple edge-disjoint trees in parallel in the same $O(d)$ runtime.
\end{proof}

\subsection{MultiColorTrial}
\label{sec:appendix-mct}

\begin{lemma}[MultiColorTrial, adapted from~\cite{HN23}]
    \label{lem:mct}
Let $\col$ be a (partial) coloring of $H$, $H'$ be an induced subgraph of $H\setminus \dom\col$, and let $\mathcal{C}(v) \subseteq [\Delta+1]$ be a reduced color space for each node. Suppose that there exists some constant $\gamma > 0$ known to all nodes such that
    \begin{enumerate}
        \item $\mathcal{C}(v)$ is known to all machines in and adjacent to $V(v)$; and
        \item\label[cond]{part:mct-slack} 
        $|L_\col(v) \cap \calC(v)| - \deg_\col(v; H') \ge \max\set{2\cdot\deg_\col(v; H'), \Theta(\log^{1.1} n)} + \gamma \card{\calC(v)}$.
    \end{enumerate}
    Then, there exists an algorithm computing a coloring $\psi \succeq \col$ such that, w.h.p., all nodes of $H'$ are colored. The algorithm runs in $O(\gamma^{-1} \log^* n)$ rounds.
\end{lemma}

\cref{lem:mct} is achieved through a repeated application of a procedure trying increasingly more colors (analog to Algorithm 11 in~\cite{HKNT22}). We follow Algorithm 10 in the arXiv version of \cite{HKNT22} where each call to \alg{MultiTrial} (Algorithm 11 in \cite{HKNT22}) is replaced by $O(\gamma^{-1})$ calls to \trymulticolor. To obtain the guarantees of \cref{lem:mct}, analog to those of \cite[Lemma 1]{HKNT22}, we need only to prove an analog of \cite[Lemma 25]{HKNT22} for \trymulticolor.

The difference between our setting and that of \cite{HKNT22} is that vertices cannot sample directly in their palettes. Instead, we use the approach of \cite{HN23}, using representative sets (see \cref{def:rep-sets}) to sample in a known universe of colors $\calC(v)$. In \cite{HN23}, authors provide a bandwidth-efficient implementation of $O(\Delta)$-coloring in $O(\log^* n)$ rounds. Hence, their result needs some adaptation which we sketch now.

\begin{algorithm}
    \caption{\trymulticolor$(x)$\label{alg:try-multi-color}}
    \nonl\Input{A coloring $\col$, a parameter $\gamma \in (0,1)$, a set $S \subseteq V \setminus \dom\col$ and $\calC(v) \subseteq [\Delta+1]$ for each $v\in S$}

    \nonl\Output{A coloring $\col' \succeq \col$}

    \nonl For each $\calC \subseteq [\Delta+1]$, $\gamma \in (0,1)$, let $\calF_{\calC,\gamma}$ be a globally known representative set family over $\calC$ of parameters $(\frac 1 2, \frac \gamma 2,\frac 1 {\poly(n)})$ and size $O(\gamma^{-1}\poly(n))$ (\cref{lem:rep-set-exists}).
    
    Each $v\in S$ samples $Y(v)$ uniformly at random in $\calF_{\calC(v),\gamma}$. 

    Each $v \in S$ samples $X(v) \subseteq Y(v)$ of size $\card{X(v)} = x$ uniformly at random.

    If $\exists c(v) \in X(v) \cap L(v)$ such that $c(v) \not \in X(N(v))$, then adopt one such $c(v)$.
\end{algorithm}

\begin{lemma}[Adapted from \cite{HN23,HKNT22}]
    \label{lem:try-multi-color}
    Let $x$ be a positive integer, $\gamma>0$ a constant, and suppose $\card{L(v) \cap \calC(v)} \geq x \deg_\col(v;H') + \gamma \card{\calC(v)}$. Then, \trymulticolor$(x)$ colors $v$ with a failure probability of at most $e^{-\gamma x /2} + 1 / \poly(n)$.
\end{lemma}
\begin{proof}
    For intuition, we first consider the case that nodes make $x$ fully random samples in $\calC(v)$. For a given node $v$, let $Z(v) = \bigcup_{u\in N_{H'}(v)} X(v)$ be the set of colors tried by neighbors of $v$. We have $\card{Z(v)} \leq x \deg_\col(v;H')$. Consider now the set $(L(v) \cap \calC(v)) \setminus Z(v)$, i.e., the set of colors in $\calC(v)$ that are neither already used by a colored neighbor of $v$ nor currently tried by an uncolored neighbor of $v$. From our Lemma's hypothesis on the size of $\card{L(v) \cap \calC(v)}$, we have:
    \[
    \card{(L(v) \cap \calC(v)) \setminus Z(v)} \geq \card{L(v) \cap \calC(v)} - x \deg_\col(v;H') \geq \gamma \card{\calC(v)}
    \ .\]
    Hence, every color that $v$ samples uniformly at random in $\calC(v)$ has an independent $\gamma$ probability of being neither tried by one of its neighbors nor already used. For $v$ to remain uncolored, all of its color trials must fail, which occurs with probability at most $(1-\gamma)^x \leq e^{-\gamma x}$.

    Consider now that the $x$ independent color trials in $\calC(v)$ are instead sampled pseudorandomly, as in \cref{{alg:try-multi-color}}: $v$ samples uniformly at random a set $X(v)$ from a $(1/2,\gamma/2,1/\poly(n))$-representative set family over $\calC(v)$, and picks $x$ values uniformly at random in $X(v)$. From the properties of representative set families, a at least a $\gamma/2$-fraction of $X(v)$ is contained in $(L(v) \cap \calC(v)) \setminus Z(v)$, w.h.p. Conditioned on that high probability event, the $x$ random colors chosen by $v$ in $X(v)$ have a probability at most $(1-\gamma/2)^x \leq e^{-\gamma x /2}$ to all be outside $(L(v) \cap \calC(v)) \setminus Z(v)$, i.e., $v$ fails to color itself with probability at most $e^{-\gamma x /2} + 1/\poly(n)$.
\end{proof}

\begin{proof}[Proof Sketch of \cref{lem:mct}]
    Run Algorithm 10 in \cite{HKNT22} with $s_{\min} = \log^{1.1} n$, $\kappa = 1/10$ and replacing all calls to \alg{MultiTrial}$(x)$ by $T = 2\gamma^{-1}\ln(2)$ calls to \trymulticolor$(x)$. 
    By \cref{lem:try-multi-color}, after $T$ calls to \trymulticolor$(x)$, the probability a vertex $v$ remains uncolored is 
    \[ 
        \parens*{ \exp(-\gamma x / 2) + \frac{1}{\poly(n)} }^T \leq 
\exp(-T\gamma/2 \cdot x) + \frac{2T}{\poly(n)} = 2^{-x} + \frac{1}{\poly(n)} \ ,
    \]
    where the equality is by our choice of $T = O(1)$, and the inequality follows from $(a+b)^T \leq a^T + \sum_{i=1}^T (Tb)^i \leq a^T + T b \sum_{i=0}^\infty 2^{-i} \leq a^T + 2Tb$ for any $a\in [0,1]$ and $T,b \geq 0$ s.t.\ $Tb \leq 1/2$. Hence, we implement the algorithm of \cite{HKNT22} with only $T = O(\gamma^{-1})$ overhead in the round complexity and worsening the success probability by a $1/\poly(n)$ additive error.
\end{proof}

\subsection{Trying Random Colors}
\label{sec:appendix-try-random-color}

\begin{algorithm}
    \caption{\trycolor\label{alg:try-random-color}}
    \nonl\Input{A coloring $\col$, a set $S \subseteq V \setminus \dom\col$ and a set $\calC(v) \subseteq [\Delta+1]$ for each $v\in S$}

    \nonl\Output{A coloring $\col' \succeq \col$}

    Each $v\in S$ activates itself with probability $p = \gamma/4$.

    \If{$v$ is active}{
        $v$ samples $c(v) \in \calC(v)$ uniformly at random.
        
        If $c(v) \in L(v)$ and $c(v) \notin \set{c(u) : u\in N(v) \text{ and } \ID_u \leq \ID_v }$, then $v$ adopts $c(v)$ as its color. Otherwise $v$ remains uncolored.
        \label{line:try-random-color-adopt}
    }
\end{algorithm}

\begin{restatable}{lemma}{LemTryColor}
    \label{lem:try-color}
    Let $c \gg 1$ and $\gamma \in (0,1)$ be universal constants known to all nodes.
    Let $\col$ be a coloring, $S \subseteq V \setminus \dom\col$ a set of uncolored nodes, and $\calC(v)$ a set of colors for each $v\in S$ such that
    \begin{enumerate}
        \item $v$ can sample a uniform color in $\calC(v)$ in $O(1)$ rounds,
        \item $|\calC(v)| \geq \gamma^{-2}c\log n$ for some large constant $c$, and
        \item\label[part]{part:rct-clique-palette} 
             \label[part]{part:rct-min-palette}
        $|\calC(v) \cap L_\col(v)| \ge \gamma \max\set{ |\calC(v)|, \deg_\col(v) }$.
    \end{enumerate}
    Let $\psi$ be the partial coloring produced by \trycolor (\cref{alg:try-random-color}).
    Then, w.h.p., each $w \in V_H$ has uncolored degree in $S$
    \[
    \deg_{\psi}(w; S) \le 
    \max\set*{ 
        \left(1 - \frac{\gamma^2}{32}\right) \deg_{\col}(w; S) ,~
        \gamma^{-2} c\log n}
     \ .
    \]
    The algorithm runs in $O(1)$ rounds.
\end{restatable}

\newcommand{\evE}{\mathcal{E}}
\begin{proof}
    Let $\evE$ be the event that all vertices with $\gamma^{-2}(c/2)\log n$ neighbors in $S$ have between $(\gamma/8)\deg(w, S)$ and $(\gamma/2)\deg(w, S)$ active neighbors in $S$. Each vertex gets activated independently with probability $p = \gamma / 4$, so by the classic Chernoff bound and union bound, $\evE$ holds with high probability. We implicitly condition on $\evE$ henceforth.

    Consider some $w$ with $\deg_\col(w, S) \geq \gamma^{-2}c\log n$, for otherwise the claim already holds. Call $d$ its active degree in $S$ and $u_1, \ldots, u_d$ its active neighbors in $S$ ordered by increasing identifier. From $\evE$, we know that $d \geq (\gamma/8)\deg_\col(w, S)$, so let us prove that each $u_i$ gets colored with constant probability (over the randomness of the sampled colors). Define $X_i$ as the indicator random variable of the event that vertex $u_i$ gets colored. Note that to get colored, $u_i$ must  sample a color from $\calC(u_i) \cap L_\col(u_i)$ that is not  sampled by active neighbors of smaller ID. The former occurs with probability at least $\gamma$ by assumption (3). Conditioning on $u_i$ sampling an available color, i.e., on
    $c(u_i) \in \calC(v) \cap L_\col(v)$, we get that $c(u_i)$ is uniformly distributed in $\calC(u_i) \cap L_\col(u_i)$. 
    If $u_i$ has $\deg(u_i, S) \leq \gamma^{-2}(c/2)\log n$ neighbors in $S$ (so $\evE$ does not apply to $u_i$), it gets colored with probability at least $1/2$ since 
    \[
    |\calC(u_i) \cap L_\col(u_i)| 
    \geq \gamma |\calC(u_i)| 
    \geq \gamma^{-1}c\log n 
    \geq 2\deg(u_i, S) \ ,
    \] 
    where the first inequality uses assumption (3) and the second one assumption (2).
    Otherwise, if $\deg(u_i, S) \geq \gamma^{-2}(c/2)\log n$, vertex $u_i$ has at most $(\gamma/2)\deg(u_i, S)$ active neighbors, since $\evE$ applies to $u_i$, hence the probability that its color conflicts with an activated neighbor is at most 
    \[
    \frac{(\gamma/2) \deg(u_i)}{|\calC(u_i) \cap L_\col(u_i)|} \leq \gamma^2/2
    \]
    from assumption (3). Also note that this upper bound on $X_i = 1$ (with the conditioning on $\evE$) holds for any conditioning on the colors of sampled by neighbors of lower ID. For all $i$, we thus have that
    \[
    \Pr[ X_i = 1 ~|~ \evE, c(u_1), \ldots, c(u_{i-1}) ]
    \geq \gamma(1 - \gamma^2/2) 
    \geq \gamma/2 \ .
    \]
    As $X_i$ depends only on the colors sampled by lower ID neighbors, we may apply the Chernoff Bound with stochastic domination (\cref{lem:chernoff}). We get that at least $(\gamma/4)d \geq (\gamma^2/32) \deg_\col(w, S)$ neighbors of $w$ get colored with high probability, using the assumption on $w$'s degree in $S$. By union bound, it holds for all such $w$ in the graph. In particular, it means that the uncolored degree in $S$ of every such $w$ decreases by at least a factor $(1 - \gamma^2/32)$.
\end{proof}

One call to \trycolor reduces the uncolored degrees only by a small constant factor. Nonetheless, a direct corollary of \cref{lem:try-color} is that repeated calls reduce uncolored degrees by any desirable constant factor. We emphasize that \cref{lem:it-try-color} makes a stronger assumption on the slack of vertices than \cref{lem:try-color}.

\begin{corollary}
    \label{lem:it-try-color}
    Let $\gamma, \col, \calC, S$ such that all vertices $v\in S$ have $|\calC(v)| \gg \gamma^{-2}\log n$ and
    \[ 
    |\calC(v) \cap L_\col(v)| 
    \geq \deg_\col(v) + \gamma|\calC(v)| \ .
    \]
    For any $\delta \in (0,1)$, after $O \parens*{ \frac{\ln(1/\delta)}{\gamma^2} }$ iterations of \trycolor, the uncolored degree of every vertex in $S$ has decreased by a factor $\delta$ or is at most $O(\gamma^{-2}\log n)$.
\end{corollary}

\begin{proof}
    We argue that assumption for \cref{lem:try-color} hold before each call.
    Note that assumption (2) always holds as $\calC(v)$ never changes. Moreover, the slack in $\calC$ does not decrease as we extend the coloring, so assumption (3) always hold. Hence, w.h.p., by \cref{lem:try-color}, the uncolored degrees decrease by a factor $(1 - \gamma^2/32)$ each call, and after $T = \frac{32\ln(1/\delta)}{\gamma^2}$ calls it decreases by a factor $\delta$.
\end{proof}

\subsection{Slack Generation}
\label{sec:appendix-slack-gen}

In this section, we elaborate on \cref{prop:slack-generation}. We show that while our version of \slackgeneration slightly differs from that from prior work \cite{EPS15,HKMT21,HKNT22}, it can be recovered from them.

\PropSlackGeneration*

\begin{algorithm}
    \caption{\slackgeneration\label{alg:slack-generation}}
    Each $v \notin \Vcabal$ joins $\Vactive$ w.p.\ $\pg=1/200$
    
    Each $v\in \Vactive$ samples $c(v) \in [\Delta + 1]\setminus [300\eps \Delta] $ uniformly at random. 
    
    Let $\colsg(v) = c(v)$ if $v\in \Vactive$ and $c(v) \notin c(N_H(v))$. Otherwise, set $\colsg(v) = \bot$.
\end{algorithm}

\begin{proof}[Proof Sketch of \cref{prop:slack-generation}]
    The slack generation lemma from \cite{HKMT21} shows that if each node is active with some probability $\pg$, and tries a random color from $[\Delta+1]$ if active, then afterwards, a node $v$ of sparsity $\zeta_v$ has slack $\Omega(\pg^2 \zeta_v)$.
    With the sparsity threshold at $\Theta(\eps^2\Delta)$ for the classification of sparse and dense nodes, we have $\CSlack \in \Theta(\eps^2\pg^2)$.
    Compared to our setting, a few differences emerge:
    \begin{itemize}
        \item We want to argue that dense nodes specifically receive \emph{reuse} slack.

        \item We want our statement to hold with nodes trying colors from a  smaller range $[\Delta+1]\setminus [300\eps \Delta]$.

        \item We want to limit the number of nodes colored in each almost-clique.
    \end{itemize}

    The last point is very straightforward. Consider an almost-clique $K$, each of its nodes has an independent probability $\pg=1/200$ of even trying a color in \slackgeneration, a prerequisite to being colored. Therefore, no more than $\card{K}/200$ nodes of $K$ get colored in expectation during \slackgeneration. As $\card{K} \in \Theta(\Delta)$ and $\Delta \geq \Deltalow \gg \Theta(\log n)$ by assumption, by \cref{lem:basicchernoff} (Chernoff bound), $\card{K \cap \dom \colsg} \leq 2\cdot \card{K}/200 = \card{K}/100$ with high probability. As there are at most $O(n/\Delta)$ almost-cliques, the statement holds with high probability for all of them, by union bound.

    To see why restricting the sampling space to one of only $(1-300\eps)\Delta$ colors has a minimal impact on affect the amount of slack we get, suppose that we were performing slack generation using the full $[\Delta+1]$ color space. The standard argument from prior works shows that with probability $1-\exp(-\Omega(\zeta_v))$, a node with $\zeta_v\Delta$ pairs of unconnected neighbors (anti-edges in $N(v)$) has $\Omega(\zeta_v)$ colors from $[\Delta+1]$ adopted by exactly two of its neighbors. Now, note that the distribution over possible configurations that the process can reach is not affected by permuting colors from $[\Delta+1]$. As a result, when conditioning the distribution on the fact that some number $x$ of colors appear exactly twice in the neighborhood of $v$, the $x$ colors are still a random uniform $x$-sized subset of in $[\Delta+1]$.

    Consider now a process in which nodes become active with probability $\pg/(1-300\eps)$, active nodes try a random color in $[\Delta+1]$, and nodes of color $\leq 300\eps \Delta$ uncolor themselves. This process is equivalent to our \slackgeneration, in that it results in the same distribution of nodes getting colored. From prior work~\cite{HKMT21}, we know that the first two steps give $\Omega(\zeta_v)$ slack to a node $v$ with $\zeta_v\Delta$ pairs of unconnected neighbors with probability $1-\exp(-\Omega(\zeta_v))$, in the form of $\Omega(\zeta_v)$ colors being repeated exactly twice in $N(v)$. Since the $s_v = \Omega(\zeta_v)$ colors appearing exactly twice are uniformly distributed in $[\Delta+1]$, and $\eps = 1/2000$ (\cref{eq:params}), $300\eps s_v < s_v/5$ of them are expected to be in $[300\eps \Delta]$. A node receiving more than $C \log n$ slack for a sufficiently large constant $C$ before some nodes uncolor themselves thus loses at most $600 \eps s_v$ slack-providing colors in this last step, with high probability, meaning it gets to keep $\Omega(\zeta_v)$ reuse slack. 
Since dense nodes have $\Omega(e_v \Delta)$ anti-edges between their in-clique neighbors and external neighborhood~\cite[Lemma 3]{HNT21}, a dense node gets $\Omega(e_v)$ reuse slack granted its external degree $e_v$ is a sufficiently large compared to $\Omega(\log n)$.
\end{proof}

\subsection{Colorful Matching When $a_K = \Omega(\log n)$}
\label{sec:appendix-colorful-matching}
We detail here the implementation of the colorful matching in almost-cliques where $a_K \geq \Omega(\log n)$. We sketch the probabilistic argument, which mostly follows from \cite{FGHKN24}.

\LemMatchingHigh*

\begin{algorithm}
    \caption{\colorfulmatching\label{alg:colorful-matching}}

    \nonl Partition the color space into ranges $R_i = \range{(i-1)C\log n + 300 \eps\Delta + 1, i C\log n + 300\eps\Delta}$ for each $i\in[k]$ where $k = \frac{\Delta+1 - 300\eps\Delta}{C\log n}$.

    Let $M_K = \emptyset$ and $\colcm[,0] = \colsg$ be the initial (partial) coloring.
    
    \For{$t=1, 2, \ldots, O(1/\eps)$}{
        Each $v\in K$ samples $c(v) \in [\Delta + 1] \setminus [300\eps\Delta]$.

        Let $S_i = \set{u\in K, c(v) \in R_i}$ for $i\in \range*{ k }$. Compute disjoint BFS trees $T_i$ spanning each $S_i$.

        \textbf{if} $c(v) \notin L(v)$ or $c(v) \in c(N(v))$, \textbf{then} drop $c(v)$.

        Aggregate $C\log n$-bit maps on the BFS tree of $S_i$ to compute which colors give an anti-edge (even after colors were dropped).
        \label{line:colorful-matching-check-anti-edge}
    
        Extend $\colcm[,t-1]$ to $\colcm[,t] \succeq \colcm[,t-1]$ by setting $\colcm[,t](v) = c$ for each $v$ and $c$ such that $c(v) = c$ \emph{and} $c$ provides an anti-edge.
    }
\end{algorithm}

\begin{claim}
    \label{claim:detect-mono-anti-edges}
    In \cref{line:colorful-matching-check-anti-edge}, vertices of $S_i$ learn for which colors in $R_i$ there exists an anti-edge.
\end{claim}

\begin{proof}
    We ignore all vertices $v\in S_i$ such that $c(v)$ was dropped. On the BFS tree $T_i$ spanning $S_i$, each vertex aggregates a bitmap representing which colors of $R_i$ are retained (not dropped). They send it to their parent in the tree. After this first aggregation, each vertex in $T_i$ knows the set of colors in $R_i$ present in each of the subtrees rooted at each of its direct descendants in $T_i$. Each vertex $v$ in $S_i$ computes the set of colors in $R_i$ such that \emph{two descendants of $v$ in $T_i$ are the roots of a subtree containing a vertex with that color}. Each such color corresponds to a monochromatic anti-edge, and these sets of colors are efficiently encoded as a second bitmap. Vertices of $S_i$ learn about all such colors by aggregating a bit-wise OR of the bitmaps.
    
    Clearly, when only one vertex uses a color there never exists two distinct subtrees with that color. Suppose there exists a monochromatic anti-edge $c(u) = c(v) = c$ for some $c\in R_i$. By definition, both $u$ and $v$ are in $S_i$. The lowest common ancestor of $u$ and $v$ in $T_i$ detects that there exists an anti-edge with $c$-colored endpoints.
\end{proof}

We need with some standard definitions.
Let $f = \set{u,v}$ be an anti-edge.
For a set $D \subseteq [\Delta+1]$ of colors, let $\avail_D(f) = |L(u) \cap L(v) \cap D|$ be the number of colors in $D$ that both endpoints of $f$ could adopt. When $D = \set{c}$, we write $\avail_D(f) = \avail_c(f)$. For a set $F$ of anti-edges, let $\avail_D(F) = \sum_{f \in F} \avail_D(f)$.

\begin{lemma}
\label{lem:colorful-matching-large-a}
    After \cref{alg:colorful-matching}, w.h.p., there exists a colorful matching of size $a_K/(18\eps)$ in each almost-clique such that $a_K \in \Omega(\log n)$ and $\avail_{D_0}(F) \geq a_K \Delta^2 / 3$, where $D_0 = [\Delta+1] \setminus [300\eps\Delta]$ and $F$ is the set of anti-edges in $K$ with both endpoints uncolored before \cref{alg:colorful-matching}.
\end{lemma}

\begin{proof}[Proof Sketch of \cref{lem:colorful-matching-large-a}]
    Fix some $D \subseteq D_0$. Suppose $\avail_D(F) \geq a_K \Delta^2/6$. 
Let $X_i$ be the random variable indicating if both endpoints of $f_i = \set{u_i, v_i}$ -- the $i\supth$ anti-edge in $F$ -- try the same color and no vertex in $N(u_i) \cup N(v_i) \setminus \set{u_i,v_i}$ try that color. We have
    \[
        \Pr[X_i = 1] =
        \frac{\avail_D(f_i)}{(1-300\eps)^2\Delta^2} \parens*{ 1 - \frac{1}{(1 - 300\eps)\Delta} }^{(1+2\eps)\Delta} \geq
        \frac{\avail_D(f_i)}{100\Delta^2} \ .
    \]
    from our choice of $\eps$, \cref{eq:params}.
    Hence, the expected number of monochromatic anti-edges of $F$ to join the colorful matching is 
    \[
        \Exp \range*{ \sum_i X_i } = \frac{\avail_D(F)}{100\Delta^2} \geq \frac{a_K}{600} \ .
    \]
    Concentration is proven by Talagrand's inequality, see e.g., \cite{AA20,HKMT21,HKNT22,FGHKN24}.

    The algorithm begins with $D = D_0$ and $F$ the set of all anti-edges induced in $K$. Each time it colors vertices, we remove its color from $D$ and incident anti-edges in $F$.
    If, after some number of iterations, we have $\avail_D(F) < a_K |K|^2/6$, then we claim the matching has size $a_K/(18\eps)$. Each time we insert an anti-edge in the matching it decreases $\avail_D(F)$ by at most $3\eps\Delta^2$. Observe that the only reason we remove an anti-edge (resp.\ color) from $F$ (resp.\ $D$) is that we inserted some anti-edge in the matching. Since we initially have $\avail_D(F) \geq a_K|K|^2 / 3$, we must have inserted at least $\frac{a_K|K|^2(1/3 - 1/6)}{3\eps|K|^2} \geq \frac{a_K}{18\eps}$ anti-edges in the matching by the time we reach this state. If this never occurs, we insert $\Omega(a_K)$ each iteration of the process and hence are done after $O(1/\eps)$ times.
\end{proof}

\begin{proof}[Proof Sketch of \cref{lem:colorful-matching-high}]
    By \cite[Lemma A.2, arXiv-version]{FGHKN23}, even after slack generation, w.p.\ $1-\exp(-\Omega(a_K)) \geq 1-1/\poly(n)$, if $F$ is the set of anti-edges induced in $K$ with both endpoints uncolored we have $\avail_{D_0}(F) \geq a_K\Delta^2/3$. By \cref{lem:colorful-matching-large-a}, the colorful matching is large enough. By \cref{claim:detect-mono-anti-edges}, \cref{alg:colorful-matching} runs in $O(1/\eps)$ rounds with high probability. Clearly, it does not use reserved colors. Observe that because of \cref{line:colorful-matching-check-anti-edge}, a vertex gets colored in \cref{alg:colorful-matching} only if some other vertex in its almost-clique adopt the same color.
\end{proof}

\subsection{Query Access to Palettes}
\label{sec:proof-query}

\LemQuery*

\begin{proof}
    Let $k \eqdef \ceil*{\frac{\Delta+1}{C\log n}}$ where $C$ is some large universal constant. 
    Partition $[\Delta+1]$ into $k$ contiguous ranges $R_i = \set{(i-1)\cdot C\log n+1, \ldots, i\cdot C\log n}$ of colors for $i\in[k]$. 
    Focus on some almost-clique $K$. Split $K$ into $k$ random groups $X_1, X_2, \ldots, X_k$.
    The $i\supth$ group computes the set $R_i \cap \col(K)$ by aggregating a bit-wise OR.
    It takes $O(1)$ rounds because each $X_i$ has diameter two and 
    $N_H(X_i) = K$ (\cref{fact:random-groups}). 
    By taking the complement, 
    vertices of $X_i$ also compute $R_i \setminus \col(K) = R_i \cap L_\col(K)$, i.e., 
    the free colors in range $R_i$.

    Fix $\calC \in \set{\col(K), L_\col(K)}$. The algorithm works the same for both. To comply with all vertices in $K$, we run the algorithm once for each value.

    First, vertices of $X_i$ learn $S_i \eqdef \sum_{j < i} |R_j \cap \calC|$ as follow.
    Choose an arbitrary vertex $w \in K$ and run a BFS of depth one in $K$. 
    It returns a $O(1)$-depth tree $T \subseteq E_G$ spanning all clusters $\bigcup_{v\in N_H(w)} V(v)$.
    Using the prefix sum algorithm on $T$ (\cref{lem:prefix-sum}), we order leaders of $\bigcup_{v\in N_H(w)} V(v)$ as $v_1, \ldots, v_{k}$. Each leader knows its index in the ordering. Since $v_i$ has a neighbor in $X_i$, it learns $|R_i \cap \calC|$ in $O(1)$ rounds.
    Using the prefix sum algorithm on $T$ again, each $v_i$ learns $S_i$. It then broadcasts $(i, S_i)$ to its neighbors, one of which belongs to $X_i$. Hence, after $O(1)$, all vertices in each $X_i$ know $S_i$.
    
    Nodes of $X_i$ broadcast the $O(\log n)$ bit message
    $(i, S_i, R_i\cap \calC)$. 
    The last part of the message is a \emph{set} of $|R_i| \leq O(\log n)$ 
    colors represented as a $O(\log n)$-bitmap.
    
    To learn $|\calC \cap [a_v, b_v]|$, vertex $v$ selects exactly one
    machine in $w_v^a \in V(v)$ incident to $X_i$ and exactly one in $w_v^b \in V(v)$ incident to $X_j$ 
    where $a_v \in R_i$ and $b_v\in R_t$.
    Because machine $w_v^a$ knows $a_v \in R_i$ and the message shared by $X_i$, 
    it computes the number of colors strictly smaller than $a_v$ in $\calC$ as 
    $S_v^a = S_i + |\set{c \in R_i \cap \calC: c < a_v}|$.
    Similarly, machine $w_v^b$ computes the number of colors smaller or equal to $b_v$ as $S_v^b = S_t + |\set{c \in R_t \cap \calC: c < b_v}|$.
    As those are two $O(\log n)$-bit integers, they can be disseminated in $V(v)$
    in $O(1)$ rounds.
    Finally, all machines in $V(v)$ know $|\calC \cap [a_v,b_v]| = S_v^b - S_v^a$.
    
    To learn the $i_v\supth$ color, first we broadcast $S_v^a$ to all machines in $V(v)$. The $i_v\supth$ color of $[a_v, b_v]$ is the $(S_v^a + i_v)\supth$ color of $\calC$.
    Then, a machine incident to some group $X_j$ can locally compute $|\calC \cap [a_v, b_v] \cap R_{\leq j}|$ from $S_a^v$ and the message received from $X_j$. In particular, it knows if $i_v \in R_j$ and, if this is the case, can send broadcast the corresponding color in $T(v)$.
\end{proof}

\begin{remark}
    \label{remark:query-compress}
    If all machines of $V(v)$ have a scheme to compress colors, we can use \cref{lem:query} to learn multiple colors at a time. Indeed, the only moment where the $O(\log n)$-bit description of the colors matters is when they are broadcasted within the cluster at the very end. In general, if nodes have a scheme to encode colors using $b$ bits, they can query up to $O(\log n / b)$ colors in the clique palette in $O(1)$ rounds. We emphasize the vertex only learns the encoded colors (e.g., the hashes).
\end{remark}

\subsection{Computing Put-Aside Sets}
\label{sec:proof-compute-put-aside}

\begin{algorithm}
    \caption{\alg{ComputePutAsideSets}\label{alg:compute-put-aside}}
    \nonl\Input{The coloring $\col$ computed by \colorfulmatchingcabal}

    \nonl\Output{Set $P_K \subseteq I_K \setminus \dom\col$ in each $K\in\Kcabal$}
    
    Each $v\in I_K \setminus \dom\col$ joins $\Pcand_K$ independently w.p.\ $p=\Theta(\lmin^2 / \Delta)$.
    
    Let $\Psafe_K = \set*{v\in \Pcand_K: N_H(v) \cap \bigcup_{K' \neq K} \Pcand_K = \emptyset}$. 
    
    Each $v\in \Psafe_K$ joins $P_K$ independently w.p.\ $\Theta(1/\lmin)$. \label{line:put-aside-subsample}

    Arbitrarily drop nodes from $P_K$ to reduce its size to $r$.
\end{algorithm}

\LemConstructPutAside*

\begin{proof}[Proof Sketch]
    The first two steps of \cref{alg:compute-put-aside} implement \cite[Algorithm 6]{HKNT22}. Since $\Delta \gg \lmin^3$, \cite[Lemma 5]{HKNT22} applies and, w.h.p, sets $\Psafe_K$ have size $\Omega(\lmin^2)$. As each vertex then joins $P_K$ w.p.\ $\Theta(1/\lmin)$, with high probability, \cref{part:put-aside-size,part:put-aside-put-aside} hold.
    For the remaining of this proof, we focus on proving \cref{part:put-aside-inliers}.
    Consider a cabal $K$, call its inliers $v_1, \ldots, v_{|K|}$.
    Let $X_i$ be the random variable indicating if $v_i$ has an external neighbor in $P_K$. 
    Observe that each vertex joins an independent set with probability at most $q= \Theta(1/\lmin)$ (\cref{line:put-aside-subsample}). As inliers have at most $O(\lmin)$ external neighbors, we can choose $q$ small enough that $\Pr[X_i = 1] \le 1/200$ by Markov inequality. Hence, in expectation, at most $|K|/200$ inliers of $K$ have an external neighbor in a put-aside set. To show concentration on $X = \sum_i X_i$, we use the read-$k$ bound with $k= \Theta(\lmin)$ (\cref{lem:k-read}). Indeed, variables depend on the random binary variables of external neighbors, each of which is an inlier with $\le k$ external neighbors. Thus, w.p.\ $1-\exp\parens*{ -\Theta(\frac{|K|}{k}) } \ge 1-1/\poly(n)$, we get that at most $|K|/100$ inliers have an external neighbor in a put-aside set.
\end{proof}

\subsection{Synchronized Color Trial}
\label{sec:sct-proof}

\LemSCT*
\begin{proof}
    Define the random variable $X_i$ indicating if $i\supth$ vertex failed to retain its color.
    By assumption, $|L_\col(K)| - r_K \ge |S_K|$, thus all nodes receive some color to try. Since they all try different colors, the only reason they could fail to retain a color is because it conflicts with an external neighbor.

    Split $S$ into sets $A = \set{1, 2, \ldots, \floor{|S|/2}}$ and $B = S \setminus A$ of size at least $\floor{|S|/2}$ each. We first show that the number of nodes to fail in $A \cap S$ over the randomness of $\pi(A)$ is small. 
    Fix $i\in A$. After revealing $\pi(1), \ldots, \pi(i-1)$, the value $\pi(i)$ is uniform in a set of $|S|-|A| \geq |S|/3 \geq (\alpha/3) |K|$ values.
    By union bound, the probability the $i\supth$ vertex fails to retain its color --- even under adversarial conditioning of $\pi(1), \ldots, \pi(i-1)$ --- is
    \[
    \Pr [X_i = 1~|~\pi(1), \pi(2), \ldots, \pi(i-1)] \le \frac{e_v}{|S| - |A|} \leq \frac{3/\alpha \cdot  e_v}{|K|} \ .
    \]
    By linearity of the expectation, 
    $$\Exp \range* { \sum_{i\in A} X_i } \le 3/\alpha \sum_{v\in K} \frac{e_v}{|K|} \leq 3/\alpha\cdot e_K \ . $$
    By the martingale inequality (\cref{lem:chernoff}), with high probability, $|(A \cap S) \setminus \dom\colsct| < 6/\alpha \cdot \max\set{e_K, \lmin}$. The same bound holds for vertices in $B \cap S_t$. By union bound, w.h.p., both bounds hold simultaneously, and hence$|S \setminus \dom\colsct| \leq 12/\alpha \cdot \max\set{e_K, \lmin}$. 
\end{proof}

\begin{lemma}
    \label{lem:rep-permutation}
    The algorithm described in \cref{lem:sct} can be implemented in parallel in all almost-cliques in $O(1)$ rounds with high probability. The number of uncolored vertices in each $S_K$ is as described in \cref{lem:sct}.
\end{lemma}

\begin{proof}
    For each integer $k \in (1 \pm \eps)\Delta$, $r \in (0, 300\eps\Delta)$ and $s \in [\alpha k, \Delta+1]$, we construct show there exist a set of $\poly(n)$ permutations $\calP_{k, r, s}$ such that if we sample almost-clique $K$ with $|K|=k$, $r_K = r$, $|S_K| = s$ samples a uniform $\pi_K \in \calP_{k,r,s}$ and runs the algorithm of \cref{lem:sct}, with high probability, the number of uncolored vertices in $S_K$ is as described by \cref{lem:sct}.

    A local configuration $\calL$ for an almost-clique where conditions of \cref{lem:sct} are verified comprises the following informations.
    The $O(\log n)$-bit cluster identifiers, the order of the vertices in $K$, $\poly(n)$ local random bits for each vertex, their incident edges (at least $(1-\eps)\Delta$ inside and at most $\Delta$ in total), any adversarial (possibly partial) coloring outside of $K$, an arbitrary choice of $S_K \subseteq K$ such that $|S_K| = s$, and any partial coloring of $K \setminus S_K$ such that $|L(K)| \geq s+r$. It is easy to see that the number of such configurations is $\leq \exp(\poly(n))$.

    Construct $\calP_{k,r,s}$ by sampling independently $t$ truly uniform permutations of $\set{1, 2, \ldots, s}$. We say permutation $\pi$ is bad for a configuration $\calL$ if the synchronized color trial using $\pi$ fails on $\calL$, i.e., too many nodes are uncolored in $S_K$. By \cref{lem:sct}, a uniform permutation permutation is bad for $\calL$ with probability at most $n^{-c}$ for some constant $c \ge 1$. Hence, by Chernoff, w.p.\ $1-\exp(-\Omega(t/n^c))$, the set $\calP_{k,r,s}$ contains at most $2t/n^c$ bad permutations for a fixed local configuration $\calL$. 
    By the probabilistic method, for $t=\poly(n)$ a large enough polynomial, there exists a set of $t$ permutations $\calP_{k,r,s}$ such that for \emph{any} given local $(k,r,s)$-configuration, there are less than $2t/n^c$ permutations for which the synchronized color trial fails. 
    
    To implement the algorithm, we count $|K|$, $|S_K|$ using a BFS. Then, with the prefix sum algorithm, in $O(1)$ rounds, we order all vertices in $S_K$ in an arbitrary order where each vertex knows its index.
    Since the \congest model assumes unbounded local computation, all vertices of $K$ can compute $\calP_{|K|, r_K, |S_K|}$ locally. One vertex in $K$ samples a uniform $\pi_K \in \calP_{|K|, r_K, |S_K|}$ and broadcast it to all vertices in $K$ using a $O(\log t) = O(\log n)$ bit message.
    Vertices run the synchronized color trial using the permutation $\pi_K$ (the $i\supth$ vertex tries the $\pi_K(i)\supth$ color in $L(K)$).

    We argue that, w.h.p., the synchronized color trial worked in all almost-cliques.
    Fox on some $K$ and the result follows by union bound.
    Condition arbitrarily on the randomness of external neighbors (in particular the color they try). This defines a local configuration for $K$. By construction of $\calP_{|K|,r_K,|S_K|}$, the algorithm fails with probability at most $\frac{2t/n^c}{|\calP_{|K|, r_K, |S_K|}|} = 2n^{-c}$.
\end{proof}

\subsection{Collision-Free Hash Function}
\label{sec:proof-collision-free-hash}

\begin{lemma}
    \label{lem:hash-clique-palette-col-free}
    Let $k$ be some integer such that $\Delta \gg k^2\log^2 n$.
    Let $\col$ be any arbitrary coloring and $K$ be an almost-clique. Call $D \subseteq L_\col(K)$ the set containing the $k$ smallest colors in $L_\col(K)$.
    There exist a $O(1)$-round algorithm in $K$ at the end of which, w.h.p., all $v\in K$ know a hash function $h : [\Delta+1] \to [2k^2]$ that can be described in $O(\log\log n + \log k)$ bits and has \emph{no collision on $D$}, i.e., $|h(D)| = |D|$.
\end{lemma}

\begin{proof}
    Let $N=[\Delta+1]$, $M=4k^2$ and $\epsilon=1$. By \cref{thm:pwi-indep}, there exists some $\epsilon$-almost pairwise independent family $\calH$ such that for any pair $c_1\neq c_2 \in L_\col(K)$, a random $h\in \calH$ has a collision $h(c_1) = h(c_2)$ with probability at most $\frac{2}{M}$. By union bound over all such pairs, 
    \[ \Pr_{h\in \calH}[\exists c_1\neq c_2\in [N], h(c_1) = h(c_2)] \leq \frac{2k^2}{M} < 1/2 \ .\]
    
    Partition the almost-clique into $k^2 \cdot \Theta(\log n)$ random groups. Each group correspond to a triplet $(c_1, c_2, t) \in [k]\times [k]\times [\Theta(\log n)]$.
    For each $t\in [\Theta(\log n)]$, one vertex samples a uniform $h_t \in \calH$ and broadcast it to all group $(c_1, c_2, t)$. 
    Vertices of group $(c_1, c_2, t)$ query for the $c_1\supth$ and $c_2\supth$ color in $L_\col(K)$ the clique palette (\cref{lem:query}). They check if $h_t(c_1) = h_t(c_2)$ has a collision. 
    
    By aggregating a bit-wise OR in all groups $\set{(c_1, c_2, t), c_1, c_2\in [k]\times[k]}$, those groups learn if $h_t$ is collision-free on $D$. Since each $h_t$ is collision-free with probability $1/2$, w.h.p., at least one must be. We find an arbitrary one by simple converge cast in $K$. It can then be broadcasted to all vertices of $K$.
\end{proof}
 
\newpage

\section{Index of Notations}
\label{app:notation}

The following concepts are used in multiple sections.

\begin{figure}[H]
    \begin{tabular}{|l|l|l|}
        \hline
        Notation & Meaning    & Reference \\\hline\hline
        $\col, \psi$
            & colorings  & \cref{sec:notation} \\\hline
        $L_\col(v) = [\Delta+1] \setminus \col(N(v))$
            & palette of $v$ w.r.t\ $\col$ & \cref{sec:notation} \\\hline
        $\dom \col$
            & set of colored nodes & \cref{sec:notation} \\\hline
        $\deg_\col(v)$
            & uncolored degree & \cref{sec:notation} \\\hline
        $s_\col(v) = |L_\col(v)| - \deg_\col(v)$
           & slack & \cref{sec:notation} \\\hline
        $\calC(v) \subseteq [\Delta+1]$
            & a color space, e.g., $[\Delta+1]$, $L_\col(K)$, $[r_K]$
            & \\\hline\hline

        $G = (V_G, E_G)$
            & $n$-vertex communication network & \cref{sec:model} \\\hline
        $H = (V_H, E_H)$
            & cluster graph on $G$ & \cref{sec:model} \\\hline
        $\Delta$
            & maximum degree of $H$ & \cref{sec:model} \\\hline
        $V(v) \subseteq V_G$
            & clusters for each $v\in  V_H$ & \cref{sec:model} \\\hline
        $T(v)$
            & support tree spanning $V(v)$ & \cref{sec:model} \\\hline
        $\dilation$
            & the dilation & \cref{sec:model} \\\hline\hline

        $\CSlack$
            & slack generation constant & \cref{prop:slack-generation} \\\hline
        $\CCSlack$
            & clique-slack generation constant & \cref{lem:reuse-slack} \\\hline
        $\eps$
            & parameter for the ACD & \cref{eq:params} \\\hline
        $\delta$
            & precision of approximation & \cref{eq:params} \\\hline
        $\Deltalow$
            & min.\ degree for ultrafast coloring & \cref{eq:params} \\\hline
        $\lmin = \Theta(\log^{1.1} n)$
            & min.\ slack for \multitrial & \cref{eq:params} \\\hline\hline

        $\Vsparse, \Vdense, \Vcabal, \Kcabal$
            & almost-clique decomposition & \cref{sec:coloring-alg} \\\hline
        $\zeta_v$
            & local sparsity & \cref{def:sparsity} \\\hline
        $L_\col(K) = [\Delta+1] \setminus \col(K)$
        & clique palette & \cref{sec:coloring-alg} \\\hline
        $e_v, e_K, a_v, a_K$
            & external / anti-degrees & \cref{sec:coloring-alg} \\\hline
        $\tilde{e}_v, \tilde{e}_K$
            & $(1\pm\delta)$ mult.\ apx.\ of $e_v$ and $e_K$ & \cref{sec:coloring-alg} \\\hline
        $x_v$
            & approximation for $a_v$ in non-cabals & \cref{sec:coloring-alg} \\\hline
        $M_K$
            & size of colorful matching & \cref{sec:coloring-alg} \\\hline
        $I_K, O_K = K \setminus I_K$
            & inliers and outliers & \cref{eq:def-inliers} \\\hline
        $P_K$
            & put-aside set & \cref{sec:coloring-alg} \\\hline
        $r, r_v, r_K$
            & reserved colors & \cref{eq:reserved} \\\hline\hline

        $u_{K,i}$
            & $i\supth$ put-aside vertex in $K$ & \cref{sec:put-aside-sets} \\\hline
        $b$
            & length of a block& \cref{eq:put-aside-params} \\\hline
        $\ls$
            & number of safe donations & \cref{eq:put-aside-params} \\\hline
        $\crecol_{K, i}$
            & replacement color for donor of $u_{K,i}$ & \cref{sec:internally-safe-swaps} \\\hline
        $\cdon_{K, i}$
            & color donated to $u_{K,i}$ & \cref{sec:put-aside-sets} \\\hline
    \end{tabular}
\end{figure}

\newpage
\bibliographystyle{alpha}
\bibliography{arxiv-v2-references-part-1}

\end{document}